\documentclass[11pt]{amsart}

\usepackage{setspace}

\usepackage{amsfonts}
\usepackage{amsmath}
\usepackage{color}
\usepackage{hyperref}

\usepackage{pict2e}
\usepackage{graphicx}

\newtheorem{thm}{Theorem}[section]

\newtheorem{cor}[thm]{Corollary}
\newtheorem{lem}[thm]{Lemma}
\newtheorem{prop}[thm]{Proposition}
\newtheorem*{prob*}{Problem}
\newtheorem*{thm*}{Theorem}

\theoremstyle{definition}

\newtheorem*{defn*}{Definition}
\newtheorem{rem}[thm]{Remark}

\newtheorem*{rem*}{Remark}
\numberwithin{equation}{section}

\DeclareMathOperator{\const}{const.}

\DeclareMathOperator{\diag}{diag}

\newcommand{\Tr}{\mathop{\mathrm{Tr}}}

\newcommand{\eins}{\leavevmode\hbox{\small1\kern-3.8pt\normalsize1}}

\def\ra{\rightarrow}

\begin{document}
\title[Finite rank perturbations in products of two coupled random matrices]
{\bf{Finite rank perturbations in products of coupled random matrices: From one correlated to two Wishart ensembles}}

\author{Gernot Akemann}
\address{Faculty of Physics and Faculty of Mathematics, Bielefeld University,  P.O. Box 100131, D-33501 Bielefeld, Germany} \email{akemann@physik.uni-bielefeld.de}

\author{Tomasz Checinski}
\address{Faculty of Physics, Bielefeld University,  P.O. Box 100131, D-33501 Bielefeld, Germany} \email{tchecinski@uni-bielefeld.de}

\author{Dang-Zheng Liu}
\address{Key Laboratory of Wu Wen-Tsun Mathematics, CAS, School of Mathematical Sciences, University of Science and Technology of China, Hefei 230026, P.R.~China
 \&   Institute of Science and Technology Austria,  Klosterneuburg 3400, Austria} 
\email{dzliu@ustc.edu.cn}

\author{Eugene Strahov}
\address{Department of Mathematics, The Hebrew University of
Jerusalem, Givat Ram, Jerusalem
91904, Israel}\email{strahov@math.huji.ac.il}


\commby{}
\begin{abstract}
We compare finite rank perturbations of the following three ensembles of complex rectangular random matrices: First, a generalised Wishart ensemble with one random and two fixed correlation matrices introduced by Borodin and P\'ech\'e, second, the product of two independent random matrices where one has correlated entries, and third, the case when the two random matrices become also coupled through a fixed matrix. The singular value statistics of all three ensembles is shown to be determinantal and we derive double contour integral representations for their respective kernels. Three different kernels are found in the limit of infinite matrix dimension at the origin of the spectrum. They  depend on finite rank perturbations of the correlation and coupling matrices and are shown to be integrable. The first kernel (I) is found for two independent matrices from the second, and two weakly coupled matrices from the third ensemble. It generalises the Meijer $G$-kernel for two independent and uncorrelated matrices. The third kernel (III) is obtained for the generalised Wishart ensemble and for two strongly coupled matrices. It further generalises the perturbed Bessel kernel of Desrosiers and Forrester. Finally, kernel (II), found for the ensemble of two coupled matrices, provides an interpolation between the kernels (I) and (III), generalising previous findings of part of the authors.
\end{abstract}
\date{}

\maketitle

\tableofcontents
\section{Introduction and Main Results}
The topic of products of random matrices has seen a very rapid development in the past few years, in particular for a finite number of factors at finite matrix size. For example, it has been understood for a fixed product of $r$ independent complex Gaussian random matrices that its singular value statistics is determinantal, and the same holds for the complex eigenvalues, cf. \cite{AIp} for a recent review and a list of references. This has opened up the possibility for a detailed analysis of the local statistics. The global density of singular values of such a product, generalising the Marchenko-Pastur distribution with $r=1$, develops a singularity at the origin depending on $r$ (see e.g.\cite{FL14}). It is thus not surprising that for the local statistics at the origin, representing a hard edge, a new family of Meijer $G$-kernels labelled by $r$ has been found \cite{KZ14}. It generalises the Bessel kernel at $r=1$. On the other hand, in the bulk and at the soft edge of the spectrum the respective Sine and Airy kernel have been recovered \cite{LWZ}, agreeing with $r=1$. The kernels at $r=1$ are well known to be universal, see \cite{ArnoOUP} and references therein.

What is known about the universality of the Meijer $G$-kernel? While it is known to appear in ensembles with Cauchy interaction \cite{Bertola3} for $r\leq3$ (and conjectured $\forall r$), or when multiplying other types of e.g. truncated unitary matrices \cite{KKS}, the question is open for products of matrices from unitary bi-invariant ensembles with  general distributions. The difficulty is that the unitary group integrals  needed after singular value decomposition are not available in general  (see however \cite{DZL}).

Another direction to address its universality - apart from Wigner matrices that drop invariance entirely - is to introduce correlations among the elements of each random matrix. This is the route we will choose here, for the product of $r=2$ matrices as a starting point. We will combine this with a coupling among the two random matrices of scalar \cite{AStr16a,AStr16b} or matrix-valued type \cite{DZL}, investigating the most general distribution that is quadratic in the two random matrices and remains determinantal. A further direction was taken in \cite{DZP} by adding an external field to the product of $r$ Gaussian matrices.
All these deformations allow to study finite rank perturbations of the known Bessel and Meijer $G$-kernel, extending the results of \cite{DF06} for a single Wishart matrix with external field for the former, and of \cite{DZP,DZL} for the latter.
Similar findings were made earlier for deformations of the Airy kernel at the soft edge \cite{BBP,DF06,BP}, where a relation to directed percolation was pointed out. We can only speculate if the deformed Bessel, and Meijer $G$-kernel which we will find here, enjoy such a relation.

It is an open question if the analysis of \cite{BBP}, deforming the Tracy-Widom distribution by finite rank perturbations at the soft edge, could be extended to the smallest singular value distribution $p(s)$ in our setup. Apart from the Wishart-Laguerre ensemble, where several equivalent representations for $p(s)$ are available for the undeformed case, see e.g. \cite[Table 3]{EGP} for a list including references, for the product of $r\geq2$ independent matrices the corresponding Painlev\'e type systems of equations \cite{EStr} become very rapidly cumbersome, cf. \cite{WF}. This is the reason why we will focus on the kernel instead.

Let us introduce our most general ensemble of two correlated coupled matrices first. We consider two rectangular complex random matrices, $G$ of size $L\times M$, and $X$ of size $M\times N$. Throughout this work we will keep the following differences fixed:
\begin{equation}
\kappa=L-N\geq0\ , \ \ \nu=M-N\geq0\,.
\label{nukappadef}
\end{equation}
These two random matrices have the probability distribution with density
\begin{equation}
\mathcal{P}(G,X)=\,c
\exp\left[-\Tr\left(WGG^*\right)
+\Tr\left(\Omega GX+X^*G^*\Omega^*\right)-\Tr\left(QXX^*\right)\right],
\label{JointDensityGX}
\end{equation}
where $c$ is a normalisation constant. Here, we have introduced three constant matrices. First, $\Omega$ is a fixed complex matrix of size $N\times L$ with squared singular values $\delta_1,\ldots,\delta_N\geq 0$. It parametrises the coupling between the matrices $G$ and $X$. Second, $Q$ is a fixed Hermitian $M\times M$ matrix with positive eigenvalues $q_1,\ldots, q_M>0$. It introduces correlations among the matrix elements of $X$ - typically $Q^{-1}$ denotes a given empirical covariance matrix. Third, $W$ is a fixed Hermitian matrix of size $L\times L$, that we will take to be proportional to the identity, $W=\alpha \eins_L$, with $\alpha>0$ constant. We will show later, why the choice of a more general fixed matrix $W$ with positive eigenvalues, introducing correlations also among the matrix elements of $G$, leads out of the class of determinantal point processes. The joint distribution (\ref{JointDensityGX}) of the matrices $G$ and $X$ is convergent if the following conditions are satisfied:
\begin{equation}
\alpha q_i-\delta_j>0\ ,\ \ \forall\ i=1,\ldots,M,\ \ \forall\  j=1,\ldots, N,\label{aqd}
\end{equation}
see the discussion around \eqref{Idef} below. In the sequel we will determine the correlation functions of the squared singular values of the product matrix $Y=GX$ of the two random matrices that are coupled and correlated. Before stating our main results for this ensemble let us state, in which special cases the distribution \eqref{JointDensityGX} has been investigated, and compare to two related ensembles.

In \cite{AKW,AIK} it was shown that the distribution of squared singular values of the product $Y$ of two independent rectangular matrices with $\Omega\equiv0$ and $W,Q\sim\eins$ forms a determinantal point process. There, its kernel was constructed in terms of biorthogonal functions for finite matrix size. In \cite{KZ14} a double contour representation was found,  leading to a limiting kernel expressed in terms of Meijer $G$-functions, whence its name. Both, the results from \cite{AIK} and \cite{KZ14}, were derived for arbitrary but fixed products of $r\geq1$ independent matrices. In \cite{AStr16a} their independence was dropped and the product $Y$ of $r=2$ coupled matrices was considered for $Y$ square, i.e. $L=N$ with $\nu\ge 0$, where
\begin{equation}
\Omega=\frac{1-\mu}{2\mu}\eins_N, \ \ W= \frac{1+\mu}{2\mu}\eins_N \ \ \mbox{and}\ \  Q= \frac{1+\mu}{2\mu}\eins_M\,,
\label{def-ad}
\end{equation}
depending on the parameter $\mu\in(0,1]$. It allowed to interpolate between the ensemble of the product of two independent matrices
($\mu=1$), and a single random matrix ($\mu\to0$), due to $G=X^*$ in this limit. Once again the interpolating ensemble was shown to be determinantal and was solved in terms of biorthogonal functions \cite{AStr16a}. In \cite{AStr16b} three different scaling limits were identified at the origin of the spectrum representing a hard edge, with the following limiting kernels: (I) the Meijer $G$-kernel for two independent matrices, proving its universality for a one-parameter family, (II) a parameter dependent kernel that interpolates between the limit (I) and the limit (III), where the well-known universal Bessel kernel (III) was obtained. In a subsequent paper \cite{DZL} a full coupling matrix $\Omega$ was introduced as in \eqref{JointDensityGX}, while keeping the conditions on $W=\frac{1+\mu}{2\mu}\eins_N$ and $Q=\frac{1+\mu}{2\mu}\eins_M$ as in \cite{AStr16a}. There, the kernels in the limits (I), (II) and (III) were extended and finite rank perturbations in the limits (II) and (III) of \cite{AStr16b} were found. In the present work we will study the most general case of a coupling matrix and correlations amongst the matrix elements of the two random matrices, that is compatible with a determinantal structure. Our findings generalise the kernels found in \cite{AStr16b,DZL}, to include also finite rank perturbations in limit (I).

Next, we introduce the following two ensembles related to \eqref{JointDensityGX}. First, consider the product of two independent matrices $Y=GX$, where the second random matrix $X$ has correlated entries
\begin{equation}
\mathcal{P}_2(G,X)=\,c_2
\exp\left[-\alpha\Tr\left(GG^*\right)-\Tr\left(QXX^*\right)\right]\ .
\label{JointDensityindepGX}
\end{equation}
All conditions on the dimensions of $G,X,$ and $Q$ are as in \eqref{JointDensityGX}, $\alpha>0$ and $c_2$ is a normalisation constant. Related ensembles of products of $r$ correlated random matrices have been considered by Forrester  in the limit of infinite product size $r\to\infty$ \cite{PFL1}, studying the Lyapunov exponents.

Second, following \cite{BP} we consider the generalised Wishart ensemble correlated from two sides:
\begin{equation}
\mathcal{P}_1(X)=\,c_1
\exp\left[-\Tr\left(X\Sigma X^*\right)-\Tr\left(QXX^*\right)\right]\ ,
\label{corrWishart}
\end{equation}
with $Q$ as before, $\Sigma$ a fixed Hermitian $N\times N$ matrix with positive eigenvalues, and normalisation $c_1$. For both $Q$ and $\Sigma$ having positive eigenvalues this ensemble is clearly convergent. However, we could also allow $\Sigma$ (or $Q$) to have several or all eigenvalues to be negative, as long as they are bounded in absolute value by the smallest eigenvalue of $Q$ (or $\Sigma$). In fact, when integrating out the random matrix $G$ in ensemble \eqref{JointDensityGX}, we arrive at $-\Sigma = \Omega\Omega^*/\alpha$, with the corresponding bounds \eqref{aqd}. The ensemble \eqref{corrWishart} has been introduced in \cite{BP} for $M=N$ and was called generalised Wishart ensemble. The authors solved it for finite matrix size $N$ using Schur processes, and then focused on the kernel at the soft edge, without considering the hard edge.

Let $\Delta_{n}(x_{1},...,x_{n})=\prod_{1\leq j<k\leq n}(x_k-x_j)=\det\left[x_j^{i-1}\right]_{i,j=1}^n$ denote the Vandermonde determinant. We are now ready to state our main results.
\begin{thm}\label{TheoremMainDensity}
Denote by $x_1,\ldots,x_N$ the squared singular values of $X$ and by $y_1,\ldots,y_N$ the squared singular values of $Y=GX$, where $X$ and $G$ are distributed according to \eqref{JointDensityGX}, all parameters $\delta_1,\ldots,\delta_N$ and $q_1,\ldots,q_M$ are mutually distinct and satisfy \eqref{aqd}. Their joint probability density function is given by
\begin{equation}\label{MainDensity}
\begin{split}
P(x_1,\ldots,x_N;y_1,\ldots,y_N)=&\frac{1}{Z}\det\left[y_{j}^{\frac{\kappa}{2}}I_{\kappa}\left(2\sqrt{\delta_{l}y_{j}}\right)\right]_{j,l=1}^N
\det\left[x_{j}^{-\kappa-1}\exp\left[-\alpha\frac{y_l}{x_j}\right]\right]_{j,l=1}^N
\\
&\times \det\left[1,q_i,\ldots,q_i^{\nu-1},\exp[-q_ix_1],\ldots,\exp[-q_ix_N]\right]_{i=1}^{M},
\end{split}
\end{equation}
where we recall \eqref{nukappadef}. Here, $I_{\kappa}(z)$ denotes the modified Bessel function of the first kind and $Z$ is a normalising constant given by
\begin{equation}\label{Z}
Z=(N!)^{2}\,(-\alpha)^{N\nu+\frac{N(N-1)}{2}}\alpha^{-N\kappa}
{\Delta_{M}(q_{1},...,q_{M})\Delta_{N}(\delta_{1},...,\delta_{N})\prod_{j=1}^{N}\delta_{j}^{\frac{\kappa}{2}}}{\prod_{i=1}^{M}\prod_{j=1}^{N}(\alpha q_{i}-\delta_{j})^{-1}}.
\end{equation}
\end{thm}
In \eqref{MainDensity}, in the determinant in the second line, the notation is such that the first $\nu$ columns are only present for $M>N$ and absent for $M=N$ ($\nu=0$). When two or more parameters become degenerate the corresponding density follows from l'H\^opital's rule, see e.g. Appendix \ref{altThm}. Integrating out the variables $x_j$ leads to the joint probability density of the $y_j$ alone, given by the following
\begin{cor}\label{Mainjpdf} The joint probability density of the squared singular values $y_1,\ldots,y_N$ of the product matrix $Y=GX$ from the ensemble \eqref{JointDensityGX} is equal to
\begin{equation}\label{biensemble}
P(y_1,\ldots,y_N)=\frac{2^NN!}{Z}\det\left[\psi_i(y_j)\right]_{i,j=1}^N\det\left[1,q_i,\ldots,q_i^{\nu-1},\varphi_i(y_1),\ldots,\varphi_i(y_N)\right]_{i=1}^{M},
\end{equation}
where $Z$ is given by \eqref{Z} and we have introduced the following notation
\begin{equation}\label{varphi,psi}
\psi_j(y)=y^{\frac{\kappa}{2}}I_{\kappa}\left(2\sqrt{\delta_j y}\right)\,,\ j=1,\ldots,N,\quad
\varphi_i(y)=\left(\frac{q_i}{\alpha y}\right)^{\frac{\kappa}{2}}\!K_{\kappa}\left(2\sqrt{\alpha q_i y}\right)\,,\ i=1,\ldots,M\,.
\end{equation}
Here, $K_{\kappa}(z)=K_{-\kappa}(z)$ denotes the modified Bessel function of the second kind.
\end{cor}
The joint probability density \eqref{biensemble} lies outside the class of polynomial ensembles \cite{ArnoDries,Arno}, that have many invariance properties \cite{CKW}. Although the two determinants have different sizes, \eqref{biensemble} can be mapped to a bona fide biorthogonal ensemble in the sense of \cite{Borodin:1998} using the Schur complement formula, see Section \ref{detPP}. When we take the limit $\delta_l\to0$ for all $l=1,\ldots,N$ in Theorem \ref{TheoremMainDensity} (and Corollary \ref{Mainjpdf}),
which corresponds to setting $\Omega=0$, we arrive at the joint probability density of the ensemble \eqref{JointDensityindepGX} given by Theorem \ref{MainP2} (and Corollary \ref{CorrP2}). If in Theorem \ref{TheoremMainDensity} we  integrate out the squared singular values $y_1,\ldots,y_N$, which corresponds to integrating out the matrix $G$ in our ensemble \eqref{JointDensityGX},
we arrive at the joint probability density for the squared singular values of the ensemble \eqref{corrWishart}, as stated in Proposition \ref{gWjpdf}. For details of these short-cuts we refer to the Appendix \ref{AppendixA}.

In special cases the joint densities from \eqref{MainDensity} and \eqref{biensemble} were known. Setting all parameters $q_j=\frac{1+\mu}{2\mu}$ equal for all $j=1,\ldots,M$, they reduce to the joint probability densities in \cite{DZL}, and setting furthermore all parameters $\delta_l=\frac{(1-\mu)^2}{4\mu^2}$ equal for all $l=1,\ldots,N$, see \eqref{def-ad}, we reobtain the joint probability densities in \cite{AStr16a}.

Our next result is an example for biorthogonal ensembles \cite{Borodin:1998}. In Proposition \ref{PropPP} we show how to construct the kernel $K_N(x,y)$ for the determinantal point processes \eqref{biensemble}, given by
\begin{equation}
\label{ppKernel}
P(y_1,\ldots,y_N)=\frac{1}{N!}\det\left[K_N(y_i,y_j)\right]_{i,j=1}^N\ ,
\end{equation}
and the resulting $k$-point correlation functions defined in \eqref{k-point}. We quote here the final answer obtained for this kernel.
\begin{thm}\label{MainKernelContourIntegral} The correlation kernel for the ensemble \eqref{biensemble} can be represented as a double contour integral:
\begin{equation}
K_{N}(y_1,y_2)=
2
\oint_{\gamma_{\delta}}\frac{d\eta}{2\pi i}\oint_{\gamma_{q}}\frac{d\zeta}{2\pi i}\left(\frac{\zeta}{\eta}\right)^{\frac{\kappa}{2}}
\frac{I_{\kappa}\left(2\sqrt{\eta y_1}\right)K_{\kappa}\left(2\sqrt{\zeta y_2}\right)}{\eta-\zeta}
\prod_{l=1}^N\frac{\zeta-\delta_l}{\eta-\delta_l}
\prod_{l=1}^M\frac{\eta-\alpha q_l}{\zeta-\alpha q_l}.
\label{MainKerneloint}
\end{equation}
Here, $\gamma_{\delta}$ is a closed contour encircling $\delta_1,\ldots,\delta_N\geq 0$ in a counter-clockwise way, and $\gamma_{q}$ is a closed contour encircling $\alpha q_1,\ldots,\alpha q_M>0$ in counter-clockwise direction, excluding the origin and not intersecting $\gamma_{\delta}$\,.
\end{thm}
Note that the same formula for the kernel remains valid when two or more of the parameters become degenerate. In \eqref{MainKerneloint} we have suppressed a factor $(y_1/y_2)^{\kappa/2}$ on the right hand side. More generally speaking, due to \eqref{ppKernel} the following modification $K_N(y_1,y_2)\to K_N(y_1,y_2) f(y_1)/f(y_2)$ leads to an {\it equivalent} kernel, with the same joint probability density and $k$-point correlation functions. We will frequently use such a transformation. Corollary \ref{Mainjpdf} together with the representation of the kernel \eqref{MainKerneloint} constitutes the solution of the ensemble \eqref{JointDensityGX} for finite matrix sizes. When setting all coupling parameters to zero, $\delta_{l=1,\ldots,N}=0$, we obtain the kernel of the ensemble \eqref{JointDensityindepGX}, see Theorem \ref{Thmker2}.

Next, we turn to the main results taking three different large-$N$ limits at the origin of the spectrum, with matrices $\Omega$ and $Q$ having finite rank perturbations from \eqref{def-ad}. In order to prepare these limits let us introduce the following partial degeneracies among the sets of parameters $\delta_1,\ldots,\delta_N$ and $q_1,\ldots,q_M$, as parametrised by a single  parameter $\mu\in(0,1]$:
\begin{equation}
\delta_{n+1}=\cdots=\delta_N
=\frac{(1-\mu)^2}{4\mu^2}, \qquad
q_{m+1}=\cdots=
q_M=
\frac{1+\mu}{2\mu}.
\label{finiterank1}
\end{equation}
In addition we set
\begin{equation}
\alpha=\frac{1+\mu}{2\mu}.
\label{alpha-def}
\end{equation}
Let us first see what the degeneracies \eqref{finiterank1} imply. Clearly, because all parameters $\delta_l,q_j$ and $\alpha$ are positive, the condition for the convergence of the model \eqref{aqd} is equivalent to
\begin{equation}
\label{Inequal}
1-\frac{q_j}{\alpha}< 1-\frac{\delta_l}{\alpha^2}\ ,\ \ j=1,\ldots,M,\ \ l=1,\ldots, N,
\end{equation}
even before degeneracies \eqref{finiterank1} and \eqref{alpha-def} are imposed. Because we can insert both, on the left hand side, or on the right hand side of \eqref{Inequal} the set of degenerate or non-degenerate parameters, together with \eqref{alpha-def} this leads to the following four different inequalities:
\begin{equation}
1-\frac{2\mu q_j}{1+\mu}< 1-\frac{4\mu^2\delta_l}{(1+\mu)^2}\ ,\ \ j=1,\ldots,m,\ \ l=1,\ldots, n,
\label{nondeg-nondeg}
\end{equation}
for the non-degenerate values,
\begin{eqnarray}
1-\frac{2\mu(1+\mu)}{(1+\mu)2\mu}=0&<& 1-\frac{4\mu^2\delta_l}{(1+\mu)^2}\ ,\ \ l=1,\ldots, n,
\label{deg-nondeg}\\
\ \ 1-\frac{2\mu q_j}{1+\mu}&<& 1-\frac{(1-\mu)^2}{(1+\mu)^2}=\frac{4\mu}{(1+\mu)^2}\ ,\ \ j=1,\ldots,m,
\label{nondeg-deg}
\end{eqnarray}
for one of each sets being degenerate, and
\begin{equation}
0< \frac{4\mu}{(1+\mu)^2}\ ,
\label{deg-deg}
\end{equation}
for both sets being degenerate, which is trivially satisfied due to $\mu\in(0,1]$. This brings us to the sets of parameters that will be relevant for our limiting kernels. In view of \eqref{nondeg-nondeg} we define the quantities
\begin{eqnarray}
\pi_l&=&\frac{(1+\mu)^2}{4\mu } \left(1-\frac{4\mu^2\delta_l}{(1+\mu)^2}\right),\ \ l=1,\ldots, n,
\label{pi-def}\\
\theta_j&=&  \frac{(1+\mu)^2}{4\mu } \left( 1-\frac{2\mu q_j}{1+\mu}\right) ,\ \ j=1,\ldots,m,
\label{theta-def}
\end{eqnarray}
which satisfy
\begin{equation}
\theta_j < \pi_l\ ,\ \ j=1,\ldots,m,\ \ l=1,\ldots, n,
\label{t<pi}
\end{equation}
due to \eqref{nondeg-nondeg}, normalised by the right hand side of \eqref{deg-deg}. The bounds for these quantities resulting from \eqref{deg-nondeg} and \eqref{nondeg-deg}, respectively, are
\begin{equation}
0<\pi_l \ ,\ \ l=1,\ldots n,\ \ \mbox{and} \ \ \theta_j<1 \ ,\ \ j=1,\ldots,m,
\label{bounds}
\end{equation}
These will be useful when taking the three different large-$N$ limits next.

At large-$N$ limit the parameter $\mu$ is now considered as a function of $N$, $\mu=\mu(N)$, taking values in $(0,1]$. We will assume that also the remaining, non-degenerate parameters may become functions of $N$, that is, $\pi_l=\pi_l(N)$, $l=1,\ldots,n$ and $\theta_j=\theta_j(N)$, $j=1,\ldots,m$. Let us recall here that the parameters $\kappa$ and $\nu$ in \eqref{nukappadef} as well as $n,m$ will be kept fixed in these limits.

In the first limiting regime (I) the function $\mu(N)$ is such that $\lim_{N\to\infty}\mu(N) N=\infty$. This regime includes the situation when $\mu$ is constant. In this case the bounds in \eqref{bounds} together with \eqref{pi-def} ensure that the limits exist  
\begin{equation}
\lim_{N\to\infty}\frac1N \pi_l(N)=0\ ,\ \   l=1,\ldots, n.
\label{pitheta-LimI}
\end{equation}
Noting the restriction \eqref{t<pi}, in order to obtain some nontrivial results we assume that $\theta_j(N)$ grows linearly with $N$ as $N\rightarrow \infty$:
\begin{equation}
\lim_{N\to\infty}\frac1N \theta_j(N) =\hat{\theta}_j \in(-\infty,0] \ ,\ \ j=1,\ldots,m.
\label{pitheta-LimI'}
\end{equation}
The corresponding limiting kernel that we will encounter in Theorem \ref{hardlimits} below is defined as
\begin{align}
\mathbb{K}_{\mathrm{I}}^{(m)}(y_1,y_2) &=
\int_{0}^{\infty}  dt  \oint_{\Gamma_{\mathrm{0}}}   \frac{ds}{2\pi i} \,     t^{\kappa-1} s^{-\kappa-1}
\exp[s-t]  \nonumber \\ & \quad \times
\oint_{\Gamma_{\mathrm{out}}}  \frac{dv}{2\pi i}  \oint_{\Gamma_{\mathrm{in}}}  \frac{du}{2\pi i}
\frac{ \exp\left[-\frac{v}{s}y_1+\frac{u}{t}y_2 \right]}{u-v}e^{-\frac{1}{u}+\frac{1}{v}}
\Big(\frac{v}{u}\Big)^{\nu-m}
\prod_{l=1}^{m}\frac{v- \hat{\theta}_{l}}{u-\hat{\theta}_{l}}.
\label{kernelsub}
\end{align}
Here, $\Gamma_{\mathrm{0}}$ is a closed contour encircling the origin counter-clockwise, $\Gamma_{\mathrm{in}}$ is a closed contour encircling $\{0, \hat{\theta}_{1}, \ldots, \hat{\theta}_{m}\}$ in counter-clockwise direction, and $\Gamma_{\mathrm{out}}$ is a   closed contour enclosing the contour $\Gamma_{\mathrm{in}}$ counter-clockwise. This kernel generalises the Meijer $G$-kernel for the product of two independent matrices \cite{KZ14} by a set of finite rank perturbations, and it holds $\mathbb{K}_{\mathrm{I}}^{(m=0)}(x,y)= \mathbb{K}_{\rm Meijer}(x,y)$, see \eqref{KMeijer} and \eqref{KIfinal} for integral representations of the latter. It was shown in \cite{CKW} for this Meijer $G$-kernel that it can also be written as a double integral of the Bessel kernel. The same relation extends to the kernels with finite rank perturbations, comparing the second line of \eqref{kernelsub} with \eqref{kernelsup} below at $n=0$, the generalised Bessel kernel. We will show that a representation of the kernel \eqref{kernelsub} with only two integrals exists, cf. \eqref{KI2intonly}.

In the second limiting regime (II) the function $\mu(N)$ behaves as $\lim_{N\to\infty} \mu(N) N=\tau/4$, with $\tau>0$. In this limit, considering the definitions \eqref{pi-def} and \eqref{theta-def} together with \eqref{t<pi} and \eqref{bounds}, we may assume that
\begin{equation}
\lim_{N\to\infty}\frac1N \pi_l(N)=\hat{\pi}_l\in[0,\frac{1}{\tau})\ ,\ \   l=1,\ldots, n,\ \ \mbox{and} \ \ \lim_{N\to\infty}\frac1N \theta_j(N)=\hat{\theta}_j \in\cap_{k=1}^{n} (-\infty,\hat{\pi}_k] \ ,\ \ j=1,\ldots,m.
\label{pitheta-LimII}
\end{equation}
The limiting kernel in this regime (II) is a parameter dependent family of interpolating limiting kernels. It is defined as
\begin{align}
\mathbb{K}_{\mathrm{II}}^{(n,m)}(y_1,y_2;\tau)&= \frac{2}{\tau}\oint_{\Gamma_{\mathrm{out}}} \frac{dv}{2\pi i} \oint_{\Gamma_{\mathrm{in}}} \frac{du}{2\pi i}  I_{\kappa}\left(2\sqrt{(1-\tau v)y_1/\tau^2}\right)K_{\kappa}\left(2\sqrt{(1-\tau u)y_2/\tau^2}\right) \nonumber\\  &  \quad \times   e^{-\frac{1}{u}+\frac{1}{v}}\frac{1}{u-v}\Big(\frac{1-\tau u}{1-\tau v}\Big)^{\frac\kappa2} \Big(\frac{v}{u}\Big)^{\nu+n-m}
\prod_{k=1}^{n}\frac{u- \hat{\pi}_{k}}{v-\hat{\pi}_{k}}  \prod_{l=1}^{m}\frac{v- \hat{\theta}_{l}}{u-\hat{\theta}_{l}}.
\label{kernelcrit}
\end{align}
Here, $\Gamma_{\mathrm{in}}$ is defined as in the first limit (I), such that  $\Re(u)<1/\tau$ for $u\in \Gamma_{\mathrm{in}}$, and
$\Gamma_{\mathrm{out}}$ is a  closed  contour encircling the parameters $\{\hat{\pi}_{1},\ldots,  \hat{\pi}_m\}$ and contour $\Gamma_{\mathrm{in}}$ in counter-clockwise direction. It generalises the interpolating kernel of \cite{AStr16b}, where no such parameters $\hat{\pi}_k$ and $\hat{\theta}_l$ were present (although the representation of the kernel in \cite{AStr16b} is different), and the kernel of \cite{DZL}, where the parameters $\hat{\pi}_{k}$ are present and $\hat{\theta}_{k}$ are absent. The kernel \eqref{kernelcrit} can also be written as a double integral of the generalised Bessel kernel \eqref{kernelsup} below, cf. \eqref{KII4int}.

The third limiting regime (III) is given by $\lim_{N\to\infty}\mu(N) N=0$, with limiting parameter assumptions
\begin{equation}
\lim_{N\to\infty}\frac1N \pi_l(N)=\hat{\pi}_l\in[0,\infty)\ ,\ \   l=1,\ldots, n,\ \ \mbox{and} \ \ \lim_{N\to\infty}\frac1N \theta_j(N)= \hat{\theta}_j \in\cap_{k=1}^{n} (-\infty,\hat{\pi}_k] \ ,\ \ j=1,\ldots,m.
\label{pitheta-LimIII}
\end{equation}
The corresponding limiting kernel is a generalisation of the Bessel kernel with finite rank perturbations,
\begin{align}
\mathbb{K}_{\mathrm{III}}^{(n,m)}(y_1,y_2) =
\oint_{\Gamma_{\mathrm{out}}}  \frac{dv}{2\pi i}  \oint_{\Gamma_{\mathrm{in}}}  \frac{du}{2\pi i}
\frac{ \exp\left[-{y_1}v+{y_2}u \right]}{u-v}e^{-\frac{1}{u}+\frac{1}{v}} \Big(\frac{v}{u}\Big)^{\nu+n-m}\prod_{k=1}^{n}\frac{u- \hat{\pi}_{k}}{v-\hat{\pi}_{k}}  \prod_{l=1}^{m}\frac{v- \hat{\theta}_{l}}{u-\hat{\theta}_{l}},
\label{kernelsup}
\end{align}
where $\Gamma_{\mathrm{in}}$ is a closed contour encircling $\{0, \hat{\theta}_{1}, \ldots, \hat{\theta}_{m}\}$ counter-clockwise, and $\Gamma_{\mathrm{out}}$ a closed contour enclosing $\{\hat{\pi}_{1},\ldots,  \hat{\pi}_m\}$ and the contour $\Gamma_{\mathrm{in}}$ counter-clockwise. Without finite rank perturbations it coincides with the Bessel kernel, $\mathbb{K}_{\mathrm{III}}^{(0,0)}(x,y)= \mathbb{K}_{\rm Bessel}(x,y)$, as shown in \cite{DF06}, see \eqref{TheBessel} and \eqref{BesselC} for representations of the latter. At $m=0$ the kernel \eqref{kernelsup} was found earlier in \cite{DZL} and agrees with the one from \cite[Theorem~15]{DF06} for a different ensemble. The kernel \eqref{kernelsup} enjoys a formal duality relation, $\mathbb{K}_{\mathrm{III}}^{(n,m)}(x,y)\to\mathbb{K}_{\mathrm{III}}^{(m,n)}(y,x)|_{\nu\to-\nu,\,\hat{\pi}_k\leftrightarrow-\hat{\theta}_l}$, as can be seen already in the corresponding ensemble \eqref{corrWishart}, by interchanging matrices $\Sigma\leftrightarrow Q$ and $N\leftrightarrow M$. It thus holds already for the kernel at finite-$N$, cf. \eqref{kernelGeneralizedWishart}, and in particular also for the extended Airy kernel of \cite{BP} at the soft edge.

In all three kernels the contours can be chosen differently, for example when showing that they are integrable in the sense of \cite{IIKS}, see Corollaries \ref{Int2}, \ref{KII-integrable} and \ref{Int1}, respectively. Theorems \ref{Thmker1}, \ref{ThmLimker1} and Corollary \ref{Int1} solve an open problem stated in \cite[Section 7.2]{DF06}. The integrability found here implies that the asymptotic analysis of all three kernels can be formulated as a Riemann-Hilbert problem. These three kernels can all be obtained from the kernel \eqref{MainKerneloint} of ensemble \eqref{JointDensityGX}.
\begin{thm} [Hard edge scaling limits] \label{hardlimits}
Consider the correlation kernel \eqref{MainKerneloint} from Theorem \ref{MainKernelContourIntegral}, with fixed non-negative integers $\nu=M-N$ and $\kappa=L-N$, together with the definitions \eqref{pi-def} and \eqref{theta-def}.
With the three kernels defined above, the following limits hold uniformly for any $x, y$ in a compact subset of $(0,\infty)$ as $N\rightarrow \infty$.
\begin{itemize}
\item[(I)] Suppose that $\mu(N) N \ra \infty$ and \eqref{pitheta-LimI'} hold true, then we have 
\begin{equation}
\lim_{N\to\infty}
\frac{\mu(N)}{N} K_N\left( \frac{\mu(N)}{N} x, \frac{\mu(N)}{N} y\right)
\left(\frac{y}{x}\right)^{\frac{\kappa}{2}}
=
\mathbb{K}_{\mathrm{I}}^{(m)}(x,y),
\nonumber
\end{equation}
with the limiting parameters $\hat{\theta}_{k=1,\ldots,m}\in(-\infty,0]$.
\item[(II)] Suppose that  $\mu(N) N \ra\tau/4$ with $\tau>0$ and \eqref{pitheta-LimII} hold true, then we have 
\begin{equation}
\lim_{N\to\infty} \frac{1}{4N^{2}} K_N\left(   \frac{x}{4N^{2}},   \frac{y}{4N^{2}}\right)
=
\mathbb{K}_{\mathrm{II}}^{(n,m)}(x,y;\tau),
\nonumber
\end{equation}
with  $\hat{\pi}_{l=1,\ldots,n}\in[0,1/\tau)$ and the limiting parameters $\hat{\theta}_{k=1,\ldots,m}\in \cap_{l=1}^{n} (-\infty,\hat{\pi}_l]$.
\item[(III)] Suppose that $\mu(N) N \ra 0$ and \eqref{pitheta-LimIII} hold true, then we have 
\begin{equation}
\lim_{N\to\infty}\frac{1}{4N^2}
K_{N}\left(\frac{{x}}{4N^2},\frac{{y}}{4N^2}\right)
\ e^{\frac{1}{2\mu(N) N}(\sqrt{{y}}-\sqrt{{x}})}
=
\frac{1}{2}(x y)^{-\frac{1}{4}}\mathbb{K}_{\mathrm{III}}^{(n,m)}(\sqrt{x},\sqrt{y}),
\nonumber
\end{equation}
with $\hat{\pi}_{l=1,\ldots,n}\in[0,\infty)$ and the limiting parameters $\hat{\theta}_{k=1,\ldots,m}\in \cap_{l=1}^{n} (-\infty,\hat{\pi}_l]$.
\end{itemize}
\end{thm}
The kernel in limit (I) is also obtained from ensemble \eqref{JointDensityindepGX} of the product of two independent random matrices $Y=GX$, where the matrix elements of $X$ are correlated, see Theorem \ref{ThmLimker2}. The kernel in limit (III) is also obtained from ensemble \eqref{corrWishart} of a single random matrix, that has matrix elements correlated from both sides, see Theorem \ref{ThmLimker1}. The fact that in Theorem \ref{hardlimits} (III) we obtain this generalised Bessel kernel, rescaled and with square root arguments, was already observed and explained in \cite{AStr16a}. It is due to the fact that in this strongly coupled limit the squared singular values of $Y=X^*X$ are obtained. The kernel in the limit (II) is called interpolating in the following sense.
\begin{thm}[Interpolating kernel]\label{Interpolate}
The parameter dependent family of kernels \eqref{kernelcrit} is interpolating between the Meijer $G$-kernel and the Bessel kernel, both with finite rank perturbations. Namely, it holds for $x,y$ in any compact subset of $\mathbb{R}_+$:
\begin{itemize}
\item[a)]
$\lim_{\tau\to\infty} \tau \mathbb{K}_{\mathrm{II}}^{(n,m)}(\tau x,\tau y;\tau)\left(\frac{y}{x}\right)^{\frac{\kappa}{2}} = \mathbb{K}_{\mathrm{I}}^{(m)}(x,y)\,,$
\item[b)]
$\lim_{\tau\to 0+} \mathbb{K}_{\mathrm{II}}^{(n,m)}(x,y;\tau) e^{2(\sqrt{y}-\sqrt{x})/\tau}= \frac{1}{2}(xy)^{-\frac14}\mathbb{K}_{\mathrm{III}}^{(n,m)}(\sqrt{x},\sqrt{y})\,.$
\end{itemize}
\end{thm}
This generalises the interpolating kernel derived in \cite{DZL} for $m=0$, and in \cite{AStr16b} for $m=n=0$, where an alternative integral representation was given. It is an open problem to map these two representations.

In \cite{BBP} a modification of the Tracy-Widom distribution for the largest eigenvalue was observed from finite rank perturbations.
One could ask if a similar phenomenon occurs here for the distribution of the smallest eigenvalue, applying a Fredholm determinant representation in terms of the three different limiting kernels from our ensemble \eqref{JointDensityGX}. Here, we stress that
the difference in scaling and thus of the fluctuations in limits (I) and (II) indicates such a transition. On the other hand, the same scaling in limits (II) and (III) indicates a smooth interpolation, consistent with the findings of \cite{AStr16b} without finite rank perturbations.

The remainder of this article is organised as follows. In the next Section \ref{JPD} we derive the joint densities of squared singular values for all three ensembles. In Section \ref{detPP} we show that all three ensembles represent determinantal point processes, and derive their kernels at finite matrix size as double contour integrals. The last Section  \ref{large-N} is devoted to the asymptotic analysis of the three kernels at the origin, their integrability and the interpolating property of the kernel in limit (II).
\section{Joint Probability Densities}\label{JPD}
For pedagogical reasons we will start with the derivation of the joint probability density of the generalised Wishart ensemble \eqref{corrWishart}, where we extend the results of \cite{BP} to rectangular matrices. Then, we move to the product of two independent matrices \eqref{JointDensityindepGX} where one matrix has correlated entries, before coupling these two matrices in our  most general ensemble \eqref{JointDensityGX}.
\subsection{Joint probability density of the generalised Wishart ensemble}\label{JPDgW}
This ensemble that is correlated from two sides is defined following \cite{BP}. Let $X$ be a complex random matrix of size $M\times N$ with $M-N=\nu\geq0$, and its matrix entries distributed as
\begin{equation}\label{JointDensityX}
\mathcal{P}_1(X)= c_1
\exp\left[-\Tr\left(X\Sigma X^*\right)-\Tr\left(QXX^*\right)\right] .
\end{equation}
Here, $\Sigma$ is a fixed Hermitian matrix of size $N\times N$ with eigenvalues $\sigma_1, \ldots, \sigma_N$, and
$Q$ is a fixed Hermitian matrix of size $M\times M$ with eigenvalues $q_1, \ldots, q_M$. For simplicity we first assume that these are all pairwise distinct. The normalisation constant reads $c_1=\pi^{-NM}\prod_{i=1}^M\prod_{j=1}^N(q_i+\sigma_j)$. When thinking of $\Sigma$ and $Q$ as originating from given empirical covariance matrices we would choose them to have only
positive eigenvalues, ensuring convergence.

For what follows below we will choose $Q$ to have positive eigenvalues, and allow $\Sigma$ to have also negative eigenvalues, which leads to the constraint
\begin{equation}
q_i+\sigma_j>0\ ,\ \ \forall i=1,\ldots,M\ ,\ \  \forall j=1,\ldots,N \   .
\label{constraint3}
\end{equation}
Note that the ensemble \eqref{JointDensityX} is different from the double correlated Wishart ensemble, with distribution $\tilde{\mathcal{P}}(X)\sim\exp[-\Tr(QX\Sigma X^*)]$, that has been considered in \cite{WWG}. Our first result is the following
\begin{prop}\label{gWjpdf} The joint probability density of the squared singular values $x_1,\ldots,x_N$ of $X$ distributed according to \eqref{JointDensityX} equals
\begin{equation}\label{biensemble2}
P_1(x_1,\ldots,x_N)=\frac{1}{Z_1}\det\left[e^{-\sigma_i x_j}\right]_{i,j=1}^N\det\left[1,q_i,\ldots,q_i^{\nu-1},e^{-q_i x_1},\ldots,e^{-q_i x_N}\right]_{i=1}^{M},
\end{equation}
where the normalising constant is given by
\begin{align}
Z_1=N!\,(-1)^{N\nu}{\Delta_{M}(q_{1},...,q_{M})\Delta_{N}(\sigma_{1},...,\sigma_{N})}{\prod_{i=1}^{M}\prod_{j=1}^{N}( q_{i}+\sigma_{j})^{-1}} .
\label{Z"}
\end{align}
\end{prop}
In the case of quadratic matrices with $\nu=0$, where the $x$-independent columns in the second determinant of the distribution \eqref{biensemble2} are absent, this result was derived in \cite{BP}.
\begin{proof}
As the first step we will reduce the rectangular matrix $X$ to a quadratic matrix, following \cite{Jonit}. In the sequel we will denote by $O$ matrices with zero entries. Set
\begin{equation}
X=U\left(\begin{array}{c}
           X_0 \\
           O
         \end{array}\right),
         \label{Xdecomp}
\end{equation}
where $U\in U(M)/U(N)\times U(\nu)$ is an $M\times M$ unitary matrix, and $X_0$ is an $N\times N$ complex matrix. We note that
$X^*X=X_0^*X_0$ is of size $N\times N$ and thus the matrices $X$ and $X_0$ have the same squared singular values. Inserting the Jacobian \cite{Jonit} the joint probability distribution of $X_0$ and $U$ following from \eqref{JointDensityX} is thus proportional to
\begin{align}
\mathcal{P}_1(X)[dX]\sim
e^{-\Tr\left(\Sigma X_{0}^{*}X_{0}\right)}e^{
-\Tr\left(QU\left(\begin{array}{cc}
                                                              X_0X_0^* & O \\
                                                              O & O
                                                            \end{array}
\right)U^*\right)}\det\left[X_0^*X_0\right]^{\nu}[dX_0]d\mu(U) .
\label{X0decomp}
\end{align}
Here, $d\mu(U) $ denotes the corresponding Haar measure, $[dX]=\prod_{i=1}^M\prod_{j=1}^N d{X_{i,j}}^R d{X_{i,j}}^I$ denotes the flat Lebesgue measure over all independent matrix elements $X_{ij}$, their real and imaginary parts ${X_{i,j}}^R$ and ${X_{i,j}}^I$.
In the following we will suppress all proportionality constants that can be determined, and fix the normalisation $Z_1$ of the joint probability density of squared singular values \eqref{biensemble2} only at the end. The singular value decomposition for $X_0$ can be written as
\begin{equation}
X_0= \mathfrak{U} \Lambda_x^{\frac{1}{2}}P^*\,, \ \
\Lambda_x^{\frac{1}{2}}=\left(\begin{array}{cccc}
                       \sqrt{x_1} & 0 & \ldots & 0 \\
                       0 & \sqrt{x_2} & \ldots & 0 \\
                        &  & \ddots &  \\
                       0 & 0 & \ldots & \sqrt{x_N}
                     \end{array}
\right),
\label{svdecomp}
\end{equation}
where $\mathfrak{U}$ and $P$ are unitary matrices with $P\in\, U(N)$ and $\mathfrak{U}\in\, U(N)/U(1)^N$. Here and further on we will use the notion $\Lambda_{x}=\diag\left(x_{1},\ldots,x_{N}\right)$ for this and other sets of variables. The measure $[dX_{0}]$ decomposes as
\begin{equation}
[dX_0]\sim\Delta_N(x_1,\ldots,x_N)^2dx_1\ldots dx_Nd\mu(\mathfrak{U})d\mu(P)\,,
\label{JacobiX}
\end{equation}
and we arrive at
\begin{equation}
\begin{split}
\mathcal{P}_1(X)[dX]\sim&\ e^{-\Tr\left(\Sigma P\Lambda_{x}P^{*}\right)}
e^{-\Tr\left[QU\left(\begin{array}{cc}
                       \mathfrak{U} \Lambda_x \mathfrak{U}^* & O \\
                       O & O
                     \end{array}
\right)U^*\right]} \prod_{k=1}^Nx_k^{\nu}\ \Delta_N(x_1,\ldots,x_N)^2\\
&\times dx_1\ldots dx_Nd\mu(U)d\mu(\mathfrak{U})d\mu(P)\,.
\label{preP1}
\end{split}
\end{equation}

To obtain \eqref{biensemble2} we need to compute the group integrals over $P$, and over $U$ and $\mathfrak{U}$ that have already decoupled. Furthermore, we have to diagonalise the fixed matrices, $\Sigma=V\Lambda_\sigma V^*$ and $Q=\tilde{V}\Lambda_q\tilde{V}^*$, and absorb these extra factors $V$ and $\tilde{V}$ of fixed unitary matrices through the invariance of the Haar measures of these group integrals. For the integral over $P$ this is straightforward and we can readily apply the standard Harish-Chandra--Itzykson--Zuber (HCIZ) integral formula, reading \cite{HC,IZ}:
\begin{align}
\int_{U(N)}d\mu(P)e^{-\Tr\left(\Lambda_\sigma P\Lambda_{x}P^{*}\right)}={\rm{const.}}
\frac{\det\left[e^{-\sigma_{i} x_{j}}\right]_{i,j=1}^{N}}{\Delta_{N}(\sigma_1,\ldots,\sigma_N)\Delta_{N}(x_1,\ldots,x_N)}\ .
\label{HCIZ}
\end{align}
The constant that is independent of the $\sigma_j$ and $x_j$ can be determined and depends only on the convention in normalising the Haar measure of the unitary group.

Next we can turn to the integrals over $U$ and $\mathfrak{U}$. Here, an additional integral over $\mathfrak{U}_{1}$ can be introduced that we choose to be over $U\left(\nu\right)$,
\begin{equation}
e^{-\Tr\left[QU\left(\begin{array}{cc}
                       \mathfrak{U} \Lambda_x \mathfrak{U}^*  & O \\
                       O & O
                     \end{array}
\right)U^*\right]}
=\mbox{const.} \int d\mu\left(\mathfrak{U}_{1}\right)e^{-\Tr\left[\tilde{V}\Lambda_q\tilde{V}^* U\left(\begin{array}{cc}
                       \mathfrak{U}  & O \\
                       O & \mathfrak{U}_{1}
                     \end{array}
\right)\widetilde{\Lambda}_{x} \left(\begin{array}{cc}
                       \mathfrak{U}^{*}  & O \\
                       O & \mathfrak{U}_{1}^{*}
                     \end{array}
\right)   U^*\right]}
\,,\nonumber
\end{equation}
where the extra zeros are denoted by the diagonal $M\times M$ matrix
\begin{equation}
\widetilde{\Lambda}_{x}=\diag\left(x_{1},\ldots,x_{N},0,\ldots,0\right)\,.\nonumber
\end{equation}
The three integrations  over $U$, $\mathfrak{U}$, and $\mathfrak{U}_{1}$ together parametrise the coset space $U(M)/U(1)^N$, which can be used to absorb $\tilde{V}$ by invariance of the corresponding Haar measure. In order to apply the HCIZ formula \eqref{HCIZ} we need to take into account that $\nu$ eigenvalues of $\widetilde{\Lambda}_{x}$ are equal to zero, which can be obtained by l'H\^opital's rule. We thus arrive at
\begin{equation}
\int d\mu(U)d\mu(\mathfrak{U}) e^{-\Tr\left[QU\left(\begin{array}{cc}
                        \mathfrak{U} \Lambda_x \mathfrak{U}^* & O \\
                       O & O
                     \end{array}
\right)U^*\right]}
=\const  \frac{{\det}\left[1,q_i,\ldots, q_i^{\nu-1},e^{-q_ix_1},\ldots,e^{-q_ix_N}\right]_{i=1}^M}{\Delta_{M}\left(q_1,\ldots,q_M\right)\Delta_{N}\left(x_1,\ldots,x_N\right) \prod_{k=1}^Nx_k^{\nu}}\,,
\label{HCIZdeg}
\end{equation}
where $U$ is integrated over $U(M)/U(N)\times U(\nu)$, $\mathfrak{U}$ over $U(N)/U(1)^N$ (and $\mathfrak{U}_{1}$ over $U\left(\nu\right)$). Integrating \eqref{preP1} over the corresponding coset spaces, from \eqref{HCIZdeg} together with \eqref{HCIZ}
we arrive at \eqref{biensemble2}, up to the normalisation constant $Z_1$. For its calculation we apply the generalisation  \cite[Appendix C]{Kieburg:2010}  of the Andr\'eief formula \cite{And} that follows from simple linear algebra. We quote the following form for later use:
\begin{eqnarray}
&&\prod_{j=1}^N\int_0^\infty dx_j
\det\left[
\begin{array}{r}
R_{a,b}\big|_{1\leq a\leq k}^{1\leq b\leq N+k}\\
\\
\psi_b(x_a)\big|_{1\leq a\leq N}^{1\leq b\leq N+k}\\
\end{array}
\right]
\det\left[
S_{b,a}\big|^{1\leq a\leq l}_{1\leq b\leq N+l}\ \
\varphi_b(x_a)\big|^{1\leq a\leq N}_{1\leq b\leq N+l}
\right]
\label{genAnd}\\
&=& (-1)^{kl}N!\det\left[
\begin{array}{cc}
O_{k\times l}
& R_{a,b}\big|_{1\leq a\leq k}^{1\leq b\leq N+k}
\\
&\\
S_{b,a}\big|^{1\leq a\leq l}_{1\leq b\leq N+l}
&\int_0^\infty dx\ \varphi_b(x)\psi_a(x) \big|^{1\leq a\leq N+k}_{1\leq b\leq N+l} \\
\end{array}
\right].
\nonumber
\end{eqnarray}
Here, we have explicitly spelled out the dimension of the matrix block with zero elements $O$. This identity is valid for two sets of functions $\psi_j(x)$ and $\varphi_k(x)$ that are suitably integrable, and two constant matrices $R$ and $S$. The integration domains can also be chosen differently, cf. \cite{Kieburg:2010}. When specifying to $k=0$, that is in the absence of matrix $R$, and to $l=\nu$ with $S_{b,a}=q_b^{a-1}$, we obtain
\begin{eqnarray}
&&\prod_{k=1}^N\int_0^\infty dx_k \det\left[{\psi_{i} (x_{j})}\right]_{i,j=1}^{N}{\det}\left[1,q_j,\ldots, q_j^{\nu-1},\varphi_j(x_1),\ldots,\varphi_j(x_N)\right]_{j=1}^M
\nonumber\\
&&=N!\det\left[1,q_{j},\ldots,q_{j}^{\nu-1},\int_0^\infty dx{\psi_{1}(x)\varphi_{j}(x)},\ldots,\int_0^\infty dx{\psi_{N}(x)\varphi_{j}(x)}\right]_{j=1}^{M}\,.
\label{gAndreief}
\end{eqnarray}
The standard Andr\'eief formula is obtained when also setting $\nu=0$, when the first $\nu$ columns on left and right hand side (and thus matrix $S$ in \eqref{genAnd}) are absent. Inserting
\begin{equation}
\label{psiphi1}
\psi_i^{(1)}(x)=e^{-\sigma_ix}\ \  \mbox{and}\ \ \varphi_j^{(1)}(x)=e^{-q_j x}
\end{equation}
into \eqref{gAndreief} we have the resulting simple integral that we define for later purpose:
\begin{equation}
\label{I1def}
I_{i,j}^{(1)}=\int_{0}^{\infty}dx\,e^{-q_{i}x}e^{-\sigma_{j}x}=\frac{1}{q_{i}+\sigma_{j}}\ ,\ \ \mbox{for}\ \ i=1,\ldots,M,\ \ j=1,\ldots,N\,.
\end{equation}
We can now apply the generalised Cauchy determinant derived by Basor and Forrester \cite[Lemma 2]{Basor:1994} to \eqref{gAndreief} and \eqref{I1def}
\begin{align}
\det\left[1,q_{j},\ldots,q_{j}^{\nu-1},\frac{1}{q_{j}+\sigma_{1}},\ldots,\frac{1}{q_{j}+\sigma_{N}}\right]_{j=1}^{M}=
(-1)^{N\nu}\frac{\Delta_N(\sigma_1,\ldots,\sigma_N)\Delta_M(q_1,\ldots,q_M)}{\prod_{i=1}^N\prod_{j=1}^M(q_j+\sigma_i)},
\label{genCauchy}
\end{align}
with $\nu=M-N$.
Equation \eqref{genCauchy} yields the normalisation in \eqref{Z"}.
\end{proof}
\subsection{Joint probability density of the product of two independent correlated matrices}\label{JPD2indep}
Let us consider two independent complex random matrices $G$ of size $L\times M$, and $X$ of size $M\times N$. The matrix elements of each $G$ and $X$ are correlated, given by the probability distribution
\begin{align}
\mathcal{P}_2(G,X)=&\,c_2\exp\left[-\Tr\left(WGG^*\right)-\Tr\left(QXX^*\right)\right]\ .
\label{P2}
\end{align}
Here, $W$ is a fixed Hermitian matrix of size $L\times L$ with positive eigenvalues, and $Q$ is a fixed Hermitian matrix of size $M\times M$ with pairwise non-degenerate eigenvalues $q_1,\ldots, q_M>0$. In what follows we will restrict ourselves to the case of $W=\alpha \eins_L$ being proportional to the identity, with $\alpha>0$. The reason is that for generic $W$ the joint probability density of squared singular values of the product $Y=GX$ is no longer determinantal, as we will show. The normalising constant $c_2$ in \eqref{P2} can be computed by performing the Gaussian integrals over $G$ and $X$, which leads to $c_2=\pi^{-M(L+N)}\alpha^{ML}\prod_{j=1}^Mq_j^{N}$. Our first result is the following
\begin{thm}\label{MainP2}
Denote by $x_1$, $\ldots$, $x_N$ the squared singular values of $X$ and by $y_1$, $\ldots$, $y_N$ the squared singular values of $Y=GX$, with $G$ and $X$ distributed according to \eqref{P2} with $W=\alpha\eins_{L}$. Their joint probability density is given by
\begin{eqnarray}
P_2(x_1,\ldots,x_N;y_1,\ldots,y_N)&=&\frac{1}{Z_2}
\prod_{j=1}^N{y_j^\kappa} \Delta_N(y_1,\ldots,y_N)
\det\left[{x_j^{-\kappa-1}}e^{-\alpha\frac{y_k}{x_j}}\right]_{j,k=1}^N\nonumber\\
&&\times\det\left[1,q_i,\ldots,q_i^{\nu-1},e^{-q_ix_1},\ldots,e^{-q_ix_N}\right]_{i=1}^{M},
\label{P2jpdf}
\end{eqnarray}
where $Z_2$ is a normalising constant given by
\begin{align}
Z_2=\left(N!\right)^{2}(-\alpha)^{N\nu+\frac{N(N-1)}{2}}\alpha^{-N\kappa-NM}
\left(\prod_{l=1}^{N}\Gamma(\kappa+l)\Gamma(l)\right){\Delta_{M}(q_1,\ldots,q_M)}{\prod_{k=1}^{M}q^{-N}_{k}}\,.
\label{ZGX}
\end{align}
\end{thm}
\begin{proof}
We start with a general fixed Hermitian matrix $W$ up to the point where it becomes clear, why only $W=\alpha \eins_L$ leads to a determinantal point process. We begin by decomposing the random matrix $X$ as in \eqref{Xdecomp}, leading immediately to
\begin{align}
\mathcal{P}_2(G,X)[dG][dX]&\sim e^{-\Tr\left(WGG^*\right)}
e^{
-\Tr\left(QU\left(\begin{array}{cc}
                                                              X_0X_0^* & O \\
                                                              O & O
                                                            \end{array}
\right)U^*\right)}\det\left[X_0^*X_0\right]^{\nu}[dG][dX_0]d\mu(U)\,,
\label{P2decomp}
\end{align}
in analogy to \eqref{X0decomp}, where the corresponding measures are defined. We are interested in the singular values of the product matrix $Y=GX$, and in view of \eqref{Xdecomp} we set $\widehat{G}=GU$ which is again a matrix of size $L\times M$. We split this matrix $\widehat{G}$ into its first $N$ columns and its remaining $\nu$ columns by introducing the matrices $\widehat{G}_0$ of size $L\times N$ and $\widehat{G}_1$ of size $L\times \nu$ as
\begin{equation}
GU=\widehat{G}=\left(\widehat{G}_0,\widehat{G}_1\right)\, ,
\label{Gsplit}
\end{equation}
with $[dG]=[d\widehat{G}_0][d\widehat{G}_1]$.
It immediately follows that  the matrices $\widehat{G}_0X_0$ and ${G}X$ have the same singular values. Furthermore, the product of $GG^{*}$ can be written as
\begin{equation}
GG^*=\widehat{G}\widehat{G}^*=\left(\widehat{G}_0,\widehat{G}_1\right)\left(\begin{array}{c}
                                                                       \widehat{G}_0^* \\
                                                                       \widehat{G}_1^*
                                                                     \end{array}
\right)=\widehat{G}_0\widehat{G}_0^*+\widehat{G}_1\widehat{G}_1^*\,.\nonumber
\end{equation}
For that reason the matrix $\widehat{G}_1$ completely decouples in the exponent in \eqref{P2decomp} and can be integrated out, being part of the normalisation. Now consider the change of variables for invertible $X_0$:
\begin{equation}
\label{changeGY}
\widehat{G}_0\longrightarrow Y=\widehat{G}_0X_0,\;\;\; X_0\longrightarrow X_0\,.
\end{equation}
Note that $Y$ is a matrix of size $L\times N$ and that the Jacobian of this transformation is given by $\det\left[X_0^*X_0\right]^{-L}$. Thus we obtain that the joint distribution of $Y$, $X_0$ and $U$ is proportional to
\begin{equation}
\begin{split}
e^{-\Tr\left(WY\left(X_0^*X_0\right)^{-1}Y^*\right)}
e^{-\Tr\left(QU\left(\begin{array}{cc}
                       X_0X_0^* & O \\
                       O & O
                     \end{array}
\right)U^*\right)}\det\left[X_0^*X_0\right]^{\nu-L}[dY][dX_0]d\mu(U).
\end{split}
\end{equation}
The singular value decomposition for $Y$ is in analogy to that of $X_0$ in \eqref{svdecomp}, using the same notation:
\begin{equation}
Y=\widetilde{U}\Lambda_y^{\frac12}V\ .
\label{Ysv}
\end{equation}
Here, $\widetilde{U}$ is an $L\times N$ matrix with $\widetilde{U}^*\widetilde{U}=\eins_N$, whereas $V\in\, U(N)$ is a unitary matrix, cf. \cite{Jonit}. The measure $[dY]$ can be expressed through the singular values $y_1,\ldots,y_N$ in analogy to \eqref{X0decomp} and \eqref{JacobiX}, leading to
\begin{equation}
[dY]\sim\prod_{l=1}^Ny_l^{\kappa}\Delta_N(y_1,\ldots,y_N)^2dy_1\ldots dy_Nd\mu(\widetilde{U})d\mu(V)\,.\nonumber
\end{equation}
The joint probability density of squared singular values of $X_0$ and $Y$ is obtained from the following relation between probability measures:
\begin{equation}
\begin{split}
&\mathcal{P}_2(G,X)[dG][dX]\sim e^{-\Tr\left[\left(VP\right)^*\Lambda_y^{\frac{1}{2}}\widetilde{U}^{*}W\widetilde{U}\Lambda_{y}^{\frac{1}{2}}\left(VP\right)\Lambda_x^{-1}\right]}
e^{-\Tr\left[QU\left(\begin{array}{cc}
                       \mathfrak{U} \Lambda_x \mathfrak{U}^* & O \\
                       O & O
                     \end{array}
\right)U^*\right]}\\
&\times
\prod_{l=1}^Nx_l^{\nu-L}
y_l^{\kappa}\Delta_N(y_1,\ldots,y_N)^2
\Delta_N(x_1,\ldots,x_N)^2dy_1\ldots dy_Ndx_1\ldots dx_N d\mu(U)d\mu(\widetilde{U})d\mu(V)d\mu(\mathfrak{U})d\mu(P).\\
\end{split}
\label{preP2}
\end{equation}
It remains to integrate over all remaining Haar measures, after diagonalising $W=V_1\Lambda_w V_1^*$ and $Q=V_2\Lambda_q V_2^*$ by unitary transformations. Clearly the integrals over $U$ and $\mathfrak{U}$ decouple and lead to the same results as in \eqref{HCIZdeg}.
The remaining integrals are over $P$, $\widetilde{U}$ and $V$, and after using the invariance of the Haar measure to absorb $VP\rightarrow P$ and $V_1^*\widetilde{U}\to\widetilde{U}$ we face the following group integral:
\begin{align}
\mathcal{J}=&\int d\mu(P) d\mu (\widetilde{U})e^{-\Tr\left[P^{*}\Lambda_y^{\frac{1}{2}}\widetilde{U}^{*}\Lambda_w\widetilde{U}\Lambda_{y}^{\frac{1}{2}}P\Lambda_x^{-1}\right]}\,.
\label{JSSM}
\end{align}
This integral was computed by Simon, Moustakas and Marinelli \cite{SMM:2005} using character expansion techniques. However, the final answer is given by a sum over representations that cannot be simplified to a determinantal expression, see \cite[Eq. (59)]{SMM:2005} and details therein. For this reason from now on we will simplify to $W=\alpha \eins_L$ as for the result stated in \eqref{P2jpdf}. In that case the integral \eqref{JSSM} simplifies to the standard HCIZ integral \eqref{HCIZ}, and we obtain
\begin{align}
\int_{U(N)}d\mu(P)e^{-\alpha\Tr\left(P^*\Lambda_yP\Lambda_x^{-1}\right)}
=&\const \frac{\det\left[e^{-\alpha\frac{y_j}{x_i}}\right]_{i,j=1}^N}{\Delta_{N}\left(y_1,\ldots,y_N\right)\Delta_{N}\left(x_1,\ldots,x_N\right)}
\prod_{j=1}^N x_j^{N-1},
\label{HCIZinv}
\end{align}
after using $\Delta_{N}\left(x_1^{-1},\ldots,x_N^{-1}\right) = \mbox{const.}\Delta_{N}\left(x_1,\ldots,x_N\right)/\prod_{j=1}^N x_j^{N-1}$. Integrating \eqref{preP2} over the corresponding coset spaces, using \eqref{HCIZdeg} and \eqref{HCIZinv} we arrive at the statement in \eqref{P2jpdf}. It remains to compute the normalisation constant $Z_2$ which we postpone to the proof of the next corollary.
\end{proof}
From \eqref{P2jpdf} we can easily deduce the joint probability density of the variables $y_j$ alone, together with the corresponding normalisation constant, as summarised in the following
\begin{cor}\label{CorrP2}
The joint probability density of the squared singular values $y_1$, $\ldots$, $y_N$ of the product matrix $Y=GX$, where $G$ and $X$ are distributed according to \eqref{P2} with $W=\alpha\eins_L$, is reading
\begin{equation}
\begin{split}
P_2(y_1,\ldots,y_N)&=\frac{2^{N}N!}{Z_2}\det\left[y_i^{\kappa+j-1}\right]_{i,j=1}^N
\det\left[1,q_i,\ldots,q_i^{\nu-1},\varphi_{i}^{(2)}(y_{1}),\ldots,\varphi_{i}^{(2)}(y_{N})\right]_{i=1}^{M}.
\end{split}
\label{P2yonly}
\end{equation}
Here, we have introduced the following notation
\begin{equation}
\label{varphi2}
\varphi_i^{(2)}(y)=\left(\frac{q_{i}}{\alpha y}\right)^{\frac{\kappa}{2}}K_{\kappa}\left(2\sqrt{\alpha q_{i} y}\right)\,,
\end{equation}
and the normalising constant $Z_2$ is defined as in \eqref{ZGX}.
\end{cor}
\begin{proof}
We make use of the generalised Andr\'eief formula \eqref{gAndreief} by integrating \eqref{P2jpdf} over the $x_{j}$:
\begin{eqnarray}
&\prod_{j=1}^N\int_0^\infty dx_j &P_2(x_1,\ldots,x_N;y_1,\ldots,y_N)=
\frac{N!}{Z_2}\prod_{j=1}^N y_{j}^{\kappa}\ \Delta_N(y_1,\ldots,y_N)\nonumber\\
&&\times\det\left[1,q_{i},\ldots,q_{i}^{\nu-1},\int_{0}^{\infty}dxx^{-\kappa-1}e^{-\frac{\alpha y_{1}}{x}-q_{i}x},\ldots,\int_{0}^{\infty}dxx^{-\kappa-1}e^{-\frac{\alpha y_{N}}{x}-q_{i}x}\right]_{i=1}^{M}.\nonumber
\end{eqnarray}
The remaining integrals are obtained using \cite[Eq. 3.471.9]{GradshteynRyzhik} and the identity $K_{-\kappa}(x)=K_\kappa(x)$ for the modified Bessel function of the second kind,
\begin{equation}
\int_0^{\infty}dxx^{-\kappa-1}e^{-\frac{a}{x}-bx}=2\left(\frac{a}{b}\right)^{-\frac{\kappa}{2}}K_{\kappa}\left(2\sqrt{ab}\right)\,,\quad\text{for }\,\Re(a)>0\,,\quad\Re(b)>0 \,.
\label{K-int}
\end{equation}
This yields \eqref{P2yonly} together with \eqref{varphi2}, after taking out factors of $2$ of the determinant and the factors $y_i^\kappa$ into the Vandermonde determinant. To finally compute the normalisation constant $Z_2$ we have to apply once again the generalised Andr\'eief formula \eqref{gAndreief} to \eqref{P2yonly}. In there we identify
\begin{equation}
\label{Psi2}
\psi_j^{(2)}(y)=y^{\kappa+j-1}\ .
\end{equation}
The integral that remains to be evaluated is thus
\begin{equation}
\label{I2def}
I_{i,j}^{(2)}=\int_{0}^{\infty}dyy^{\kappa+j-1}\left(\frac{q_{i}}{\alpha y}\right)^{\frac{\kappa}{2}}K_{\kappa}\left(2\sqrt{\alpha q_{i} y}\right)
=\frac{\Gamma(\kappa+j)\Gamma(j)}{2\alpha^{\kappa+j}}q_{i}^{-j}\ , \ \ \mbox{for}\ \ i=1,\ldots,M,\ \ j=1,\ldots,N\,,
\end{equation}
which we obtain by using the formula \cite[Eq. 6.561.16]{GradshteynRyzhik}. We thus arrive at
\begin{equation}
Z_2=(N!)^2\left(\prod_{l=1}^{N}\Gamma(\kappa+l)\alpha^{-(\kappa+l)}\Gamma(l)\right)
\det\left[1,q_{i},\ldots,q_{i}^{\nu-1},q_{i}^{-1},\ldots,q_{i}^{-N}\right]_{i=1}^{M},
\end{equation}
which is equivalent to \eqref{ZGX}, after taking out factors of $q_i^{-N}$ and rearranging columns.
\end{proof}
\subsection{Joint probability density of the product of two correlated coupled matrices}\label{JPDcoupled}
Following the derivations from the previous two subsections we are now in the position to prove Theorem \ref{TheoremMainDensity}, the joint probability density of the squared singular values of the random matrix $X$ and of the product matrix $Y=GX$, distributed according to \eqref{JointDensityGX}:
\begin{equation}
\mathcal{P}(G,X)=\,c
\exp\left[-\alpha\Tr\left(GG^*\right)
+\Tr\left(\Omega GX+X^*G^*\Omega^*\right)-\Tr\left(QXX^*\right)\right].
\label{JointDensityGXb}
\end{equation}
Here, $X$ and $G$ are random, $Q$ is fixed, as given before in Subsection \ref{JPD2indep}, the second correlation matrix is $W=\alpha \eins_L$, with $\alpha>0$, and the fixed matrix $\Omega$ that provides the coupling is of size $N\times L$, with squared singular values $\delta_1,\ldots,\delta_N\geq 0$. The normalisation is $c=\pi^{-M(L+N)}\prod_{i=1}^M\prod_{j=1}^N(\alpha q_i-\delta_j)$.
\begin{proof}[Proof of Theorem \ref{TheoremMainDensity} and Corollary \ref{Mainjpdf}]
Following the parametrisations \eqref{Xdecomp}, \eqref{Gsplit} and the change of variables \eqref{changeGY} of the previous subsections, as well as the singular value decompositions \eqref{svdecomp} and \eqref{Ysv}, we immediately obtain that the probability measure from above is proportional to
\begin{align}
&\mathcal{P}(G,X)[dG][dX]\sim e^{-\alpha\Tr\left[\left(VP\right)^*\Lambda_y\left(VP\right)\Lambda_x^{-1}\right]}
e^{\Tr\left[\Omega\widetilde{U}\Lambda_y^{\frac{1}{2}}V +V^*\Lambda_y^{\frac{1}{2}}\widetilde{U}^*\Omega^*\right]}
e^{-\Tr\left[QU\left(\begin{array}{cc}
                       \mathfrak{U} \Lambda_x \mathfrak{U}^* & O \\
                       O & O
                     \end{array}
\right)U^*\right]}
\nonumber\\
&\times\prod_{j=1}^Nx_j^{\nu-L}
y_j^{\kappa}
\Delta_N(y_1,\ldots,y_N)^2
\Delta_N(x_1,\ldots,x_N)^2dx_1\ldots dx_Ndy_1\ldots dy_N
d\mu(U)d\mu(\widetilde{U})d\mu(V)d\mu(\mathfrak{U})d\mu(P).
\nonumber\\
\label{prePDF}
\end{align}
The unitary integrals over the third exponential factor in the first line obviously decouple. We can decouple also the first and second exponential factor by exploiting the invariance of the Haar measure $d\mu(P)$ under $VP\to P$. Furthermore, in the same way we can absorb the fixed unitary matrices from the diagonalisation of $Q=\tilde{V}\Lambda_q\tilde{V}^*$ and the singular value decomposition $\Omega=V_1\Lambda_\delta^{\frac12}V_2^*$, respectively. The new group integral over $\widetilde{U}$ and $V$ that we encounter compared to the previous two subsections, and that is due to the coupling  matrix $\Omega$, is called Berezin-Karpelevich integral \cite{BerezinKarpelevich}
\begin{equation}\label{IZExtended3}
\int_{\widetilde{U}^*\widetilde{U}=\eins_N}d\mu\left(\widetilde{U}\right)\int_{U(N)}d\mu(V) e^{\Tr\left(\Lambda_\delta^{\frac12}\widetilde{U}\Lambda_y^{\frac{1}{2}}V +V^*\Lambda_y^{\frac{1}{2}}\widetilde{U}\Lambda_\delta^{\frac12}\right)}
=\const\frac{\det\left[I_{\kappa}\left(2\sqrt{\delta_{k}y_{j}}\right)\right]_{j,k=1}^{N}\prod_{k=1}^{N}y_{k}^{-\frac{\kappa}{2}}}{\Delta_{N}\left(y_1,\ldots,y_N\right)\Delta_{N}\left(\delta_1,\ldots,\delta_N\right)},
\end{equation}
where the constant does not depend on $y_1$, $\ldots$, $y_N$. This integral is an analogue of the Harish-Chandra--Itzykson--Zuber integral. Such integrals were studied in Guhr and Wettig \cite{GuhrWettig}, and Jackson, \c{S}ener and Verbaarschot \cite{JacksonSenerVerbaarschot}. In Liu \cite{DZL} the same integral appears in the context of coupling uncorrelated Gaussian random matrices, see \cite[Eq. (2.16)]{DZL}. For a similar integral we refer to \cite[Proposition 11.6.2]{ForresterLogGases}. Integrating over the coset spaces in \eqref{prePDF}, we obtain from the HCIZ integrals \eqref{HCIZdeg} and \eqref{HCIZinv} together with this integral \eqref{IZExtended3} the following result for the joint probability density
\begin{eqnarray}
P(x_1,\ldots,x_N;y_1,\ldots,y_N) &=&\const
\prod_{j=1}^{N}y_{j}^{\frac{\kappa}{2}}
\det\left[I_{\kappa}\left(2\sqrt{\delta_{l}y_{j}}\right)\right]_{j,l=1}^{N}
\prod_{j=1}^N x_j^{-\kappa-1}
\det\left[e^{-\alpha\frac{y_j}{x_i}}\right]_{i,j=1}^N
\nonumber\\
&&\times
{\det}\left[1,q_i,\ldots, q_i^{\nu-1},e^{-q_ix_1},\ldots,e^{-q_ix_N}\right]_{i=1}^M\,.
\label{preP}
\end{eqnarray}
This is equivalent to Theorem \ref{TheoremMainDensity}, recalling \eqref{nukappadef}, up to the normalisation constant $Z$ that remains to be determined. We will combine its calculation with the proof of Corollary \ref{Mainjpdf}. For this purpose we apply the generalised Andr\'eief formula \eqref{gAndreief} twice to the joint probability distribution \eqref{preP}: First, integrating over the $x_j$ we obtain
\begin{eqnarray}
&\prod_{j=1}^N\int_0^\infty dx_j &P(x_1,\ldots,x_N;y_1,\ldots,y_N)=
\frac{N!\ 2^N}{Z}\det\left[y_{j}^{\frac{\kappa}{2}}I_{\kappa}\left(2\sqrt{\delta_{l}y_{j}}\right)\right]_{j,l=1}^{N}
\nonumber\\
&&\times\det\left[1,q_{i},\ldots,q_{i}^{\nu-1},
\left(\frac{q_i}{\alpha y_1}\right)^{\frac{\kappa}{2}}K_\kappa(2\sqrt{\alpha q_i y_1}),\ldots,
\left(\frac{q_i}{\alpha y_N}\right)^{\frac{\kappa}{2}}K_\kappa(2\sqrt{\alpha q_i y_N})
\right]_{i=1}^{M},
\label{intP}
\end{eqnarray}
after using \eqref{K-int} and properties of the determinant. This is the statement \eqref{biensemble} in Corollary \ref{Mainjpdf}, together with the identification \eqref{varphi,psi} that we repeat here:
\begin{equation}
\psi_i(y)=y^{\frac{\kappa}{2}}I_{\kappa}\left(2\sqrt{\delta_i y}\right)
\,,\quad \varphi_i(y)=\left(\frac{q_i}{\alpha y}\right)^{\frac{\kappa}{2}}K_{\kappa}\left(2\sqrt{\alpha q_i y}\right)\,.
\label{varpsi}
\end{equation}
For the determination of $Z$ given by the second intergration of \eqref{preP}, this time over the $y_j$, we have
\begin{align}\label{preZ}
Z=(N!)^{2}2^N{\det}\left[1,q_{i},\ldots,q_{i}^{\nu-1},I_{i,1},\ldots,I_{i,N}\right]_{i=1}^M\,.
\end{align}
We are left with the following integral
\begin{eqnarray}
\label{Idef}
I_{i,j}=
\int_{0}^{\infty} dy \left(\frac{q_i}{\alpha y}\right)^{\frac{\kappa}{2}}K_{\kappa}\left(2\sqrt{\alpha q_i y}\right)    y^{\frac{\kappa}{2}} I_{\kappa}(2\sqrt{\delta_{j}y})
=\frac{\delta_j^{\frac{\kappa}{2}}}{2\alpha^\kappa}\ \frac{1}{(\alpha q_i-\delta_j)}\ ,
\end{eqnarray}
for $i=1,\ldots,M$, and $j=1,\ldots,N$, which is obtained using \cite[Eq. 6.576.7]{GradshteynRyzhik}. Its convergence follows from \eqref{aqd}. Moreover, the determinant resulting from (\ref{preZ}) can be identified with the degenerate Cauchy determinant from \cite{Basor:1994}, cf. \eqref{genCauchy} where it was applied before. It reads
\begin{align}\label{DegCauchy}
{\det}\left[1,q_{i},\ldots,q_{i}^{\nu-1},\frac{1}{\alpha q_{i}-\delta_{1}},\ldots,\frac{1}{\alpha q_{i}-\delta_{N}}\right]_{i=1}^M
=(-\alpha)^{MN-\frac{1}{2}N(N+1)}\frac{\Delta_{M}(q_{1},...,q_{M})\Delta_{N}(\delta_{1},...,\delta_{N})}{\prod_{i=1}^{M}\prod_{j=1}^{N}(\alpha q_{i}-\delta_{j})}\,.
\end{align}
The last three equations together yield the normalisation constant $Z$ in \eqref{Z}.
\end{proof}
\section{Determinantal Point Process, Correlation Kernel and its Contour Integral Representation}\label{detPP}
In this section we will proceed in two steps. First, we will show that all our three ensembles are indeed representing determinantal point processes. Second, we use the inverse Gram matrix to explicitly construct complex contour integral representations for all three kernels in separate subsections.

We begin by recalling that the joint probability densities of all our ensembles \eqref{biensemble}, \eqref{biensemble2} and \eqref{P2yonly} are of the form
\begin{equation}
P\left(y_{1},\ldots,y_{N}\right)=\frac{1}{N!\ {\det}\left[A_{i,j}\right]_{i,j=1}^M}\det\left[\psi_{i}(y_{j})\right]_{i,j=1}^{N}
\det\left[1,q_{i},\ldots,q_{i}^{\nu-1},\varphi_{i}(y_{1}),\ldots,\varphi_{i}(y_{N})\right]_{i=1}^{M}\,.
\label{jpdf-gen}
\end{equation}
Here, the Gram type matrix $A$ of size $M\times M$ is defined as
\begin{equation}\label{MatrixA}
A=\left(\begin{array}{ccccccc}
          1 & q_1 & \ldots & q_1^{\nu-1} & I_{1,1} & \ldots & I_{1,N} \\
          1 & q_2 & \ldots & q_2^{\nu-1} & I_{2,1} & \ldots & I_{2,N} \\
          \vdots &  \vdots& &  \vdots&    \vdots&  &    \vdots\\
          1 & q_M & \ldots & q_M^{\nu-1} & I_{M,1} & \ldots & I_{M,N}
        \end{array}
\right),
\end{equation}
with
\begin{equation}
\label{Iij}
I_{i,j} = \int_0^\infty dy\ \varphi_i(y)\psi_j(y)\ ,\ \ \mbox{for}\ \ i=1,\ldots,M,\ \ j=1,\ldots,N\,.
\end{equation}
Part of showing that the class of joint densities \eqref{jpdf-gen} is determinantal includes to determine the $k$-point correlation functions, defined as
\begin{equation}
\label{k-point}
\rho_{k}\left(y_{1},\ldots,y_{k}\right)=\frac{N!}{(N-k)!}\int_{0}^{\infty}dy_{k+1}\ldots dy_{N}P\left(y_{1},\ldots,y_{N}\right)\,,
\end{equation}
in terms of the kernel of the point process. Note that for $k=N$ there is no integral and the $N$-point function is just $N!$ times the joint probability density \eqref{jpdf-gen} itself, cf. \eqref{ppKernel}. Our strategy is to first map the joint probability density \eqref{jpdf-gen} to the standard form of a biorthogonal ensemble of Borodin \cite{Borodin:1998}, having two determinants of equal size instead of \eqref{jpdf-gen}, which shows that this density is indeed determinantal. In a second step we rewrite the resulting kernel in an alternative form, involving directly (part of) the inverse of the Gram type matrix \eqref{MatrixA}, that will be more convenient for later use. This is stated by the following
\begin{prop}\label{PropPP}
Provided that the integrals in \eqref{Iij} of the two sets of functions $\varphi_i$ and $\psi_j$ from the joint probability density \eqref{jpdf-gen} exist, the $k$-point correlation functions are given by
\begin{equation}
\label{MehtaDyson}
\rho_{k}\left(y_{1},\ldots,y_{k}\right)=\det[K_N(y_i,y_j)]_{i,j=1}^k\,.
\end{equation}
The corresponding correlation kernel can be written as
\begin{equation}
\label{gen-kernel}
K_{N}(x,y)=\sum_{i=1}^N\sum_{j=1}^M \psi_{i}(x) C_{i+\nu,j}\varphi_j(y)\,,
\end{equation}
where we denote the inverse Gram type matrix with $C=A^{-1}$. In particular the joint probability density \eqref{jpdf-gen} itself is determinantal, with $k=N$ in \eqref{MehtaDyson}.
\end{prop}
\begin{proof}
It is well known that for a block matrix $D=\left(\begin{matrix}a&c\\b&d\end{matrix}\right)$ with square blocks $a$ and $d$ the determinant of $D$ can be reduced to determinants of smaller size as follows,
\begin{equation}
\det[D]=\det\left[a\right]\det\left[d-b\,a^{-1}c\right]\,,
\label{Schur}
\end{equation}
provided that matrix $a$ is invertible. A similar formula exists for matrix $d$ being invertible, and the matrix $d-b\,a^{-1}c$ is called the Schur complement of matrix $a$ in $D$. Choosing $\left(a\right)_{i,j}=q_{i}^{j-1}$ as the $\nu\times \nu$ matrix from the upper left block of the last determinant in \eqref{jpdf-gen} we obtain
\begin{align}
\det\left[1,q_{i},\ldots,q_{i}^{\nu-1},\varphi_{i}(y_{1}),\ldots,\varphi_{i}(y_{N})\right]_{i=1}^{M}=&\,
\det[a]
\,\det\left[\widetilde{\varphi}_{i}\left(y_{j}\right)\right]_{i,j=1}^{N}\,.
\end{align}
For the Schur complement we obtain
\begin{align}\label{IntermadiateVarphi}
\widetilde{\varphi}_{i}\left(y\right)=\varphi_{i+\nu}\left(y\right)-
\sum\limits_{k,l=1}^{\nu} b_{i,k}
\left(a^{-1}\right)_{k,l}\varphi_{l}\left(y\right)\,,\quad\text{for }\; i=1,\ldots, N\,,
\end{align}
with $b_{i,k}=q_{i+\nu}^{k-1}$. Clearly, for all $q_{i=1,\ldots,\nu}$ being mutually distinct, matrix $a$ is invertible and in fact $\det[a]=\Delta_{\nu}\left(q_{1},\ldots,q_{\nu}\right)$. We can thus apply the result of Borodin \cite{Borodin:1998} for the biorthogonal ensemble obtained from  (\ref{jpdf-gen}),
\begin{equation}
P\left(y_{1},\ldots,y_{N}\right)=\frac{\Delta_{\nu}\left(q_{1},\ldots,q_{\nu}\right)}{N!\ {\det}\left[A_{i,j}\right]_{i,j=1}^M}\det\left[\psi_{i}(y_{j})\right]_{i,j=1}^{N}
\det\left[\widetilde{\varphi}_{i}\left(y_{j}\right)\right]_{i,j=1}^{N}\,,
\label{jpdf-gen-det}
\end{equation}
and conclude that it is indeed a determinantal point process, with its correlation kernel given by
\begin{align}\label{NNKernel}
K_{N}\left(x,y\right)=\sum\limits_{i,j=1}^{N}\psi_{i}(x)\left(g^{-1}\right)_{i,j}\widetilde{\varphi}_{j}(y)\,,\quad\text{with }\;g_{i,j}=\int_{0}^{\infty}dy\,\widetilde{\varphi}_{i}(y)\psi_{j}(y)\,.
\end{align}

It remains to show \eqref{gen-kernel}. For that we insert \eqref{IntermadiateVarphi} into our kernel \eqref{NNKernel} to obtain
\begin{equation}\label{NNKernel2}
K_{N}\left(x,y\right)=\sum\limits_{i=1}^{N}\psi_{i}(x)\left[
\sum\limits_{j=1}^{N}\left(g^{-1}\right)_{i,j}\varphi_{j+\nu}\left(y\right)
-\sum\limits_{l=1}^{\nu}\left(
\sum\limits_{j=1}^{N}\sum\limits_{k=1}^{\nu}
\left(g^{-1}\right)_{i,j}b_{j,k}\left(a^{-1}\right)_{k,l}
\right)\varphi_{l}\left(y\right)
\right]\,.
\end{equation}
Considering now matrix $A$ from \eqref{MatrixA} as a block matrix, $A=\left(\begin{matrix}a&J\\b&I\end{matrix}\right)$, with matrices $a$ and $b$ as defined before, we immediately realise that
\begin{align}
\det\left[A\right]=\det[a]\det[I_{i+\nu,j}-(b(a^{-1})J)_{i,j}]=
\det[a]\det\left[g_{i,j}\right]_{i,j=1}^{N}\,
\end{align}
holds for the corresponding Schur complement. By making use of this block decomposition of $A$, it is well known that its inverse, $C=A^{-1}$, can be written in the following block form, c.f. \cite[Section 3.1]{VP}:
\begin{equation}
C=\left(\begin{matrix}a^{-1}+a^{-1}Jg^{-1}ba^{-1}&-a^{-1}Jg^{-1}\\-g^{-1}ba^{-1}&g^{-1}\end{matrix}\right).
\label{Cexplicit}
\end{equation}
In particular its two lower blocks are given by
\begin{align}\label{LowerPartC}
C_{i+\nu,j}=\left\{\begin{array}{cl}-\sum\limits_{l=1}^{N}\sum\limits_{k=1}^{\nu}\left(g^{-1}\right)_{i,l}b_{l,k}
\left(a^{-1}\right)_{k,j}\,,&\text{for }\; j=1,\ldots,\nu\,,\\
\left(g^{-1}\right)_{i,j-\nu} \,,&\text{for }\; j=1+\nu,\ldots,N+\nu
\end{array}\right.
\end{align}
where $i=1,\ldots, N$\,. Together with \eqref{NNKernel2} this yields \eqref{gen-kernel}.
\end{proof}
\begin{rem}
As an alternative to the formulation of the kernel \eqref{NNKernel} in terms of  the inverse Gram matrix, in \cite{Borodin:1998} the two sets of functions constituing the joint density \eqref{jpdf-gen-det} can also be orthogonalised. For the example from the last subsection, \eqref{varpsi}, this seems to be challenging, as (for $\nu=0$) in each determinant these functions differ only by the parameters in the arguments. A biorthogonalisation can still be performed, see \cite{ACK} for a similar example.
\end{rem}
\begin{rem}
There exist alternative proofs of Proposition \ref{PropPP} without applying \cite{Borodin:1998}. While in Appendix \ref{B} we use simple ideas from functional analysis, we present here a short calculation applying the extended Andr\'eief formula \eqref{genAnd} \cite{Kieburg:2010}. Choosing $l=k+\nu$ and $N\to(N-k)$ in \eqref{genAnd}, we can directly perform the  integration of \eqref{jpdf-gen} over $(N-k)$ variables as prescribed in \eqref{k-point}:
\begin{eqnarray}
\rho_{k}(y_{1},\ldots,y_{k})&=&\frac{(-1)^{k(k+\nu)}}{{\det}[A]}\det\left[\begin{matrix}
O_{k\times \nu}&O_{k\times k}&\left.\psi_{i}(y_{j})\right|_{j=1,\ldots,k}^{i=1,\ldots, N}\\
&&\\
q_{j}^{i-1}|_{j=1,\ldots,M}^{i=1,\ldots,\nu}&\varphi_{j}(y_{i})|_{j=1,\ldots,M}^{i=1,\ldots,k}&\int_{0}^{\infty}dy\varphi_{j}(y)\psi_{i}(y)|_{j=1,\ldots,M}^{i=1,\ldots,N}\end{matrix}\right]\,
\nonumber\\
&=&\frac{(-1)^{k^2}}{{\det}[A]}\det\left[\begin{matrix}
O_{k\times k}&O_{k\times \nu}&\left.\psi_{i}(y_{j})\right|_{j=1,\ldots,k}^{i=1,\ldots, N}\\
&&\\
\left.\varphi_{j}(y_{i})\right|_{j=1,\ldots,M}^{i=1,\ldots,k}&q_{j}^{i-1}|_{j=1,\ldots,M}^{i=1,\ldots,\nu}&I_{j,i}|_{j=1,\ldots,M}^{i=1,\ldots,N}\end{matrix}\right]\,.
\end{eqnarray}
Here, the functions $\psi_i(y)$ and $\varphi_j(y)$ of unintegrated variables $y_1,\ldots,y_k$ are corresponding to the matrices $R$ and $S$ in \eqref{genAnd}, respectively. In the second step we have simply interchanged rows, such that the matrix $A$ from \eqref{MatrixA} is formed by the two lower right blocks. Using the equivalent formula to \eqref{Schur} for invertible $d=A$ this time, $\det[D]=\det\left[d\right]\det\left[a-cd^{-1}b\right]$, we can choose $a=O_{k\times k}$ here. After taking out all minus signs of the determinant and using that part of matrix $c$ is $O_{k\times\nu}$, we arrive at the statement of Proposition \ref{PropPP}:
\begin{align}
\rho_{k}\left(y_{1},\ldots,y_{k}\right)=
\det\left[\sum_{i=1+\nu}^{M}\sum_{j=1}^{M}\psi_{i-\nu}(y_{n})C_{i,j}\varphi_{j}(y_{m})\right]_{n,m=1}^{k}.
\end{align}
\end{rem}
\subsection{Kernel of the generalised Wishart ensemble}\label{kernelgW}
We begin by deriving an explicit form of the kernel of the generalised Wishart ensemble with joint probability density \eqref{biensemble2}. In the simplest case, when $M=N$ ($\nu=0$), its Gram type matrix \eqref{MatrixA} reads $A_{i,j}=I^{(1)}_{i,j}=(q_i+\sigma_j)^{-1}$, from \eqref{I1def}. For its inversion we use the following result of \cite[Lemma 3.1]{Borodin:1998} (cf. \cite{Schechter} for an earlier work)
\begin{lem}[Borodin]\label{LemmaCauchyInverse}
The inverse $C_{i,j}$ of matrix $A_{i,j}=(q_i+\sigma_j)^{-1}$ is given by
\begin{align}\label{Inverse}
C_{i,j}=\frac{1}{(q_j+\sigma_i)}\frac{\prod_{l=1}^N\left(q_l+\sigma_i\right)\left(q_j+\sigma_l\right)}{
{\prod_{k=1;k\neq i}^N}(\sigma_i-\sigma_k){\prod_{l=1;l\neq j}^N}(q_j-q_l)}
\,.
\end{align}
\end{lem}
From \eqref{gen-kernel} together with \eqref{psiphi1} this explicitly determines the kernel of the generalised Wishart ensemble for $\nu=0$:
\begin{eqnarray}
\label{kern1M=N}
K_N^{(1)}(x,y)&=& \sum_{i,j=1}^N  \frac{e^{-\sigma_i x-q_j y}}{(q_j+\sigma_i)}\frac{\prod_{l=1}^N\left(q_l+\sigma_i\right)\left(q_j+\sigma_l\right)}{
{\prod_{k=1;k\neq i}^N}(\sigma_i-\sigma_k){\prod_{l=1;l\neq j}^N}(q_j-q_l)}\\
&=&\oint_{\gamma_{\sigma}}\frac{d\eta}{2\pi i}
\oint_{\gamma_{q}} \frac{d\zeta }{2\pi i}
\frac{e^{x\eta-y\zeta}}{\eta-\zeta}
\prod_{l=1}^N\frac{\zeta+\sigma_l}{\eta+\sigma_l}
\prod_{l=1}^N\frac{\eta-  q_l}{\zeta- q_l}\ .\nonumber
\end{eqnarray}
In the second step we have used the Residue Theorem to express the double sum as a double contour integral. The contours are defined such that the closed contour $\gamma_\sigma$ includes the poles at $-\sigma_l$, $l=1,\ldots,N$, running in counter-clockwise direction, and likewise $\gamma_q$ includes the poles at $q_l$, $l=1,\ldots,N$ in counter-clockwise direction, such that the two contours do not intersect. Because of $q_i+\sigma_j>0$, $\forall i,j$, this is always possible. For different choices of integration contours see Figure \ref{gammaqs1} below. Note that the form of the kernel \eqref{kern1M=N} valid for $M=N$ can be found already in \cite{BP}, see also \cite{DF06} for the Multiple Laguerre kernel. It is very suggestive to expect that a similar form holds also for $M>N$, which is our main result of this subsection as stated below.
\begin{thm}\label{Thmker1} The correlation kernel $K_{N}^{(1)}(x,y)$ of the generalised Wishart ensemble \eqref{corrWishart} permits the following  double contour integral representation
\begin{equation}
K_N^{(1)}(x,y)=
\oint_{\gamma_{\sigma}}\frac{d\eta}{2\pi i}
\oint_{\gamma_{q}} \frac{d\zeta }{2\pi i}
\frac{e^{x\eta-y\zeta}}{\eta-\zeta}
\prod_{l=1}^N\frac{\zeta+\sigma_l}{\eta+\sigma_l}
\prod_{k=1}^M\frac{\eta-  q_k}{\zeta- q_k}\,,\label{kernelGeneralizedWishart}
\end{equation}
where $\gamma_{\sigma}$ is a closed contour encircling $-\sigma_1,\ldots,-\sigma_N$ counter-clockwise, and $\gamma_{q}$ is a closed contour encircling $q_1,\ldots,q_M$ counter-clockwise, without intersecting $\gamma_{\sigma}$, see Figure \ref{gammaqs1}.
\end{thm}
\begin{proof}
The idea of the proof is to obtain the double contour  integral representation \eqref{kernelGeneralizedWishart} without explicitly computing the inverse matrix $C$ as we did for $M=N$. For that purpose we restate the orthogonality relation $AC=\eins_M$ for the Gram type matrix \eqref{MatrixA}, with $I^{(1)}_{i,j}=(q_j+\sigma_i)^{-1}$ for our ensemble from \eqref{I1def}:
\begin{equation}
\label{AC1}
\sum_{k=1}^{\nu}q_j^{k-1}C_{k,l}+ \sum_{k=1}^N \frac{1}{q_j+\sigma_k}
C_{k+\nu,l}=\delta_{j,l}\,,\quad \mbox{for} \quad 1\leq j,l\leq M\,.
\end{equation}
This leads us to define the following set of $l=1,\ldots,M$ meromorphic functions
\begin{equation}\label{fldef}
f_{l}(z)=\sum_{k=1}^{\nu}z^{k-1}C_{k,l}+ \sum_{k=1}^N   \frac{1}{z+\sigma_k} C_{k+\nu,l}\,.
\end{equation}
They are uniquely determined in the complex plane by specifying all their zeros, poles, and by providing the value of the function at one further point. Namely, without specifying the constant matrix $C$ on the right hand side, the functions $f_l(z)$ satisfy:
\begin{enumerate}
\item due to \eqref{AC1} each function $f_l$ has $M-1$ zeros, $f_l(q_i)=0$ for $i=1,\ldots,l-1,l+1,\ldots,M$\,,
\item because of definition \eqref{fldef}, each function $f_l$ has $N$ poles at $z=-\sigma_k$ for  $k=1,\ldots,N$\,,
\item the condition $f_l(q_l)=1$ from \eqref{AC1} uniquely fixes the remaining constant coefficient,
\end{enumerate}
leading to
\begin{equation}\label{flfinal}
f_{l}(z)=\prod_{i=1,i\neq l}^M\frac{z-  q_i}{q_l- q_i}
\prod_{k=1}^N\frac{q_l+\sigma_k}{z+\sigma_k}.
\end{equation}
The fact that these are all poles and zeros follows from the behaviour at infinity, $\lim_{|z|\to\infty}f_l(z)=\mathcal{O}(z^{\nu-1})$, as required from the definition \eqref{fldef}. The next step is to bring the kernel \eqref{gen-kernel} to a form containing \eqref{fldef}, such that we can apply \eqref{flfinal}, without determining $C$ explicitly. From \eqref{psiphi1} we can rewrite for ensemble \eqref{corrWishart}
\begin{equation}\label{psigw}
\psi^{(1)}_j(x)=e^{-\sigma_jx}=\oint_{\gamma_{\sigma}}
\frac{d\eta}{2\pi i}
 \frac{e^{x\eta}}{\eta+\sigma_j}\,, \quad j=1, \ldots,N .
\end{equation}
Here, $\gamma_{\sigma}$ denotes a closed contour encircling $-\sigma_j$ in counter-clockwise direction. For later we choose $\gamma_\sigma$ to contain already all $\sigma_l$, $l=1,\ldots,N$. Likewise, we may write zero in the form
\begin{equation}\label{0gw}
0= \oint_{\gamma_{\sigma}}
\frac{d\eta}{2\pi i}
\eta^{j-1}  e^{x\eta}\,, \quad j=1,2, \ldots\,,
\end{equation}
which holds trivially for any closed contour, due to the analyticity of the integrand. With these preparations, using the definitions \eqref{psiphi1} we can rewrite the kernel \eqref{gen-kernel} for our ensemble as
\begin{eqnarray}
K_{N}^{(1)}(x,y)&=&\sum_{i=1}^N\sum_{j=1}^M\psi_{i}^{(1)}(x)C_{i+\nu,j}\varphi_j^{(1)}(y)\nonumber\\
&=&\sum_{j=1}^M\left[
\sum_{i=1}^NC_{i+\nu,j}\oint_{\gamma_{\sigma}}\frac{d\eta}{2\pi i} \frac{e^{x\eta}}{\eta+\sigma_i}
+\sum_{i=1}^\nu C_{i,j}\oint_{\gamma_{\sigma}}\frac{d\eta}{2\pi i} \eta^{i-1}  e^{x\eta}
\right]\varphi_j^{(1)}(y)\nonumber\\
&=&\oint_{\gamma_{\sigma}}\frac{d\eta}{2\pi i}\sum_{j=1}^Me^{x\eta}\left[\sum_{i=1}^{\nu}
\eta^{i-1}C_{i,j}+\sum_{i=1}^N\frac{1}{\eta+\sigma_i}C_{i+\nu,j}\right]\varphi_j^{(1)}(y)\nonumber\\
&=& \oint_{\gamma_{\sigma}}\frac{d\eta}{2\pi i}\sum_{j=1}^M e^{x\eta}e^{-q_jy}\left[
\prod_{m=1,m\neq j}^M\frac{\eta-  q_m}{q_j- q_m}
\prod_{k=1}^N\frac{q_j+\sigma_k}{\eta+\sigma_k}
\right].
\label{ker1int}
\end{eqnarray}
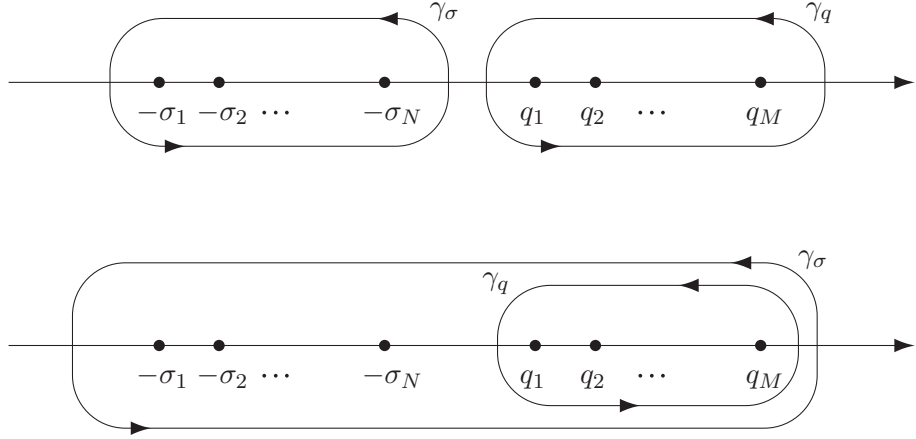
\begin{figure}[h]
\unitlength1cm
\begin{picture}(12.2,7.2)
\thinlines
\put(0,5.5){\line(1,0){12}}
\linethickness{0.4mm}
\put(11.95,5.5){\vector(1,0){0.1}}
\thinlines
\thinlines
\put(1.7,5){$-\sigma_{1}$}
\put(2,5.5){\circle*{0.15}}
\put(2.5,5){$-\sigma_{2}$}
\put(2.8,5.5){\circle*{0.15}}
\multiput(3.4,5.1)(0.15,0){3}{\circle*{0.05}}
\put(4.7,5){$-\sigma_{N}$}
\put(5,5.5){\circle*{0.15}}
\thinlines
\put(6.8,5){$q_{1}$}
\put(7,5.5){\circle*{0.15}}
\put(7.6,5){$q_{2}$}
\put(7.8,5.5){\circle*{0.15}}
\multiput(8.4,5.1)(0.15,0){3}{\circle*{0.05}}
\put(9.8,5){$ q_{M}$}
\put(10,5.5){\circle*{0.15}}
\put(5.6,6.4){$\gamma_{\sigma}$}
\put(10.6,6.4){$\gamma_{q}$}
\put(3.6,5.5){\oval(4.5,1.7)}
\put(8.6,5.5){\oval(4.5,1.7)}
\linethickness{0.4mm}
\multiput(5,6.35)(5,0){2}{\vector(-1,0){0.1}}
\multiput(2.25,4.65)(5,0){2}{\vector(1,0){0.1}}
\thinlines
\put(0,2){\line(1,0){12}}
\linethickness{0.4mm}
\put(11.95,2){\vector(1,0){0.1}}
\thinlines
\put(1.7,1.5){$-\sigma_{1}$}
\put(2,2){\circle*{0.15}}
\put(2.5,1.5){$-\sigma_{2}$}
\put(2.8,2){\circle*{0.15}}
\multiput(3.4,1.6)(0.15,0){3}{\circle*{0.05}}
\put(4.7,1.5){$-\sigma_{N}$}
\put(5,2){\circle*{0.15}}
\thinlines
\put(6.8,1.5){$q_{1}$}
\put(7,2){\circle*{0.15}}
\put(7.6,1.5){$q_{2}$}
\put(7.8,2){\circle*{0.15}}
\multiput(8.4,1.6)(0.15,0){3}{\circle*{0.05}}
\put(9.8,1.5){$ q_{M}$}
\put(10,2){\circle*{0.15}}
\put(6.3,2.8){$\gamma_{q}$}
\put(10.5,3.1){$\gamma_{\sigma}$}
\put(8.5,2){\oval(4.0,1.6)}
\put(5.8,2){\oval(9.9,2.2)}
\linethickness{0.4mm}
\put(9,2.8){\vector(-1,0){0.1}}
\put(9.7,3.1){\vector(-1,0){0.1}}
\put(8.3,1.2){\vector(1,0){0.1}}
\put(1.8,0.9){\vector(1,0){0.1}}
\end{picture}
\caption{Possible choices for the integration contours in Theorem \ref{Thmker1}: Independent non-intersecting coutours (top) and nested non-intersecting coutours (bottom). For simplicity we have ordered the parameters $\sigma_i$ and $q_j$ according to their index. We do not display a third possible choice that is also nested, where the inner contour encircles the $-\sigma_j$ and not the $q_j$.}
\label{gammaqs1}
\end{figure}
In the second step we have taken out the contour integrals, and in the third step we have inserted \eqref{flfinal} and the explicit representation $\varphi_j^{(1)}(y)=e^{-q_jy}$ from \eqref{psiphi1}. Finally a simple application of the Residue Theorem leads to \eqref{kernelGeneralizedWishart}, when choosing $\gamma_q$ as a closed contour that encircles all poles at $q_l$, $l=1,\ldots,M$ in counter-clockwise direction, and that does not intersect $\gamma_\sigma$. In view of the condition \eqref{constraint3}, $q_j+\sigma_i>0$ $\forall i,j$, this is always possible. Two possible choices of such contours are depicted in Figure \ref{gammaqs1}.
\end{proof}
At first sight the second, nested choice in Figure \ref{gammaqs1} bottom may not seem to be useful: It forces us to do the integral over the inner contour $\gamma_q$ {\it first}, before performing the second integral over $\gamma_\sigma$. However, when taking the large-$N$ limit in Section \ref{large-N} later, we will encounter the situation that two or more parameters $-\sigma_i$ and $q_j$ coalesce. In the case of non-nested contours the contours would touch then, which is not allowed. In the nested case there is no such problem, as all $-\sigma_i$ and $q_j$ remain enclosed by $\gamma_\sigma$, and none of the $-\sigma_i$ is a pole of the integral over $\gamma_q$. Of course there is a third choice, by letting $\gamma_q$ enclose the contour $\gamma_\sigma$. Then the integral over the inner contour $\gamma_\sigma$ has to be done first.

The double contour integral in \eqref{kernelGeneralizedWishart} can be factorised, at the expense of a further real integral. For this to be possible we have to choose the contours to be non-nested as in Figure \ref{gammaqs1} top, and thus the order of integration to be independent. In that case the non-intersection condition of the contours $\gamma_q$ and $\gamma_\sigma$ and the fact that $q_j+\sigma_i>0$, $\forall i,j$, implies that $\Re(\zeta-\eta)>0$. Therefore, we can rewrite the term coupling the two contour integrals as
\begin{equation}
\label{factor-int}
\frac{1}{\eta-\zeta}=-\int_0^{1}du\,u^{\zeta-\eta-1}\,,\quad\text{for }\quad \Re(\zeta-\eta)>0\,.
\end{equation}
This immediately leads to the following
\begin{cor} The kernel $K_{N}^{(1)}(x,y)$ given by (\ref{kernelGeneralizedWishart}), with integration contours chosen as in Figure \ref{gammaqs1} top, can be written as
\begin{equation}
K_{N}^{(1)}(x,y)=-\int_{0}^{1}\frac{du}{u}F_1^{(1)}(x;u)F_2^{(1)}(y;u)\,,
\end{equation}
where the functions $F_1^{(1)}(x;u)$ and $F_2^{(1)}(y;u)$ are defined by the formulae
\begin{equation}
F_1^{(1)}(x;u)=\oint_{\gamma_{\sigma}}\frac{d\eta}{2\pi i}
u^{-\eta} e^{x\eta}\
\frac{\prod_{l=1}^M\left(\eta- q_l\right)}{\prod_{l=1}^N\left(\eta+\sigma_l\right)}\,,\quad
F_2^{(1)}(y;u)=\oint_{\gamma_{q}}\frac{d\zeta}{2\pi i}
u^{\zeta}e^{-y\zeta}\
\frac{\prod_{l=1}^N\left(\zeta+\sigma_l\right)}{\prod_{l=1}^M\left(\zeta- q_l\right)}\,.
\end{equation}
\end{cor}
\subsection{Kernel of the product of two correlated coupled matrices}\label{kernelcoupled}
Next we immediately turn to the ensemble \eqref{JointDensityGX} of two correlated coupled random matrices. The reason is that the Gram type matrix is very similar to the previous subsection, $I_{i,j}=\delta_j^{\frac{\kappa}{2}}(2\alpha^\kappa(\alpha q_i-\delta_j))^{-1}$ from \eqref{Idef}, making it straightforward to generalise the results from the previous subsection. For $N=M$ we can apply Lemma \ref{LemmaCauchyInverse}, replacing $q_j\to\alpha q_j$ and $\sigma_i\to-\delta_i$, and, apart from a trivial factor, directly read off the inverse matrix $C_{i,j}$ of $A_{i,j}=I_{i,j}$:
\begin{equation}
C_{i,j}=\frac{2\alpha^\kappa}{\delta_j^{\frac{\kappa}{2}}}\frac{1}{(\alpha q_j-\delta_i)}\frac{\prod_{l=1}^N\left(\alpha q_l-\delta_i\right)\left(\alpha q_j-\delta_l\right)}{
{\prod_{k=1;k\neq i}^N}(\delta_k-\delta_i){\prod_{l=1;l\neq j}^N}(\alpha q_j-\alpha q_l)}
\,.
\end{equation}
Consequently, inserting this expression together with \eqref{varphi,psi} into \eqref{gen-kernel} we obtain the following explicit expression for the kernel at $\nu=0$:
\begin{eqnarray}
K_{N}(x,y)&=&
\left(\frac{x}{y}\right)^{\frac{\kappa}{2}}
\sum_{i=1}^N\sum_{j=1}^M I_\kappa(2\sqrt{\delta_i x})K_\kappa(2\sqrt{\alpha q_j y})
\frac{2\left(\frac{\alpha q_j}{\delta_j }\right)^{\frac{\kappa}{2}}}{(\alpha q_j-\delta_i)}\frac{\prod_{l=1}^N\left(\alpha q_l-\delta_i\right)\left(\alpha q_j-\delta_l\right)}{
{\prod_{k=1;k\neq i}^N}(\delta_k-\delta_i){\prod_{l=1;l\neq j}^N}(\alpha q_j-\alpha q_l)}
\nonumber\\
&=&
\left(\frac{x}{y}\right)^{\frac{\kappa}{2}}2
\oint_{\gamma_{\delta}}\frac{d\eta}{2\pi i}\oint_{\gamma_{q}}\frac{d\zeta}{2\pi i}\left(\frac{\zeta}{\eta}\right)^{\frac{\kappa}{2}}
\frac{I_{\kappa}\left(2\sqrt{\eta x}\right)K_{\kappa}\left(2\sqrt{\zeta y}\right)}{\eta-\zeta}\frac{\prod_{l=1}^N\left(\zeta-\delta_l\right)\left(\eta-\alpha q_l\right)}{\prod_{l=1}^N\left(\eta-\delta_l\right)\left(\zeta-\alpha q_l\right)}\,.
\label{kernelN=M}
\end{eqnarray}
Here, the contours are defined analogously to Theorem \ref{Thmker1}, with the difference that only the {\it two} choices are possible that are depicted in Fig \ref{gammaqs1}. After the replacement $-\sigma_l\to\delta_l$ and $q_l\to \alpha q_l$, the closed contour $\gamma_{\delta}$ is encircling $\delta_1,\ldots,\delta_N\geq 0$ counter-clockwise, including or excluding all $q_l$. Note that $\eta^{-\frac{\kappa}{2}}I_{\kappa}\left(2\sqrt{\eta x}\right)$ does not have a branch cut in $\eta$, cf. \eqref{BesselI}. The closed contour $\gamma_{q}$ is encircling $\alpha q_1,\ldots,\alpha q_N>0$ counter-clockwise and, in contrast, excludes the origin,
because $K_{\kappa}\left(2\sqrt{\zeta y}\right)$ has a logarithmic singularity there. Thus the contour $\gamma_q$ may {\it not} include all $\delta_l\geq 0$. The requirement of non-intersecting contours is always possible, due to the condition \eqref{aqd}
that $\alpha q_i-\delta_j>0$ $\forall i,j$. The last equality in \eqref{kernelN=M} is easy to see with the help of the Residue Theorem, where the order of integration may depend on the nesting of the contours. The prefactor $(x/y)^{\frac{\kappa}{2}}$ before the two integrals can be dropped as it cancels out in the determinant \eqref{MehtaDyson}, leading to an equivalent kernel (see also the remark after Theorem \ref{MainKernelContourIntegral}). Let us present the proof of this theorem now for general $M\geq N$.
\begin{proof}[Proof of Theorem \ref{MainKernelContourIntegral}]
In view of the Gram matrix $I_{i,j}$ \eqref{Idef}, it is advantageous for $\nu>0$ to slightly modify the Gram type matrix \eqref{MatrixA} by including the appropriate powers of $\alpha$:
\begin{equation}\label{MatrixAtilde}
\tilde{A}=\left(\begin{array}{ccccccc}
          1 & \alpha q_1 & \ldots & (\alpha q_1)^{\nu-1} & I_{1,1} & \ldots & I_{1,N} \\
          1 & \alpha q_2 & \ldots & (\alpha q_2)^{\nu-1} & I_{2,1} & \ldots & I_{2,N} \\
          \vdots &  \vdots& &  \vdots&    \vdots&  &    \vdots\\
          1 & \alpha q_M & \ldots & (\alpha q_M)^{\nu-1} & I_{M,1} & \ldots & I_{M,N}
        \end{array}
\right).
\end{equation}
This can be trivially achieved by multiplying numerator and denominator of \eqref{jpdf-gen} by $\alpha^{\nu(\nu-1)/2}$. Its inverse is now denoted by $\tilde{C}$, with $\tilde{A}\tilde{C}=\eins_M$. Following the ideas of the proof of Theorem \ref{Thmker1} from the previous subsection, it is then not difficult to relate the inversion of the corresponding full Gram type matrix \eqref{MatrixAtilde},
\begin{equation}
\label{AC1gen}
\sum_{k=1}^{\nu}(\alpha q_j)^{k-1}\tilde{C}_{k,l}+ \sum_{k=1}^N \frac{\delta_j^{\frac{\kappa}{2}}}{2\alpha^\kappa}\frac{1}{(\alpha q_j-\delta_k)}
\tilde{C}_{k+\nu,l}=\delta_{j,l}\,,\quad \mbox{for} \quad 1\leq j,l\leq M\,,
\end{equation}
to a set of $l=1,\ldots,M$ meromorphic functions
\begin{equation}
f_{l}(z)=\sum_{k=1}^{\nu}z^{k-1}\tilde{C}_{k,l}+ \sum_{k=1}^N
\frac{\delta_j^{\frac{\kappa}{2}}}{2\alpha^\kappa}\frac{1}{(z-\delta_k)}
\tilde{C}_{k+\nu,l}\
=\prod_{i=1,i\neq l}^M\frac{z-  \alpha q_i}{\alpha q_l- \alpha q_i}
\prod_{k=1}^N\frac{\alpha q_l-\delta_k}{z-\delta_k}\,.
\end{equation}
Its zeros at $z=\alpha q_{i\neq l}$, poles at $z=\delta_j$, the condition $f_l(\alpha q_l)=1$ and checking its correct behaviour at infinity completely fixes the right hand side. With only little more thought we can also write the analogue of the conditions \eqref{psigw} and \eqref{0gw} for the respective function $\psi_i$ from \eqref{varphi,psi}:
\begin{equation}
\label{psi-int}
\psi_i(x)=x^{\frac{\kappa}{2}}I_\kappa(2\sqrt{\delta_i x})=
2\oint_{\gamma_\delta}\frac{d\eta}{2\pi i} \frac{1}{\eta-\delta_i}  \left(\frac{\delta_i x}{\eta}\right)^{\frac{\kappa}{2}}\frac12 I_{\kappa}\left(2\sqrt{\eta x}\right)\,,
\quad i=1, \ldots, N\,,
\end{equation}
and, due to \eqref{BesselI}
\begin{equation}\label{0-int}
0= 2\oint_{\gamma_\delta}\frac{d\eta}{2\pi i}
\eta^{j-1} \alpha^\kappa \left(\frac{x}{\eta}\right)^{\frac{\kappa}{2}}I_{\kappa}\left(2\sqrt{\eta x}\right), \quad j=1,2, \ldots\,.
\end{equation}
We can then rewrite the kernel \eqref{gen-kernel} as in the previous subsection:
\begin{eqnarray}
K_{N}(x,y)
&=&\sum_{j=1}^M\left[
\sum_{i=1}^N\tilde{C}_{i+\nu,j}2\oint_{\gamma_\delta}\frac{d\eta}{2\pi i} \frac{1}{\eta-\delta_i}
\left(\frac{\delta_i x}{\eta}\right)^{\frac{\kappa}{2}}\frac12 I_{\kappa}\left(2\sqrt{\eta x}\right)\right.\nonumber\\
&&\left.\quad\quad
+\sum_{i=1}^\nu \tilde{C}_{i,j}
2\oint_{\gamma_\delta}\frac{d\eta}{2\pi i}
\eta^{i-1}\alpha^\kappa  \left(\frac{x}{\eta}\right)^{\frac{\kappa}{2}}I_{\kappa}\left(2\sqrt{ \eta x}\right)
\right]
\varphi_j(y)\nonumber\\
&=&\oint_{\gamma_{\delta}}\frac{d\eta}{2\pi i} 2\alpha^\kappa  \left(\frac{x}{\eta}\right)^{\frac{\kappa}{2}}
\sum_{j=1}^MI_{\kappa}\left(2\sqrt{\eta x}\right)
\left[\sum_{i=1}^{\nu}
\eta^{i-1}\tilde{C}_{i,j}+\sum_{i=1}^N\frac{\delta_j^{\frac{\kappa}{2}}}{2\alpha^\kappa}\frac{1}{(\eta-\delta_i)}\tilde{C}_{i+\nu,j}
\right]
\varphi_j(y)
\nonumber\\
&=&
\left(\frac{x}{y}\right)^{\frac{\kappa}{2}}
2\oint_{\gamma_{\delta}}\frac{d\eta}{2\pi i}
\sum_{j=1}^M I_{\kappa}\left(2\sqrt{\eta x}\right)
\left[
\prod_{m=1,m\neq j}^M\frac{\eta-  \alpha q_m}{\alpha q_j- \alpha q_m}
\prod_{k=1}^N\frac{\alpha q_j-\delta_k}{\eta-\delta_k}
\right]\left(\frac{\alpha q_j}{\eta}\right)^{\frac{\kappa}{2}}
\nonumber\\
&&
\times
K_{\kappa}\left(2\sqrt{\alpha q_j y}\right)
\,.
\label{kerint}
\end{eqnarray}
We have inserted $\varphi_j(y)$ from \eqref{varphi,psi} and in the last step used the Residue Theorem, leading to a kernel equivalent to \eqref{MainKerneloint}. As discussed previously we have two choices for the contours not to intersect, being either nested or separated. In case they are nested the inner integration has to be done first.
\end{proof}
Along the same lines as in the previous subsection we can derive the following equivalent factorised form of the kernel, using the identity \eqref{factor-int}. For this factorised form we have to choose again the contours to be non-nested (cf. Figure \ref{gammaqs1} top), for the integrals to factorise and become independent. Here we also removed the prefactor $(x/y)^{\frac{\kappa}{2}}$ in \eqref{kerint}.
\begin{cor}\label{kernelRealDefiniteIntegral} The kernel $K_{N}(x,y)$ given by Theorem \ref{MainKernelContourIntegral} is equal
to the following kernel
\begin{equation}
K_{N}(x,y)=-\int_{0}^{1}\frac{du}{u}F_1(x;u)F_2(y;u)\,.
\end{equation}
The functions $F_1(x;u)$ and $F_2(y;u)$ are defined by the formulae
\begin{equation}
F_1(x;u)=\oint_{\gamma_{\delta}}
\frac{d\eta}{2\pi i}  u^{-\eta}{\eta^{-\frac{\kappa}{2}}}I_{\kappa}\left(2\sqrt{\eta x}\right)
\frac{\prod_{l=1}^M\left(\eta-\alpha q_l\right)}{\prod_{l=1}^N\left(\eta-\delta_l\right)}\,,
\end{equation}
where $\gamma_\delta$ encloses all $\delta_l$ in a counter-clockwise way, and
\begin{equation}
F_2(y;u)=2\oint_{\gamma_{q}}
\frac{d\zeta}{2\pi i}  u^{\zeta}\zeta^{\frac{\kappa}{2}}K_{\kappa}\left(2\sqrt{\zeta y}\right)
\frac{\prod_{l=1}^N\left(\zeta-\delta_l\right)}{\prod_{l=1}^M\left(\zeta-\alpha q_l\right)}\,.
\end{equation}
The contour $\gamma_q$ encloses all $\alpha q_l$ in a counter-clockwise way, excludes the origin and all $\delta_l$..
\end{cor}
\subsection{Kernel of the product of two independent correlated matrices}\label{kernelindep}
We turn to the kernel of the ensemble \eqref{JointDensityindepGX} of two independent matrices, one of which has correlated entries. Rather than trying to first invert the Gram matrix for $N=M$, we immediately turn to the procedure from the previous two subsections, that directly leads to the following double contour integral representation.
\begin{thm}\label{Thmker2} The correlation kernel $K_{N}^{(2)}(x,y)$ of the ensemble \eqref{JointDensityindepGX} permits the following double contour integral representation
\begin{equation}
K_N^{(2)}(x,y)= \left(\frac{x}{y}\right)^{\frac{\kappa}{2}}
2\oint_{\gamma_0}\frac{d\eta}{2\pi i} \oint_{\gamma_q}\frac{d\zeta}{2\pi i}
\left(\frac{\zeta}{\eta}\right)^{N+\frac{\kappa}{2}}\frac{I_{\kappa}\left(2\sqrt{\eta x}\right)K_{\kappa}\left(2\sqrt{\zeta y}\right)}{\eta-\zeta}
\prod_{l=1}^{M}\frac{\eta-\alpha q_{l}}{\zeta-\alpha q_{l}}\,,
\label{kernel-indep}
\end{equation}
where $\gamma_0$ is a closed contour encircling the origin in counter-clockwise direction, and $\gamma_{q}$ is a closed contour encircling $\alpha q_1,\ldots,\alpha q_M>0$ counter-clockwise, excluding the origin and not intersecting $\gamma_0$.
\end{thm}
\begin{figure}[h]
\unitlength1cm
\begin{picture}(12.2,3.5)
\thinlines
\put(0,2){\line(1,0){12}}
\linethickness{0.4mm}
\put(11.95,2){\vector(1,0){0.1}}
\thinlines
\put(3.6,1.5){$0$}
\linethickness{0.4mm}
\put(3.7,1.85){\line(0,1){0.3}}
\thinlines
\thinlines
\put(6.8,1.5){$\alpha q_{1}$}
\put(7,2){\circle*{0.15}}
\put(7.6,1.5){$\alpha q_{2}$}
\put(7.8,2){\circle*{0.15}}
\multiput(8.4,1.6)(0.15,0){3}{\circle*{0.05}}
\put(9.8,1.5){$\alpha q_{M}$}
\put(10,2){\circle*{0.15}}
\put(6.1,2.8){$\gamma_{q}$}
\put(2.8,3.1){$\gamma_{0}$}
\put(8.7,2){\oval(4.2,1.6)}
\put(7,2){\oval(8,2.2)}
\linethickness{0.4mm}
\put(9.4,2.8){\vector(-1,0){0.1}}
\put(10.2,3.1){\vector(-1,0){0.1}}
\put(7.3,1.2){\vector(1,0){0.1}}
\put(4.3,0.9){\vector(1,0){0.1}}
\end{picture}
\caption{Nested choice of the contours $\gamma_{0}$ and $\gamma_{q}$. For the second choice of non-intersecting contours that is non-nested $\gamma_0$ is including only the origin, but none of the $\alpha q_l$. Because $\gamma_q$ is excluding the origin we don't have a third choice here, where $\gamma_0$ is lying inside $\gamma_q$.}\label{Fig2}
\end{figure}
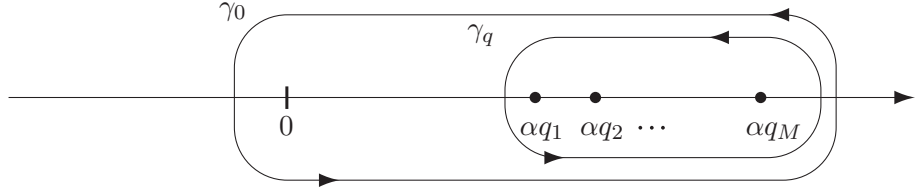
\begin{proof}
In view of the Gram matrix \eqref{I2def}, $I^{(2)}_{i,j}=\Gamma(\kappa+j)\Gamma(j)/(2\alpha^{\kappa+j}q_i^j)$, it is again useful to start with the modified Gram type matrix \eqref{MatrixAtilde} as in the previous subsection. The corresponding equation $\tilde{A}\tilde{C}=\eins_M$ thus reads:
 \begin{equation}
\label{A2gen}
\sum_{k=1}^{\nu}(\alpha q_j)^{k-1}\tilde{C}_{k,l}+ \sum_{m=1}^N \frac{\Gamma(\kappa+m)\Gamma(m)}{2\alpha^{\kappa+m}q_j^m}
\tilde{C}_{m+\nu,l}=\delta_{j,l}\,,\quad \mbox{for} \quad 1\leq j,l\leq M\,.
\end{equation}
Once again this can be used to define the following set of $l=1,\ldots,M$ meromorphic functions
\begin{equation}
f_{l}(z)=\sum_{k=1}^{\nu}z^{k-1}\tilde{C}_{k,l}
+ \sum_{m=1}^N \frac{\Gamma(\kappa+m)\Gamma(m)}{2\alpha^{\kappa}}z^{-m}
\tilde{C}_{m+\nu,l}
=\left(\frac{\alpha q_l}{z}\right)^{N}
\prod_{i=1,i\neq l}^M\frac{z-  \alpha q_i}{\alpha q_l- \alpha q_i}\,.
\end{equation}
These functions are determined by their $M-1$ zeros at $z=\alpha q_{i\neq l}$, the poles of up to order $N$ at $z=0$, the condition $f_l(\alpha q_l)=1$, and by checking its behaviour at infinity. Given that in this ensemble we have \eqref{Psi2}, we can rewrite (cf. \eqref{BesselI})
\begin{equation}
\label{psi2int}
\psi_i^{(2)}(x)=x^{\kappa+i-1}=\frac{1}{(2\pi i)^2}\oint_{\gamma_0}d\eta\eta^{-i}\left( \frac{x}{\eta}\right)^{\frac{\kappa}{2}}I_\kappa(2\sqrt{\eta x}) \Gamma(\kappa+i)\Gamma(i)\ ,\ \ i=1,2,\ldots,
\end{equation}
where $\gamma_0$ is a closed contour encircling the origin counter-clockwise. For the additional condition we can reuse \eqref{0-int} which is true also for $\gamma_0$, due to the analyticity of the integrand. We thus obtain from \eqref{gen-kernel} in our case
\begin{eqnarray}
K_{N}^{(2)}(x,y)
&=&\sum_{j=1}^M\left[
\sum_{i=1}^N\tilde{C}_{i+\nu,j}\oint_{\gamma_0}\frac{d\eta}{2\pi i} \eta^{-i}
\left(\frac{x}{\eta}\right)^{\frac{\kappa}{2}}I_{\kappa}\left(2\sqrt{\eta x}\right)\Gamma(\kappa+i)\Gamma(i)
\right.\nonumber\\
&&\left.\quad\quad
+\sum_{i=1}^\nu \tilde{C}_{i,j}
\oint_{\gamma_0}\frac{d\eta}{2\pi i}
\eta^{i-1}2\alpha^\kappa  \left(\frac{x}{\eta}\right)^{\frac{\kappa}{2}}I_{\kappa}\left(2\sqrt{ \eta x}\right)
\right]
\varphi_j^{(2)}(y)\nonumber\\
&=&\oint_{\gamma_0}\frac{d\eta}{2\pi i} \,2\alpha^\kappa  \left(\frac{x}{\eta}\right)^{\frac{\kappa}{2}}
\sum_{j=1}^M I_{\kappa}\left(2\sqrt{\eta x}\right)
\left[\sum_{i=1}^{\nu}
\eta^{i-1}\tilde{C}_{i,j}+\sum_{i=1}^N\frac{\Gamma(\kappa+i)\Gamma(i)}{2\alpha^\kappa}\eta^{-i}\tilde{C}_{i+\nu,j}
\right]
\varphi_j^{(2)}(y)
\nonumber\\
&=&
\left(\frac{x}{y}\right)^{\frac{\kappa}{2}}
2\oint_{\gamma_0}\frac{d\eta}{2\pi i}
\sum_{j=1}^M
I_{\kappa}\left(2\sqrt{\eta x}\right)
\left[
\frac{(\alpha q_j)^N}{\eta^N}
\prod_{m=1,m\neq j}^M\frac{\eta-  \alpha q_m}{\alpha q_j- \alpha q_m}
\right]
\left(\frac{\alpha q_j}{\eta}\right)^{\frac{\kappa}{2}}
K_{\kappa}\left(2\sqrt{\alpha q_j y}\right)
\,.
\label{ker2int}
\end{eqnarray}
Here, we have inserted $\varphi_j^{(2)}(y)$ from \eqref{varphi2} and used the Residue Theorem, with two possible choices of contours, cf. Figure \ref{Fig2}.
\end{proof}
Note that as a check \eqref{kernel-indep} agrees with the kernel of the coupled ensemble \eqref{MainKerneloint}, when setting all $\delta_{k=1,\ldots,N}=0$ there.

Using again the identity \eqref{factor-int}, from \eqref{kernel-indep} together with the choice of non-nested contours we obtain the following factorised integral representation of a kernel equivalent to that in Theorem \ref{Thmker2}.
\begin{cor}\label{kernelRealDefiniteIntegral2} The kernel $K_{N}^{(2)}(x,y)$ given by Theorem \ref{Thmker2} is equivalent to the following kernel
\begin{equation}
K_{N}^{(2)}(x,y)=-\int_{0}^{1}\frac{du}{u}F_1^{(2)}(x;u)F_2^{(2)}(y;u)\,.
\end{equation}
The functions $F_1^{(2)}(x;u)$ and $F_2^{(2)}(y;u)$ are defined by the formulae
\begin{equation}
F_1^{(2)}(x;u)=\oint_{\gamma_{0}}
\frac{d\eta}{2\pi i}  u^{-\eta}{\eta^{-\frac{\kappa}{2}-N}}I_{\kappa}\left(2\sqrt{\eta x}\right)
{\prod_{l=1}^M\left(\eta-\alpha q_l\right)}\,,
\end{equation}
with $\gamma_0$ encircling the origin counter-clockwise, and
\begin{equation}
F_2^{(2)}(y;u)={2}\oint_{\gamma_q}
\frac{d\zeta}{2\pi i}  u^{\zeta}\zeta^{\frac{\kappa}{2}+N}K_{\kappa}\left(2\sqrt{\zeta y}\right)
\prod_{l=1}^M\frac{1}{\zeta-\alpha q_l}\,,
\end{equation}
with $\gamma_q$ including all $\alpha q_l$ counter-clockwise and excluding the origin.
\end{cor}
This corollary agrees with Corollary \ref{kernelRealDefiniteIntegral} of the coupled ensemble, after setting all $\delta_{k=1,\ldots,N}=0$ therein.
\section{Large-$N$ Limit at the Origin and Integrability}\label{large-N}
In this section we will study the limit of large matrix size $N\to\infty$ at the origin of the spectrum, in all three ensembles
separately. It turns out that the kernel of the generalised Wishart ensemble \eqref{corrWishart} will lead to the generalised Bessel kernel $\mathbb{K}_{\rm III}$ \eqref{kernelsup} in the large-$N$ limit. This will be shown in the first Subsection \ref{Lim-kernelgW}. The kernel of the second ensemble \eqref{JointDensityindepGX} of independent matrices with correlated entries leads to the limiting kernel $\mathbb{K}_{\rm I}$ \eqref{kernelsub}. It generalises the limiting Meijer $G$-kernel \cite{KZ14} obtained for the product of two independent random matrices by adding finite rank perturbations, as will be shown in Subsection \ref{Lim-kernelindep}. In the last Subsection \ref{Lim-kernelcoupled} we will show that the kernel of the product of two correlated coupled random matrices leads to three different limiting kernels, depending on the coupling parameter $\mu=\mu(N)$ as a function of $N$:
The kernel $\mathbb{K}_{\rm III}$ follows in limit (III) $\mu(N) N\to 0$, and kernel $\mathbb{K}_{\rm I}$  in limit (I) $\mu(N) N\to\infty$. A third kernel $\mathbb{K}_{\rm II}$ given in \eqref{kernelcrit} follows in limit (II) when $\mu(N) N\to\tau/4$,
with $\tau>0$, and interpolates between the kernels obtained in limits (I) and (III). It generalises the interpolating kernel of \cite{AStr16b} and of \cite{DZL}, by adding further finite rank perturbations. In that sense all three limiting kernels $\mathbb{K}_{\rm I, II, III}$ are universal as they follow from different ensembles. For all three kernels we provide their corresponding integrable form, in the sense of \cite{IIKS}.
\subsection{Bessel kernel with finite rank perturbations from the generalised Wishart ensemble}\label{Lim-kernelgW}
We begin by recalling the generalised Wishart ensemble \eqref{corrWishart}
\begin{equation}
\mathcal{P}_1(X)=\,c_1
\exp\left[-
\Tr\left(X\Sigma X^*\right)-\Tr\left(QXX^*\right)\right]\ ,
\nonumber
\end{equation}
with eigenvalues $\sigma_1,\ldots,\sigma_N$ of $\Sigma$ and $q_1,\ldots,q_M$ of $Q$, respectively. In the following we will consider finite rank perturbations of the fully degenerate case, $\Sigma=\sigma\eins_N$ and $Q=q\eins_M$, by setting
\begin{equation}
\sigma_{n+1}=\cdots=\sigma_{N}=\sigma\quad\text{and }\;q_{m+1}=\cdots=q_{M}=q\,,
\label{sqdegenerate}
\end{equation}
with $n$ and $m$ independent of $N$. Thus we consider a perturbation around the standard Wishart-Laguerre ensemble  $\mathcal{P}_1(X)=c_1\exp[-(q+\sigma)\Tr(XX^*)]$, with $q+\sigma>0$.

If we want to compare to our most general ensemble \eqref{JointDensityGX} later, e.g. by integrating out random matrix $G$ there, we would have to identify $-\Sigma = \Omega\Omega^*/\alpha$, or $-\sigma_j=\delta_j/\alpha$ for $j=1,\ldots,N$. In Subsection \ref{Lim-kernelcoupled} we will make the parameters $\alpha$ and $\delta_j$ there $\mu$- and thus $N$-dependent, which would lead to identify $q+\sigma=2/(1+\mu)$. In this subsection, however, there is no need to introduce such an extra parameter $\mu=\mu(N)$, as the large-$N$ limit at the origin that we will take here does not depend on it.

Inserting the degeneracy \eqref{sqdegenerate} into the kernel at finite-$N$ \eqref{kernelGeneralizedWishart} from Theorem \ref{Thmker1}, we obtain
\begin{eqnarray}
K_{N}^{(1)}(x,y)&=&\oint_{\gamma_{\sigma}}\frac{d\eta}{2\pi i}\oint_{\gamma_{q}}\frac{d\zeta}{2\pi i}\frac{e^{x\eta-y\zeta}}{\eta-\zeta}\left(\frac{\zeta+\sigma}{\eta+\sigma}\right)^{N-n}\left(\frac{\eta-q}{\zeta-q}\right)^{M-m}\prod_{l=1}^{n}\frac{\zeta+\sigma_{l}}{\eta+\sigma_{l}}\prod_{k=1}^{m}\frac{\eta-q_{k}}{\zeta-
q_{k}}\,.
\nonumber\\
&=&{e^{q(x-y)}}\left(q+\sigma\right)\oint_{\Gamma_{\pi}}\frac{dv}{2\pi i}\oint_{\Gamma_{\theta}}\frac{du}{2\pi i}\frac{e^{(q+\sigma)(yu-xv)}}{u-v}\left(\frac{1-\frac{1}{u}}{1-\frac{1}{v}}\right)^{N-n}\left(\frac{v}{u}\right)^{\nu-m+n}\nonumber\\
&&\times
\prod_{l=1}^{n}\frac{u-\frac{q+\sigma_{l}}{q+\sigma}}{v-\frac{q+\sigma_{l}}{q+\sigma}}\prod_{k=1}^{m}\frac{v-\frac{q-q_{k}}{q+\sigma}}{u-\frac{q-q_{k}}{q+\sigma}}\,.
\label{K1start}
\end{eqnarray}
In the second line we have made the following substitution:
\begin{equation}
\zeta=q-\left(q+\sigma\right)u\quad\text{and }\;\eta=q-\left(q+\sigma\right)v\,.
\label{sub1}
\end{equation}
Starting from the nested contours as in Figure \ref{gammaqs1} bottom, the integration contours resulting from this substitution are given in Figure \ref{gWcontours2}. Due to $q+\sigma>0$ the substitution is non-singular.
\begin{figure}[h]
\unitlength1cm
\begin{picture}(12.2,3.5)
\thinlines
\put(0,2){\line(1,0){12}}
\linethickness{0.4mm}
\put(11.95,2){\vector(1,0){0.1}}
\thinlines
\put(1.9,1.5){$\theta_{m}$}
\put(2,2){\circle*{0.15}}
\multiput(2.5,1.6)(0.15,0){3}{\circle*{0.05}}
\put(3.6,1.5){$0$}
\linethickness{0.4mm}
\put(3.7,1.85){\line(0,1){0.3}}
\thinlines
\put(4.9,1.5){$\theta_{1}$}
\put(5,2){\circle*{0.15}}
\put(6.9,1.5){$\pi_{n}$}
\put(7,2){\circle*{0.15}}
\multiput(7.5,1.6)(0.15,0){3}{\circle*{0.05}}
\put(8.6,1.5){$1$}
\linethickness{0.4mm}
\put(8.7,1.85){\line(0,1){0.3}}
\thinlines
\put(9.9,1.5){$\pi_{1}$}
\put(10,2){\circle*{0.15}}
\put(5.6,2.7){$\Gamma_{\theta}$}
\put(10.8,3.1){$\Gamma_{\pi}$}
\put(3.6,2){\oval(4.2,1.6)}
\put(6,2){\oval(9.6,2.2)}
\linethickness{0.4mm}
\put(4.7,2.8){\vector(-1,0){0.1}}
\put(10,3.1){\vector(-1,0){0.1}}
\put(2.5,1.2){\vector(1,0){0.1}}
\put(2.3,0.9){\vector(1,0){0.1}}
\end{picture}
\caption{The contours $\Gamma_{\theta}$ and $\Gamma_{\pi}$ resulting from the substitution \eqref{sub1} are shown. From \eqref{pithetadef1} the poles $\theta_k$ are centred around the degenerate value $\sigma$ that has been mapped to the origin, and likewise the poles $\pi_j$ are centred around the degenerate value $q$ mapped to unity.}
\label{gWcontours2}
\end{figure}
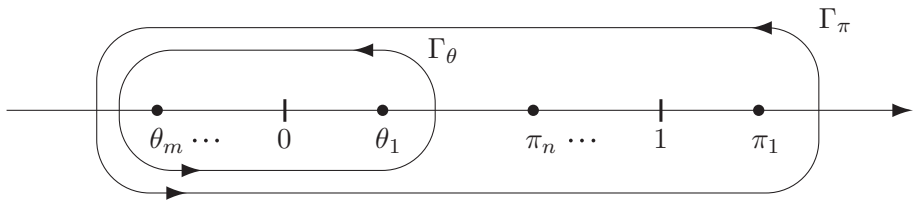

In order to take the large-$N$ limit let us introduce the following notation for the locations of the nontrivial poles of the integrand \eqref{K1start}:
\begin{equation}
\pi_{l}=\frac{q+\sigma_{l}}{q+\sigma}\,,\ \ l=1,\ldots,n\,, \quad\text{and }\;\theta_{k}=\frac{q-q_{k}}{q+\sigma}\,,\ \ k=1,\ldots,m\,.
\label{pithetadef1}
\end{equation}
These will be rescaled as
\begin{equation}
\hat{\pi}_{l}=\lim_{N\to\infty}\frac{{\pi}_{l}}{N}\,,\quad \hat{\theta}_{k}=\lim_{N\to\infty}\frac{{\theta}_{k}}{N}\,,  \quad  \hat{\pi}_{l}-\hat{\theta}_{k} \geq 0,  \quad \hat{\pi}_{l} \geq 0, \forall k,l,
\label{ptscale1}
\end{equation}
as well as the integration variables
\begin{equation}
v=N\hat{v}\,,\quad u= N\hat{u}\,,
\label{intscale1}
\end{equation}
and the arguments of the kernel
\begin{equation}
x=\frac{\hat{x}}{(q+\sigma)N}\,,\quad y=\frac{\hat{y}}{(q+\sigma)N}\,.
\label{xyscale1}
\end{equation}
Note that after this rescaling in Figure \ref{gWcontours2} unity is mapped to $1/N$ that moves to the origin when $N\to \infty$. Consequently, the limiting variables $\hat{\pi}_j$ and $\hat{\theta}_k$ may no longer be separated in the large-$N$ limit. Because of the choice of nested contour integrals this is not a problem. This leads to the following result for the limiting kernel at the hard edge.
\begin{thm}\label{ThmLimker1}
Consider the finite rank perturbations \eqref{sqdegenerate} and define the parameters $\hat{\pi}_l$, $l=1,\ldots,n$ and  $\hat{\theta}_k$, $k=1,\ldots,m$ as in \eqref{pithetadef1} and \eqref{ptscale1}. Then the following limit of the kernel \eqref{kernelGeneralizedWishart} leads to a kernel equivalent to
\begin{eqnarray}
\mathbb{K}_{\mathrm{III}}^{(n,m)}(\hat{x},\hat{y})&=&\lim_{N\to\infty}\frac{1}{(q+\sigma)N} K_N^{(1)}\left(\frac{\hat{x}}{(q+\sigma)N},\frac{\hat{y}}{(q+\sigma)N}\right)
e^{-\frac{q}{(q+\sigma)N}(\hat{x}-\hat{y})}
\nonumber\\
&=&\oint_{{\Gamma}_{\rm out}}\frac{d\hat{v}}{2\pi i}\oint_{{\Gamma}_{\rm in}}\frac{d\hat{u}}{2\pi i}
\frac{\exp\left[-\hat{x}\hat{v} +\hat{y}\hat{u}\right]}{\hat{u}-\hat{v}}
\ e^{-\frac{1}{\hat{u}}+\frac{1}{\hat{v}}}
\left(\frac{\hat{v}}{\hat{u}}\right)^{\nu+n-m}
\prod_{l=1}^n\frac{\hat{u}-\hat{\pi}_l}{\hat{v}-\hat{\pi}_l}
\prod_{k=1}^m\frac{\hat{v}-\hat{\theta}_k}{\hat{u}-\hat{\theta}_k}\,,
\label{KIIIfinal}
\end{eqnarray}
which is a Bessel kernel with finite rank perturbations. The closed integration contour ${\Gamma}_{\rm in}$ encircles the $\hat{\theta}_{j=1,\ldots,m}$ including the origin counter-clockwise, and the closed contour ${\Gamma}_{\rm out}$ contains the $\hat{\pi}_{l=1,\ldots,n}$ and encircles them counter-clockwise. It also contains the contour ${\Gamma}_{\rm in}$ without intersecting it.
\end{thm}
\begin{proof}
We take \eqref{K1start} as starting point, being equal to \eqref{kernelGeneralizedWishart}, and insert the finite rank conditions \eqref{sqdegenerate}. The prefactor $\exp[q(x-y)]$ can be removed, as it leads to an equivalent kernel. The scaling variables are defined in \eqref{intscale1} and \eqref{xyscale1}. Clearly the $N$-dependence drops out everywhere, except in the factor
\begin{equation}
\left(\frac{1-\frac{1}{\hat{u}N}}{1-\frac{1}{\hat{v}N}}\right)^{N-n}
\sim \exp\left[-\frac{1}{\hat{u}}+\frac{1}{\hat{v}}\right]\,,
\quad \mbox{as} \quad N\to\infty\,.
\label{explim}
\end{equation}
Due to Lebesgue's Dominated Convergence Theorem we can interchange the limit $N\to\infty$  and the double contour integral to apply this limit. Recalling that $n,m$ and $\nu=M-N$ are fixed in this limit we arrive at \eqref{KIIIfinal}.
\end{proof}
Note that in this ensemble the domains of the parameters $\hat{\theta}_{k=1,\ldots,m}\leq\hat{\pi}_{l=1,\ldots,n}$ are not restricted and can be the entire real line. As mentioned already, without parameters $\hat{\theta}_{k}$ the kernel \eqref{KIIIfinal} at $m=0$ was found in \cite{DZL} for the ensemble \eqref{JointDensityGX} with $Q\sim\eins_M$. The very same kernel was found previously in \cite[Theorem~15]{DF06} for the Wishart ensemble with an external field. There, it was also shown that without any finite rank perturbations, that is when $n=0$ (and $m=0$ here), it agrees with the Bessel kernel
\begin{equation}
\mathbb{K}_{\rm III}^{(0,0)}(x,y)=
\mathbb{K}_{\rm Bessel}(x,y)=\oint_{{\Gamma}_{\rm out}}\frac{d\hat{v}}{2\pi i}\oint_{{\Gamma}_{\rm in}}\frac{d\hat{u}}{2\pi i}
\frac{\exp\left[-x\hat{v} +y\hat{u}\right]}{\hat{u}-\hat{v}}
\ e^{-\frac{1}{\hat{u}}+\frac{1}{\hat{v}}}
\left(\frac{\hat{v}}{\hat{u}}\right)^{\nu} \,,
\label{BesselC}
\end{equation}
with ${\Gamma}_{\rm out}$ and ${\Gamma}_{\rm in}$ two nested, non-intersecting contours that both enclose the origin in counter-clockwise direction. The relation to the Bessel kernel can be spelled out more explicitly by bringing the kernel \eqref{KIIIfinal} to a form that is called integrable. Namely, a kernel is called integrable  in the sense of \cite{IIKS}, if it can be written as
\begin{equation}
K(x,y)=\sum_{l=1}^L \frac{F_l(x)G_l(y)}{x-y}\,,\quad \mbox{with}\quad \sum_{l=1}^L {F_l(x)G_l(x)}=0\,,
\end{equation}
holding for some given functions $F_l(x)$ and $G_l(x)$ and fixed value of $L$.
\begin{cor}
\label{Int1}
The Bessel kernel with finite rank perturbations \eqref{KIIIfinal} is integrable and can be written in the following form:
\begin{equation}
\mathbb{K}_{\mathrm{III}}^{(n,m)}(x,y)=\mathbb{K}_{\rm Bessel}(x,y)\big|_{\nu\to \nu+n-m}- \sum_{i=1}^m\tilde{\Lambda}^{(i)}_{\mathrm{III}}(x)\tilde{\Xi}^{(i)}_{\mathrm{III}}(y)
+\sum_{j=1}^n{\Lambda}^{(j)}_{\mathrm{III}}(x){\Xi}^{(j)}_{\mathrm{III}}(y),
\label{K1integrab}
\end{equation}
where we introduce four functions
\begin{eqnarray}
\tilde{\Lambda}^{(i)}_{\mathrm{III}}(x)&=&\oint_{{\Gamma}_{\rm \hat{\pi}}}\frac{d\hat{v}}{2\pi i}
{\exp\left[-x\hat{v} +\frac{1}{\hat{v}}\right]}\hat{v}^{\nu+n-m} \prod_{k=1}^{i-1}(\hat{v}-\hat{\theta}_k)\,,
\nonumber\\
{\Lambda}^{(j)}_{\mathrm{III}}(x)&=&\oint_{{\Gamma}_{\rm \hat{\pi}}}\frac{d\hat{v}}{2\pi i}
{\exp\left[-x\hat{v} +\frac{1}{\hat{v}}\right]}\hat{v}^{\nu+n-m} \frac{\prod_{k=1}^{m}(\hat{v}-\hat{\theta}_k)}{\prod_{l=1}^{j}(\hat{v}-\hat{\pi}_l)}\,,
\nonumber\\
\tilde{\Xi}^{(i)}_{\mathrm{III}}(y)&=&\oint_{{\Gamma}_{\rm \hat{\theta}}}\frac{d\hat{u}}{2\pi i}
{\exp\left[y\hat{u} -\frac{1}{\hat{u}}\right]}\hat{u}^{-\nu-n+m} \prod_{k=1}^i\frac{1}{\hat{u}-\hat{\theta}_k}\,,
\nonumber\\
{\Xi}^{(j)}_{\mathrm{III}}(y)&=&\oint_{{\Gamma}_{\rm \hat{\theta}}}\frac{d\hat{u}}{2\pi i}
{\exp\left[y\hat{u} -\frac{1}{\hat{u}}\right]}\hat{u}^{-\nu-n+m} \frac{\prod_{l=1}^{j-1}(\hat{u}-\hat{\pi}_l)}{\prod_{k=1}^{m}(\hat{u}-\hat{\theta}_k)}\,.
\label{LXdef}
\end{eqnarray}
Here the closed contour ${\Gamma}_{\rm \hat{\theta}}$ contains the poles at $\hat{\theta}_{k=1,\ldots,m}$ and the origin, encircling them counter-clockwise, and the closed contour ${\Gamma}_{\hat{\pi}}$ contains the poles at $\hat{\pi}_{l=1,\ldots,n}$ and the origin, encircling them counter-clockwise.
\end{cor}
This corollary solves an open problem stated in \cite[Section 7.2]{DF06} to find such an integrable representation. Note that
as mentioned in the introduction the generalised Bessel kernel \eqref{KIIIfinal} enjoys a formal duality relation, interchanging the parameters $\hat{\pi}_l\leftrightarrow-\hat{\theta}_l$ and $\nu\to-\nu$, which ultimately amounts to interchange matrices $\Sigma$ and $Q$, and $N$ and $M$ in \eqref{corrWishart}. Thus this duality holds already for the kernel at finite-$N$, see \eqref{kernelGeneralizedWishart}.
\begin{proof}
The crucial step for the integrability is the following identity \cite[Eq. (5.12)]{DF06}
\begin{equation}
\frac{1}{\hat{u}-\hat{v}}\prod_{l=1}^n\frac{\hat{u}-\hat{\pi}_l}{\hat{v}-\hat{\pi}_l}=\frac{1}{\hat{u}-\hat{v}}+\sum_{k=1}^n\frac{1}{\hat{v}-\hat{\pi}_k}
\prod_{l=1}^{k-1}\frac{\hat{u}-\hat{\pi}_l}{\hat{v}-\hat{\pi}_l}\,,
\label{DFid}
\end{equation}
which we have to apply twice, in view of the two products in \eqref{KIIIfinal}:
\begin{equation}
\frac{1}{\hat{u}-\hat{v}}\prod_{l=1}^n\frac{\hat{u}-\hat{\pi}_l}{\hat{v}-\hat{\pi}_l}
\prod_{k=1}^m\frac{\hat{v}-\hat{\theta}_k}{\hat{u}-\hat{\theta}_k}
=\frac{1}{\hat{u}-\hat{v}}-
\sum_{i=1}^m\frac{\prod_{k=1}^{i-1}(\hat{v}-\hat{\theta}_k)}{\prod_{k=1}^i(\hat{u}-\hat{\theta}_k)}
+
\sum_{j=1}^n\frac{\prod_{l=1}^{j-1}(\hat{u}-\hat{\pi}_l)}{\prod_{k=1}^{m}(\hat{u}-\hat{\theta}_k)}
\frac{\prod_{k=1}^{m}(\hat{v}-\hat{\theta}_k)}{\prod_{l=1}^{j}(\hat{v}-\hat{\pi}_l)}
\,.
\label{DFid2}
\end{equation}
Inserting this identity into \eqref{KIIIfinal} the right-hand sides of Equations \eqref{BesselC} and \eqref{LXdef} can be read off. Regarding contours, only in the integral \eqref{BesselC}, where the pole $\frac{1}{\hat{u}-\hat{v}}$ remains present, the condition of non-intersection contours remains. In all other integrals the contour ${\Gamma}_{\rm out}$ can be deformed to ${\Gamma}_{\hat{\theta}}$ to contain the poles at $\hat{\theta}_{k=1,\ldots,m}$ and the origin, and the contour ${\Gamma}_{\rm in}$ to ${\Gamma}_{\hat{\pi}}$ to contain the poles at $\hat{\pi}_{l=1,\ldots,n}$ and the origin.

The kernel is integrable due to two observations. First, as shown in \cite{DF06} the contour integral \eqref{BesselC} is equivalent to the more common representation of the Bessel kernel, see e.g. \cite{ForresterLogGases}
\begin{equation}
\mathbb{K}_{\rm Bessel}(x,y)=
\frac{-\sqrt{x}J_{\nu+1}(\sqrt{x})J_\nu(\sqrt{y})+\sqrt{y}J_{\nu+1}(\sqrt{y})J_\nu(\sqrt{x})}{2(x-y)}\ ,
\label{TheBessel}
\end{equation}
making the first term in \eqref{K1integrab} integrable. The simple observation in \cite{DZL} states that any factorising sum, such as the two sums on the right-hand side of \eqref{K1integrab}, can be brought to an integrable form,
\begin{equation}
\sum_{l=1}^L f_l(x)g_l(y)=\frac{1}{x-y}\sum_{l=1}^L (xf_l(x)g_l(y)-f_l(x)yg_l(y))\,.
\label{factor}
\end{equation}
Finally we remark that an alternative representation to \eqref{K1integrab} could be obtained, by applying the identity \eqref{DFid} first to the product containing $\hat{\theta}_k$'s, and then to the product containing $\hat{\pi}_l$'s. This leads to an alternative identity to \eqref{DFid2} and different functions in \eqref{LXdef} that we do not display.
\end{proof}
It is well known that the Bessel kernel \eqref{TheBessel} is universal for various deformations of the Wishart-Laguerre ensemble, see  e.g. \cite{KV} for invariant ensembles and \cite{DF06} for external fields. Theorem \ref{Lim-kernelgW} adds a further ensemble \eqref{corrWishart} to this list, namely when the finite rank perturbations are chosen such that their values vanish, $\hat{\theta}_k,\hat{\pi}_l\to0$ $\forall\, k,l$. This leads from \eqref{KIIIfinal} to \eqref{BesselC}.
\subsection{Meijer $G$-kernel with finite rank perturbations from   two independent correlated matrices}\label{Lim-kernelindep}
The ensemble \eqref{JointDensityindepGX} of two independent random matrices where one has correlated entries reads
\begin{equation}
\mathcal{P}_2(G,X)=\,c_2
\exp\left[-\alpha
\Tr\left(GG^*\right)-\Tr\left(QXX^*\right)\right]\ .
\nonumber
\end{equation}
The Hermitian matrix $Q$ has positive eigenvalues $q_1,\ldots,q_M>0$, and $\alpha>0$ is a constant (that does not depend on $\mu$ here). We will consider finite rank perturbations of $Q=q\eins_M$ by setting
\begin{equation}
\;q_{m+1}=\cdots=q_{M}=q>0\,,
\label{qdegenerate}
\end{equation}
with $m$ independent of $N$.

Putting the degeneracy \eqref{qdegenerate} inside the kernel \eqref{kernel-indep} for finite-$N$ from Theorem \ref{Thmker2}, we obtain
\begin{eqnarray}
K_{N}^{(2)}(x,y)&=&\left(\frac{x}{y}\right)^{\frac{\kappa}{2}}2\oint_{\gamma_{0}}\frac{d\eta}{2\pi i}\oint_{\gamma_{q}}\frac{d\zeta}{2\pi i}\left(\frac{\zeta}{\eta}\right)^{N+\frac{\kappa}{2}}\frac{I_{\kappa}\left(2\sqrt{\eta x}\right)K_{\kappa}\left(2\sqrt{\zeta y}\right)}{\eta-\zeta}\left(\frac{\eta-\alpha q}{\zeta-\alpha q}\right)^{M-m}\prod_{l=1}^{m}\frac{\eta-\alpha q_{l}}{\zeta-\alpha q_{l}}
\nonumber\\
&=&\left(\frac{x}{y}\right)^{\frac{\kappa}{2}}2\alpha q\oint_{\Gamma_{1}}\frac{dv}{2\pi i}\oint_{\Gamma_{\theta}}\frac{du}{2\pi i}\frac{I_{\kappa}\left(2\sqrt{\alpha q(1-v) x}\right)K_{\kappa}\left(2\sqrt{\alpha q\left(1-u\right) y}\right)}{u-v}\nonumber\\
&&\times\left(\frac{1-\frac{1}{u}}{1-\frac{1}{v}}\right)^{N+\frac{\kappa}{2}}\left(\frac{v}{u}\right)^{\nu-m-\frac{\kappa}{2}}\prod_{k=1}^{m}\frac{v-\left(1-\frac{q_{k}}{q}\right)}{u-\left(1-\frac{q_{k}}{q}\right)}\,,
\label{K2start}
\end{eqnarray}
with the following substitution
\begin{align}
\zeta=\alpha q\left(1-u\right)\quad\text{and }\; \eta=\alpha q(1-v)\,.
\label{sub2indep}
\end{align}
It is nonsingular due to $\alpha q>0$, and we define for later the poles in the new integration variables
\begin{align}
\theta_{k}
=1-\frac{q_{k}}{q}\,,\ \ k=1,\ldots,m\,.
\label{thetadef}
\end{align}
The integration contours in the new variables obtained from Figure \ref{Fig2} that we take to be nested here are given in Figure \ref{Fig2indep}.
\begin{figure}[h]
\unitlength1cm
\begin{picture}(12.2,3.5)
\thinlines
\put(0,2){\line(1,0){12}}
\linethickness{0.4mm}
\put(11.95,2){\vector(1,0){0.1}}
\thinlines
\put(1.9,1.5){$\theta_{m}$}
\put(2,2){\circle*{0.15}}
\multiput(2.5,1.6)(0.15,0){3}{\circle*{0.05}}
\put(3.6,1.5){$0$}
\linethickness{0.4mm}
\put(3.7,1.85){\line(0,1){0.3}}
\thinlines
\put(4.9,1.5){$\theta_{1}$}
\put(5,2){\circle*{0.15}}
\put(8.6,1.5){$1$}
\linethickness{0.4mm}
\put(8.7,1.85){\line(0,1){0.3}}
\thinlines
\put(5.6,2.7){$\Gamma_{\theta}$}
\put(8.9,3.1){$\Gamma_{1}$}
\put(3.6,2){\oval(4.2,1.6)}
\put(5.3,2){\oval(8,2.2)}
\linethickness{0.4mm}
\put(4.7,2.8){\vector(-1,0){0.1}}
\put(8.2,3.1){\vector(-1,0){0.1}}
\put(2.5,1.2){\vector(1,0){0.1}}
\put(2.3,0.9){\vector(1,0){0.1}}
\end{picture}
\caption{The contours $\Gamma_{\theta}$ and $\Gamma_{1}$ resulting from the substitution \eqref{sub2indep}, the $\theta_k$ are now centred around the origin and the origin has been mapped to unity.}
\label{Fig2indep}
\end{figure}
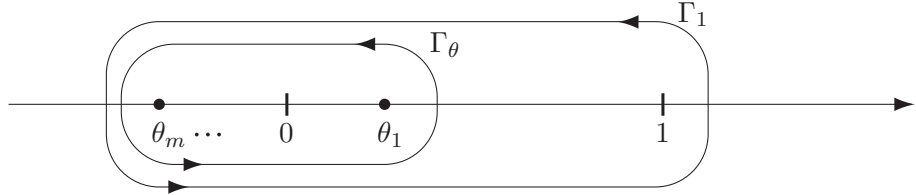

The integration variables will be rescaled, and the poles take the following limiting values:
\begin{equation}
v=N\hat{v}\,,\quad u= N\hat{u}\,,\quad\hat{\theta}_{k}= \lim_{N\to\infty}\frac{1}{N} {\theta}_{k}\,,\quad k=1, \ldots, m,
\label{scale2}
\end{equation}
mapping the identity in Figure \ref{Fig2indep} to $1/N$, and thus to the origin in the limit $N\to\infty$. This is not a problem, due to the nesting of the contours. Consequently, the $\hat{\theta}_{k=1,\ldots,m}$ become non-positive in the limit $N\to\infty$. The arguments of the kernel will be rescaled as
\begin{equation}
x=\frac{\hat{x}}{\alpha qN}\,,\quad y=\frac{\hat{y}}{\alpha qN}\,.
\label{xyscale2}
\end{equation}
Putting all together yields the following result for the limiting kernel at the hard edge.
\begin{thm}\label{ThmLimker2}
Define the finite rank perturbations \eqref{qdegenerate} and the parameters $\hat{\theta}_k$, $k=1,\ldots,m$, as in \eqref{thetadef} and \eqref{scale2}. The following limit of the kernel \eqref{kernel-indep} leads to a kernel equivalent to
\begin{eqnarray}
\mathbb{K}_{\mathrm{I}}^{(m)}(\hat{x},\hat{y})&=&\lim_{N\to\infty}\frac{1}{\alpha qN} K_N^{(2)}\left(\frac{\hat{x}}{\alpha qN},\frac{\hat{y}}{\alpha qN}\right)
\left(\frac{\hat{x}}{\hat{y}}\right)^{-{\kappa}}
\nonumber\\
&=&\oint_{{\gamma}_0}\frac{d{s}}{2\pi i}\int_0^\infty dt\,s^{-\kappa-1}t^{\kappa-1}e^{s-t}\
\mathbb{K}_{\mathrm{III}}^{(n=0,m)}\left(\frac{\hat{x}}{s},\frac{\hat{y}}{t}\right).
\label{KIfinal}
\end{eqnarray}
The limiting kernel is a Meijer $G$-kernel with finite rank perturbations.
\end{thm}
\begin{proof}
In order to take the large-$N$ limit we use the following integral representations of the modified Bessel functions of the second kind in \eqref{K2start}:
\begin{eqnarray}
I_\kappa(z)&=&\left(\frac{z}{2}\right)^{\kappa}\oint_{\gamma_0}\frac{ds}{2\pi i}\,s^{-\kappa-1}\exp\left[s+\frac{z^2}{4s}\right]\,,
\nonumber\\
K_\kappa(z)=K_{-\kappa}(z)&=&\frac{z^{-\kappa}}{2^{-\kappa+1}}\int_0^\infty dt\,t^{\kappa-1}\exp\left[-t-\frac{z^2}{4t}\right]\,.
\label{Besselint}
\end{eqnarray}
Here, $\gamma_0$ is a closed contour encircling the origin in counter-clockwise way. We use Fubini's Theorem to rewrite the rescaled kernel from \eqref{K2start} as
\begin{eqnarray}
\frac{1}{\alpha qN} K_N^{(2)}\left(\frac{\hat{x}}{\alpha qN},\frac{\hat{y}}{\alpha qN}\right)
\left(\frac{\hat{x}}{\hat{y}}\right)^{-{\kappa}}&=&
\oint_{{\gamma}_0}\frac{d{s}}{2\pi i}\int_0^\infty dt\frac{t^{\kappa-1}}{s^{\kappa+1}}
\oint_{\Gamma_{1}}\frac{d\hat{v}}{2\pi i}\oint_{\Gamma_{\theta}}\frac{d\hat{u}}{2\pi i}
\frac{e^{\,s+\frac1s(\frac1N-\hat{v})\hat{x}-t-\frac1t(\frac1N-\hat{u})\hat{y}}}{\hat{u}-\hat{v}}
\nonumber\\
&&\times
\left(\frac{\frac1N-\hat{v}}{\frac1N-\hat{u}}\right)^{\frac{\kappa}{2}}
\left(\frac{1-\frac{1}{N\hat{u}}}{1-\frac{1}{N\hat{v}}}\right)^{N+\frac{\kappa}{2}}
\left(\frac{\hat{v}}{\hat{u}}\right)^{\nu-m-\frac{\kappa}{2}}
\prod_{k=1}^{m}\frac{\hat{v}-\frac{{\theta}_k}{N}}{\hat{u}-\frac{{\theta}_k}{N}}\,.
\label{K2start2}
\end{eqnarray}
The application of Lebesgue's Dominated Convergence Theorem together with the limits of each individual factor inside the integrand (such as \eqref{explim}) leads to \eqref{KIfinal}, after comparing with \eqref{KIIIfinal} at $n=0$. We recall here that $\kappa=L-N$ is kept fixed in the large-$N$ limit.

Notice that in this ensemble the domain of the paramters  $\hat{\theta}_{k=1,\ldots,m}$ is non-positive. Furthermore, in \cite[Prop. 5.1]{DZL} it was shown that the kernel \eqref{KIfinal} at $m=0$ equals the Meijer $G$-kernel \cite{KZ14} for two independent matrices,
\begin{equation}
\mathbb{K}_{\mathrm{I}}^{(m=0)}({x},{y})=\mathbb{K}_{\rm Meijer}(x,y)=
\int_0^1du\
G^{1,0}_{0,3}\left(\begin{array}{c}
                       -   \\
                      0, -\nu, -\kappa
                    \end{array}
\biggr| uy\right)
G^{2,0}_{0,3}\left(\begin{array}{c}
                       -   \\
                      \nu,\kappa,0
                    \end{array}
\biggr| ux\right).
\label{KMeijer}
\end{equation}
Therefore, the kernel \eqref{KIfinal} represents a finite rank perturbation of the Meijer $G$-kernel. In the absence of finite rank perturbations the kernel was also shown to agree with \eqref{KMeijer} in \cite{AStr16b}.
\end{proof}
We note that using Fubini's Theorem and \eqref{Besselint} the integrals over $s$ and $t$ in \eqref{KIfinal} can be done. Using the relations \cite[Eqs. 8.406-7]{GradshteynRyzhik}
\begin{equation}
I_{\kappa}(iz)=i^{\kappa}J_{\kappa}(z)\,,
\quad\text{and }\; K_{-\kappa}(iz)=-\frac{\pi}{2}i^{\kappa+1}H_{\kappa}^{(2)}(z)
\end{equation}
in terms of the Bessel functions of first and third kind (also called Hankel functions), we obtain
\begin{eqnarray}
\mathbb{K}_{\mathrm{I}}^{(m)}(\hat{x},\hat{y})
&=&(-1)^{\kappa+1}i\pi\left(\frac{\hat{y}}{\hat{x}}\right)^{\frac{\kappa}{2}}
\oint_{\Gamma_{\rm out}}\frac{d\hat{v}}{2\pi i}\oint_{\Gamma_{\rm in}}\frac{d\hat{u}}{2\pi i}\frac{J_{\kappa}\left(2\sqrt{\hat{v}\hat{x}}\right)H_{\kappa}^{(2)}\left(2\sqrt{\hat{u} \hat{y}}\right)}{u-v} e^{-\frac{1}{\hat{u}}+\frac{1}{\hat{v}}}\left(\frac{\hat{v}}{\hat{u}}\right)^{\nu-m-\frac{\kappa}{2}}
\nonumber\\
&&\times\prod_{k=1}^{m}\frac{\hat{v}-\hat{\theta}_{k}}{\hat{u}-\hat{\theta}_{k}}.
\label{KI2intonly}
\end{eqnarray}

It remains to show that the kernel \eqref{KIfinal} is integrable.
\begin{cor}
\label{Int2}
The Meijer $G$-kernel with finite rank perturbations \eqref{KIfinal} is integrable and can be written as follows
\begin{equation}
\mathbb{K}_{\mathrm{I}}^{(m)}(x,y)=\mathbb{K}_{\rm Meijer}(x,y)\big|_{\nu\to \nu-m} - \sum_{i=1}^m\tilde{\Lambda}^{(i)}_{\mathrm{I}}(x)\tilde{\Xi}^{(i)}_{\mathrm{I}}(y)\,,
\label{K2integrab}
\end{equation}
where the remaining functions read
\begin{eqnarray}
\tilde{\Lambda}^{(i)}_{\mathrm{I}}(x)&=&
\oint_{{\gamma}_0}\frac{d\hat{s}}{2\pi i}\,s^{-\kappa-1}e^{s}\left. \tilde{\Lambda}^{(i)}_{\mathrm{III}}\left(\frac{x}{s}\right)\right|_{n=0},
\nonumber\\
\tilde{\Xi}^{(i)}_{\mathrm{I}}(y)&=&\int_0^\infty dt\,t^{\kappa-1}e^{-t}\left. \tilde{\Xi}^{(i)}_{\mathrm{III}}\left(\frac{y}{t}\right)\right|_{n=0},
\label{LX2def}
\end{eqnarray}
together with \eqref{LXdef} at $n=0$.
\end{cor}
\begin{proof}
Equation \eqref{K2integrab} together with \eqref{LX2def} immediately follow from Corollary \ref{Int1}, by inserting \eqref{K1integrab} at $n=0$ into \eqref{KIfinal}. While the integrability of the Meijer $G$-kernel \eqref{KMeijer} was shown in \cite{KZ14}, the integrability of the sum in \eqref{K2integrab} follows as in the proof of Corollary \ref{Int1}, using \eqref{factor}. The integrals over $s$ and $t$ in \eqref{LX2def} can be done in analogy to \eqref{KI2intonly} but will not be displayed here.
\end{proof}
As mentioned already in the introduction relatively little was known so far about the universality of the Meijer $G$-kernel. In the case when the finite rank perturbations are such that their limiting values vanish, $\hat{\theta}_k\to0$ $\forall k$, this leads from \eqref{KIfinal} to \eqref{KMeijer}. Thus in this particular case the ensemble of two independent random matrices where one has correlated entries \eqref{JointDensityindepGX} leads to the same Meijer $G$-kernel as the product of two independent matrices without correlations, showing its universality.
\subsection{Universal limiting kernels from  two correlated coupled matrices}\label{Lim-kernelcoupled}
Our most general ensemble of two coupled matrices where one has correlated entries \eqref{JointDensityGX} is given by
\begin{equation}
\mathcal{P}(G,X)=\,c
\exp\left[-\alpha\Tr\left(GG^*\right)
+\Tr\left(\Omega GX+X^*G^*\Omega^*\right)-\Tr\left(QXX^*\right)\right]\,.
\nonumber
\end{equation}
Here $\alpha>0$, the matrix $\Omega$ has squared singular values $\delta_1,\ldots,\delta_N$, and the Hermitian matrix $Q$ has positive eigenvalues $q_1,\ldots,q_M>0$. In this subsection we will consider the following parameter dependent finite rank perturbations around $\Omega\Omega^*=
\frac{(1-\mu)^2}{4\mu^2}
\eins_N$ and $Q=
\frac{1+\mu}{2\mu}
\eins_M$, by setting
\begin{equation}
\delta_{n+1}=\cdots=\delta_N=
\frac{(1-\mu)^2}{4\mu^2}, \qquad
q_{m+1}=\cdots=
q_M=
\frac{1+\mu}{2\mu},
\label{finiterank}
\end{equation}
together with
\begin{equation}
\alpha=\frac{1+\mu}{2\mu}\ ,
\label{alpha2}
\end{equation}
with $n,m$ independent of $N$. The ensemble with $m=n=0$ of two coupled random matrices was studied in \cite{AStr16a,AStr16b}. In contrast to the previous two subsections we have the parameter $\mu=\mu(N)\in(0,1]$ here, that can vary with $N$. As it was already found in \cite{AStr16a,AStr16b}, depending on the scaling of $\mu(N)$ we will obtain three different limits $N\to\infty$: (I) $\mu(N) N\to \infty$, (II) $\mu(N) N\to\tau/4$ and (III) $\mu(N) N\to0$.

Because it was already discussed in detail in the Introduction after \eqref{nondeg-nondeg}, what is the range that the various parameters take due to the condition \eqref{aqd}, in particular after setting some of these equal as in \eqref{finiterank}, we only summarise the findings that we need below. Let us recall the following quantities:
\begin{eqnarray}
\pi_l&=&\frac{(1+\mu)^2}{4\mu } \left(1-\frac{4\mu^2\delta_l}{(1+\mu)^2}\right),\ \ l=1,\ldots n,
\label{pi-def2}\\
\theta_j&=&  \frac{(1+\mu)^2}{4\mu } \left( 1-\frac{2\mu q_j}{1+\mu}\right) ,\ \ j=1,\ldots,m,
\label{theta-def2}
\end{eqnarray}
which were shown to satisfy
\begin{equation}
0<\pi_l \ ,\ \  \theta_j < \pi_l\ ,\ \ \mbox{and} \ \ \theta_j<1 \ ,\ \ \mbox{for}\ \  j=1,\ldots,m,\ \ l=1,\ldots n.
\label{allbounds}
\end{equation}
Compared to Subsection \ref{Lim-kernelgW}, the natural scale is here $\frac{(1+\mu)^2}{4\mu^2}-\frac{(1-\mu)^2}{4\mu^2}=\frac{1}{\mu}$. We therefore make the following substitutions in the integrals in the kernel at finite-$N$ for this ensemble \eqref{MainKerneloint}, which are non-singular:
\begin{equation}
\zeta=
\frac{(1+\mu)^2}{4\mu^2}-\frac{1}{\mu}u
\ ,\quad \eta=
\frac{(1+\mu)^2}{4\mu^2}-\frac{1}{\mu}v
\,.
\label{substitute}
\end{equation}
Together with equations \eqref{finiterank}, \eqref{alpha2}, \eqref{pi-def2}, and \eqref{theta-def2}, the kernel \eqref{MainKerneloint} results into
\begin{eqnarray}
K_{N}(x,y)&=& \frac{2}{\mu}\oint_{\Gamma_{\pi}}\frac{dv}{2\pi i}\oint_{\Gamma_{\theta}}\frac{du}{2\pi i}
\frac{I_{\kappa}\left(2
\sqrt{((1+\mu)^2-4\mu v)\frac{x}{4\mu^2}}
\right)K_{\kappa}\left(2
\sqrt{((1+\mu)^2-4\mu u)\frac{y}{4\mu^2}}
\right)}{u-v}\nonumber\\
&&\times \left(\frac{x}{y}\right)^{\frac{\kappa}{2}}\left(\frac{u-
\frac{(1+\mu)^2}{4\mu}
}{v-
\frac{(1+\mu)^2}{4\mu}
}\right)^{\frac{\kappa}{2}}
\left(\frac{1-\frac{1}{u}}{1-\frac{1}{v}}\right)^{N-n}\left(\frac{v}{u}\right)^{\nu-m+n}\prod_{l=1}^{n}\frac{u-\pi_{l}}{v-\pi_{l}}\prod_{k=1}^{m}\frac{v-\theta_{k}}{u-\theta_{k}}\,,
\label{Kstart}
\end{eqnarray}
For simplicity here and from now on we will suppress the $N$-dependence of all parameters. Taking \eqref{Kstart} as a starting point, together with the scaling
\begin{equation}
u=N\hat{u}\ ,\ \ v=N\hat{v}\ ,
\end{equation}
and
\begin{equation}
\hat{\pi}_{l}=\lim_{N\to\infty}\frac1N{\pi}_{l}\,,\quad \hat{\theta}_{j}= \lim_{N\to\infty}\frac1N{\theta}_{j}\,,\quad  \hat{\pi}_{l}-\hat{\theta}_{j} \geq 0,  \quad \hat{\pi}_{l} \geq 0,\quad  \forall j,l,
\label{ptscale}
\end{equation}
we are in the position to present the proof of Theorem \ref{hardlimits}.
\begin{proof}[Proof of Theorem \ref{hardlimits} (I)]
In limit (I) leading to the Meijer $G$-kernel with finite rank perturbations we will let $\mu N\to\infty$, and rescale the arguments of the kernel as
\begin{equation}
x=\frac{\mu\hat{x}}{N}\,,\quad y=\frac{\mu\hat{y}}{N}\,.
\label{xyscale}
\end{equation}
From the definition \eqref{pi-def2} and the bounds \eqref{allbounds} we have that in the limit \eqref{ptscale} the following parameters vanish, $\hat{\pi}_l=
0$, for all $l=1,\ldots,n$. Furthermore, from \eqref{theta-def2} we have that the
parameters $\hat{\theta}_k
\leq0$, for all $k=1,\ldots,m$, will become negative (or zero) in the large-$N$ limit. The contours in the integral \eqref{Kstart} are thus as in Figure \ref{Fig2indep}, with all ${\theta}_{k=1,\ldots,m}$ to the left of the origin.

Following the previous Subsection \ref{Lim-kernelindep}, we apply the integral representations of the modified Bessel functions \eqref{Besselint}, use Fubini's Theorem to interchange integrals and Lebesgue's Dominated Convergence Theorem to exchange the limit with the integrations. We obtain the following for the kernel in terms of rescaled variables:
\begin{eqnarray}
\frac{\mu}{N}
K_{N}\left(\frac{\mu\hat{x}}{N},\frac{\mu\hat{y}}{N}\right)\frac{\hat{y}^{\kappa}}{\hat{x}^{\kappa}}
&=&\oint_{{\gamma}_0}\frac{d{s}}{2\pi i}\int_0^\infty dt
\oint_{\Gamma_{1}}\frac{d\hat{v}}{2\pi i}\oint_{\Gamma_{\theta}}\frac{d\hat{u}}{2\pi i}
\frac{t^{\kappa-1}}{s^{\kappa+1}}
\frac{1}{\hat{u}-\hat{v}}\left(\frac{\hat{v}}{\hat{u}}\right)^{\nu-m+n}
\prod_{k=1}^{m}\frac{\hat{v}-\frac{{\theta}_{k}}{N}}{\hat{u}-\frac{{\theta}_{k}}{N}}
\prod_{l=1}^{n}\frac{\hat{u}-\frac{{\pi}_{k}}{N}}{\hat{v}-\frac{{\pi}_{k}}{N}}
\nonumber\\
&&\times
\frac{\left(
((1+\mu)^2-4\mu N \hat{v})\frac{1}{4\mu N}
\right)^{\frac{\kappa}{2}}}{\left(
((1+\mu)^2-4\mu N \hat{u})\frac{1}{4\mu N}
\right)^{\frac{\kappa}{2}}}
\ e^{s+\frac1s
((1+\mu)^2-4\mu N \hat{v})\frac{\hat{x}}{4\mu N}
-t-\frac1t
((1+\mu)^2-4\mu N \hat{u})\frac{\hat{y}}{4\mu N}
}
\nonumber\\
&&\times\left(\frac{\hat{u}-\frac{(1+\mu)^2}{4\mu N}}{\hat{v}-\frac{(1+\mu)^2}{4\mu N}}\right)^{\frac{\kappa}{2}}\left(\frac{1-\frac{1}{\hat{u}N}}{1-\frac{1}{\hat{v}N}}\right)^{N-n}.
\label{Klastlim}
\end{eqnarray}
The large-$N$ limit $\mu N\to\infty$ of the second and third line are easily taken, 
and after cancelling several factors we arrive at
\begin{equation}
\lim_{N\to\infty}\frac{\mu}{N}
K_{N}\left(\frac{\mu\hat{x}}{N},\frac{\mu\hat{y}}{N}\right)\frac{\hat{y}^{\kappa}}{\hat{x}^{\kappa}}=
\oint_{{\gamma}_0}\frac{d\hat{s}}{2\pi i}\int_0^\infty dt\,s^{-\kappa-1}t^{\kappa-1}e^{s-t}\
\mathbb{K}_{\mathrm{III}}^{(n=0,m)}\left(\frac{\hat{x}}{s},\frac{\hat{y}}{t}\right),
\end{equation}
the right hand side of \eqref{KIfinal}. This is the limiting kernel as stated in the theorem. Obviously, it is true uniformly for any $\hat{x}$, $\hat{y}$ in a compact subset of $(0,\infty)$.
\end{proof}
The fact that this kernel is integrable has already been stated in Corollary \ref{Int2}.
\begin{proof}[Proof of Theorem \ref{hardlimits} (III)] In limit (III) leading to the Bessel kernel with finite rank perturbations we will let $\mu N\to0$ and rescale the arguments of the kernel as
\begin{equation}
x=\frac{\hat{x}}{4N^2}\,,\quad y=\frac{\hat{y}}{4N^2}\,.
\label{xyscaleIII}
\end{equation}
Following \eqref{pi-def2}, the positivity of the parameters $\hat{\pi}_{l=1,\ldots,n}$ \eqref{allbounds} leads to their limiting domain to be $[0,\infty)$. The domain of the limiting parameters $\hat{\theta}_{k=1,\ldots,m}$ becomes $\cap_{l=1}^{n}(-\infty,\hat{\pi}_l ]$. In the scaling limit the arguments of the modified Bessel functions in \eqref{Kstart} become large,
\begin{eqnarray}
2\sqrt{((1+\mu)^2-4\mu N \hat{v})\frac{\hat{x}}{(4\mu N)^2}}
=\frac{(1+\mu)\sqrt{\hat{x}}}{2\mu N}\left(1-\frac{4\mu N}{(1+\mu)^2}\hat{v}\right)^{\frac12}=\frac{\sqrt{\hat{x}}}{2\mu N}\left(1-2\mu N\hat{v}+\mathcal{O}((\mu N)^2)\right),
\nonumber
\end{eqnarray}
and likewise for the other argument with $\hat{u}$ and $\hat{y}$. For that reason we need to make use of asymptotic formulas of modified Bessel functions as $z\ra \infty$
\begin{equation} I_{\kappa}(z)\sim \frac{e^z}{\sqrt{2\pi z}} \ \mathrm{for} \ |\mathrm{arg}(z)|\leq \frac{1}{2}\pi -\beta \quad \mathrm{and} \quad K_{\kappa}(z)\sim   \sqrt{\frac{\pi}{2z}} e^{-z}  \  \mathrm{for} \  |\mathrm{arg}(z)|\leq \frac{3}{2}\pi-\beta, \label{asymptoticsIK}
\end{equation}
uniformly for arbitrary $0<\beta<\pi/2$, from which simple calculations give us
\begin{equation} I_{\kappa}\left(2
\sqrt{((1+\mu)^2-4\mu N \hat{v})\frac{\hat{x}}{(4\mu N)^2}}
\right) \sim \frac{\sqrt{\mu N}}{\sqrt{\pi}}
\hat{x}^{-\frac{1}{4}}
e^{\frac{1}{2\mu N}\sqrt{\hat{x}}  -\hat{v}\sqrt{\hat{x}}},\quad \mathrm{as}\quad  \mu N\ra 0,
\end{equation}
and
\begin{equation}
K_{\kappa}\left(2
\sqrt{((1+\mu)^2-4\mu N \hat{u})\frac{\hat{y}}{(4\mu N)^2}}
\right) \sim \sqrt{\mu N\pi}\
\hat{y}^{-\frac{1}{4}}
e^{-\frac{1}{2\mu N} \sqrt{\hat{y}} +\hat{u}\sqrt{\hat{y}}},\quad \mathrm{as}\quad  \mu N\ra 0.
\end{equation}

Apart from the asymptotic result \eqref{explim} we also need that
\begin{equation}
\left(\frac{4\mu N\hat{u}-(1+\mu)^2}{4\mu N\hat{v}-(1+\mu)^2}\right)^{\frac{\kappa}{2}}\sim 1 \quad \mbox{as}\quad N\to\infty\,.
\end{equation}
Finally, we use Lebesgue's Dominated Convergence Theorem to exchange the limit with the integrations in \eqref{Kstart}. Collecting all the above expansions results into the following limit
\begin{eqnarray}
\lim_{N\to\infty}\frac{1}{4N^2}
K_{N}\left(\frac{\hat{x}}{4N^2},\frac{\hat{y}}{4N^2}\right)\frac{\hat{y}^{\frac{\kappa}{2}}}{\hat{x}^{\frac{\kappa}{2}}}\ e^{\frac{1}{2\mu N}(\sqrt{\hat{y}}-\sqrt{\hat{x}})}=
\frac12 (\hat{x} \hat{y})^{-\frac{1}{4}}\mathbb{K}_{\rm III}^{(n,m)}\left(\sqrt{\hat{x}},\sqrt{\hat{y}}\right)\,,
\end{eqnarray}
as stated in the theorem.
\end{proof}
Corollary \ref{Int1} implies that this kernel is integrable.
\begin{proof}[Proof of Theorem \ref{hardlimits} (II)]
In limit (II) leading to an interpolating kernel we will let $\mu N\to\tau/4$ with fixed $\tau>0$ of $\mathcal{O}(1)$. The scaling of the arguments of the kernel is given as in \eqref{xyscaleIII}. In this limit, due to \eqref{allbounds} the domain of
$\hat{\pi}_{l=1,\ldots,n}
$ \eqref{ptscale} becomes $[0,1/\tau)$, while the limiting $\hat{\theta}_{k=1,\ldots,m}$ remain in the interval $\cap_{l=1}^{n}(-\infty,\hat{\pi}_l]$. Let us write the kernel \eqref{Kstart} in terms of these scaling variables:
\begin{eqnarray}
\frac{1}{4N^2}
K_{N}\left(\frac{\hat{x}}{4N^2},\frac{\hat{y}}{4N^2}\right)\frac{\hat{y}^{\frac{\kappa}{2}}}{\hat{x}^{\frac{\kappa}{2}}}
&=& \frac{2}{4\mu N}\oint_{\Gamma_{1}}\frac{d\hat{v}}{2\pi i}\oint_{\Gamma_{\theta}}\frac{d\hat{u}}{2\pi i} \frac{1}{\hat{u}-\hat{v}}
\left(\frac{\hat{v}}{\hat{u}}\right)^{\nu-m+n}
\prod_{l=1}^{n}\frac{\hat{u}-\frac{{\pi}_{l}}{N}}{\hat{v}-\frac{{\pi}_{l}}{N}}\prod_{k=1}^{m}\frac{\hat{v}-\frac{{\theta}_{k}}{N}}{\hat{u}-\frac{{\theta}_{k}}{N}}\nonumber\\
&&\times I_\kappa\left(2\sqrt{((1+\mu)^2-4\mu N\hat{v})\frac{\hat{x}}{(4\mu N)^2}}\right)\left(\frac{4\mu N\hat{u}-(1+\mu)^2}{4\mu N\hat{v}-(1+\mu)^2}\right)^{\frac{\kappa}{2}}
\nonumber\\
&&\times K_\kappa\left(2\sqrt{((1+\mu)^2-4\mu N\hat{u})\frac{\hat{y}}{(4\mu N)^2}}\right)\left(\frac{1-\frac{1}{\hat{u}N}}{1-\frac{1}{\hat{v}N}}\right)^{N-n}.
\end{eqnarray}
Using once again Lebesgue's Dominated Convergence Theorem, the limit of the two last lines is easily taken, and we obtain
\begin{eqnarray}
\lim_{N\to\infty}
\frac{1}{4N^2}
K_{N}\left(\frac{\hat{x}}{4N^2},\frac{\hat{y}}{4N^2}\right)\frac{\hat{y}^{\frac{\kappa}{2}}}{\hat{x}^{\frac{\kappa}{2}}}
&=& \frac{2}{\tau}\oint_{\Gamma_{\rm out}}\frac{d\hat{v}}{2\pi i}\oint_{\Gamma_{\rm in}}\frac{d\hat{u}}{2\pi i} \frac{1}{\hat{u}-\hat{v}}
\left(\frac{\hat{v}}{\hat{u}}\right)^{\nu-m+n}
\prod_{l=1}^{n}\frac{\hat{u}-\hat{\pi}_{l}}{\hat{v}-\hat{\pi}_{l}}\prod_{k=1}^{m}\frac{\hat{v}-\hat{\theta}_{k}}{\hat{u}-\hat{\theta}_{k}}
\nonumber\\
&&\times I_\kappa\left(2\sqrt{\frac{\hat{x}}{\tau^2}(1-\tau\hat{v})}\right)K_\kappa\left(2\sqrt{\frac{\hat{y}}{\tau^2}(1-\tau\hat{u})}\right)
\left(\frac{1-\tau\hat{u}}{1-\tau\hat{v}}\right)^{\frac{\kappa}{2}}e^{-\frac{1}{\hat{u}}+\frac{1}{\hat{v}}}
\nonumber\\
&=&\mathbb{K}_{\rm II}^{(n,m)}\left({\hat{x}},{\hat{y}};\tau\right)\,.
\label{KIIfinal}
\end{eqnarray}
\end{proof}
We note that the interpolating kernel can also be written as a double integral of the generalised Bessel kernel with finite rank perturbations. Using Fubini's Therorem and the integral representations \eqref{Besselint} we obtain
\begin{eqnarray}
\mathbb{K}_{\rm II}^{(n,m)}\left({\hat{x}},{\hat{y}};\tau\right)&=& \frac{1}{\tau}
\left(\frac{\hat{x}}{\hat{y}}\right)^{\frac\kappa2}
\oint_{{\gamma}_0}\frac{d\hat{s}}{2\pi i}\int_0^\infty dt\,s^{-\kappa-1}t^{\kappa-1}e^{s-t}e^{\frac{\hat{x}}{s\tau^2}-\frac{\hat{y}}{t\tau^2}}
\mathbb{K}_{\mathrm{III}}^{(n,m)}\left(\frac{\hat{x}}{s\tau},\frac{\hat{y}}{t\tau}\right),
\label{KII4int}
\end{eqnarray}
We use a  different definition here for the interpolating kernel, compared to \cite{DZL} where $m=0$. There, the scaling \eqref{xyscaleIII} was made $\mu$-dependent, leading to
$\mathbb{K}_{\rm II}^{(n,0)}\left({\hat{x}},{\hat{y}};\tau\right)\to\frac{1}{\tau^2}\mathbb{K}_{\rm II}^{(n,0)}\left(\frac{\hat{x}}{\tau^2},\frac{\hat{y}}{\tau^2};\tau\right)$. When the finite rank parameters $\hat{\theta}_k\to0$, for all $k=1,\ldots, m$, we obtain the same kernel as in \cite{DZL} (up to rescaling), which is thus universal. The question of universality under further deformations is open.

In analogy to the previous two subsections we can show that the interpolating kernel is integrable.
\begin{cor}
\label{KII-integrable}
The family of limiting kernels with finite rank perturbations \eqref{KIIfinal} is integrable and can be written in the following form:
\begin{equation}
\mathbb{K}_{\mathrm{II}}^{(n,m)}(x,y)=\mathbb{K}_{\mathrm{II}}^{(0,0)}(x,y)\big|_{\nu\to \nu+n-m}
- \sum_{i=1}^m\tilde{\Lambda}^{(i)}_{\mathrm{II}}(x)\tilde{\Xi}^{(i)}_{\mathrm{II}}(y)
+\sum_{j=1}^n{\Lambda}^{(j)}_{\mathrm{II}}(x){\Xi}^{(j)}_{\mathrm{II}}(y)\,,
\label{Kintegrab}
\end{equation}
where
\begin{eqnarray}
\tilde{\Lambda}^{(i)}_{\mathrm{II}}(x)&=&\oint_{{\Gamma}_{\rm \hat{\pi}}}\frac{d\hat{v}}{2\pi i}I_\kappa\left(2\sqrt{\frac{x}{\tau^2}(1-\tau\hat{v})}\right)
{\exp\left[\frac{1}{\hat{v}}\right]}
\left(1-\tau\hat{v}\right)^{-\frac{\kappa}{2}}
\hat{v}^{\nu+n-m} \prod_{k=1}^{i-1}(\hat{v}-\hat{\theta}_k)\,,
\nonumber\\
{\Lambda}^{(j)}_{\mathrm{II}}(x)&=&\oint_{{\Gamma}_{\rm \hat{\pi}}}\frac{d\hat{v}}{2\pi i}I_\kappa\left(2\sqrt{\frac{x}{\tau^2}(1-\tau\hat{v})}\right)
{\exp\left[\frac{1}{\hat{v}}\right]}
\left(1-\tau\hat{v}\right)^{-\frac{\kappa}{2}}
\hat{v}^{\nu+n-m} \frac{\prod_{k=1}^{m}(\hat{v}-\hat{\theta}_k)}{\prod_{l=1}^{j}(\hat{v}-\hat{\pi}_l)}\,,
\nonumber\\
\tilde{\Xi}^{(i)}_{\mathrm{II}}(y)&=&\frac2\tau\oint_{{\Gamma}_{\rm \hat{\theta}}}\frac{d\hat{u}}{2\pi i}K_\kappa\left(2\sqrt{\frac{y}{\tau^2}(1-\tau\hat{u})}\right)
{\exp\left[-\frac{1}{\hat{u}}\right]}
\left(1-\tau\hat{u}\right)^{\frac{\kappa}{2}}
\hat{u}^{-\nu-n+m} \prod_{k=1}^i\frac{1}{\hat{u}-\hat{\theta}_k}\,,
\nonumber\\
\ \ {\Xi}^{(j)}_{\mathrm{II}}(y)&=&\frac2\tau\oint_{{\Gamma}_{\rm \hat{\theta}}}\frac{d\hat{u}}{2\pi i}K_\kappa\left(2\sqrt{\frac{y}{\tau^2}(1-\tau\hat{u})}\right)
{\exp\left[-\frac{1}{\hat{u}}\right]}
\left(1-\tau\hat{u}\right)^{\frac{\kappa}{2}}
\hat{u}^{-\nu-n+m} \frac{\prod_{l=1}^{j-1}(\hat{u}-\hat{\pi}_l)}{\prod_{k=1}^{m}(\hat{u}-\hat{\theta}_k)}\,.
\label{LXdefII}
\end{eqnarray}
Here, the closed contour ${\Gamma}_{\hat{\theta}}$ contains the poles at $\hat{\theta}_{l=1,\ldots,m}$ and the origin, encircling them counter-clockwise, and the closed contour ${\Gamma}_{\hat{\pi}}$ contains the poles at $\hat{\pi}_{k=1,\ldots,n}$ and the origin, encircling them counter-clockwise.
\end{cor}
\begin{proof}
The fact that the kernel without finite rank perturbations, $\mathbb{K}_{\mathrm{II}}^{(0,0)}(x,y)$, is integrable was already shown in \cite[Proposition 5.3]{DZL}. We thus can use the identity \eqref{DFid2} to split off the remaining parts as written in \eqref{Kintegrab} and arrive at \eqref{LXdefII}. The integrability of these additional terms follows immediately from \eqref{factor}.
\end{proof}
We move to the interpolating property of the kernel as stated in Theorem \ref{Interpolate}.
\begin{proof}[Proof of Theorem \ref{Interpolate}]
Part a) In this limit $\tau\to\infty$ the $\tau$-dependent domain $[0,1/\tau)$ of the finite rank perturbation parameters $\hat{\pi}_{l=1,\ldots,n}$ shrinks to the origin, while the $\hat{\theta}_{k=1,\ldots,m}$ remain non-positive. This explains why in this limit we obtain a kernel that contains only the second set of parameters, setting all $\hat{\pi}_l=0$ (or $n=0$). In order to show that the limit of the kernel $\mathbb{K}_{\rm II}^{(n,m)}\left({\hat{x}},{\hat{y}};\tau\right)$ is interpolating to the Meijer $G$-kernel with finite rank perturbations, $\mathbb{K}_{\rm I}^{(m)}\left({\hat{x}},{\hat{y}}\right)$, when $\tau\to\infty$, we need again the integral representations of the modified Bessel functions \eqref{Besselint}, yielding
\begin{eqnarray}
I_\kappa\left(2\sqrt{\frac{\hat{x}}{\tau}(1-\tau\hat{v})}\right) &=& \left(\frac{\hat{x}}{\tau}(1-\tau\hat{v})\right)^{\frac{\kappa}{2}}
\oint_{{\gamma}_0}\frac{d{s}}{2\pi i}{s^{-\kappa-1}}\exp\left[s+\frac1s\left(\frac{\hat{x}}{\tau}-\hat{x}\hat{v}\right)\right],
\nonumber\\
K_{-\kappa}\left(2\sqrt{\frac{\hat{y}}{\tau}(1-\tau\hat{u})}\right) &=&\frac12 \left(\frac{\hat{y}}{\tau}(1-\tau\hat{u})\right)^{-\frac{\kappa}{2}}
\int_0^\infty dt t^{\kappa-1}\exp\left[-t-\frac1t\left(\frac{\hat{y}}{\tau} -\hat{y}\hat{u}\right)\right].
\end{eqnarray}
Note that compared to \eqref{KIIfinal} we have already used the rescaled arguments $\tau\hat{x}$ and $\tau\hat{y}$ of the kernel as in Theorem \ref{Interpolate} Part a). With the same arguments as before we exchange integrals and limits to obtain
\begin{eqnarray}
\lim_{\tau\to\infty}\tau \mathbb{K}_{\rm II}^{(n,m)}\left({\tau\hat{x}},{\tau\hat{y}};\tau\right)\frac{\hat{y}^{\frac{\kappa}{2}}}{\hat{x}^{\frac{\kappa}{2}}}
&=&
\oint_{{\gamma}_0}\frac{d{s}}{2\pi i}
\int_0^\infty dt \frac{t^{\kappa-1}}{s^{\kappa+1}} e^{s-t}
\oint_{\Gamma_{\rm out}}\frac{d\hat{v}}{2\pi i}\oint_{\Gamma_{\rm in}}\frac{d\hat{u}}{2\pi i} \frac{\exp\left[-\frac1s\hat{x}\hat{v}
+\frac1t \hat{y}\hat{u}\right]}{\hat{u}-\hat{v}}
\
\nonumber\\
&&\times  e^{-\frac{1}{\hat{u}}+\frac{1}{\hat{v}}} \left(\frac{\hat{v}}{\hat{u}}\right)^{\nu-m}
\prod_{k=1}^{m}\frac{\hat{v}-\hat{\theta}_{k}}{\hat{u}-\hat{\theta}_{k}}
 \,,
\end{eqnarray}
which is the statement to be proven.\\

Part b) In order to show that the kernel $\mathbb{K}_{\rm II}^{(n,m)}\left({\hat{x}},{\hat{y}};\tau\right)$ given in \eqref{KIIfinal} is interpolating to the Bessel kernel with finite rank perturbations, $\mathbb{K}_{\rm III}^{(n,m)}\left({\hat{x}},{\hat{y}}\right)$, when $\tau\to0+$, we need to expand the Bessel functions in the integrand, due to their argument becoming large. Using \eqref{asymptoticsIK} we obtain in this limit
\begin{eqnarray}
 I_\kappa\left(2\sqrt{\frac{\hat{x}}{\tau^2}(1-\tau\hat{v})}\right)&\sim& \frac{\sqrt{\tau}}{\sqrt{4\pi}\ \hat{x}^{\frac14}}\ e^{2\sqrt{\hat{x}}\frac{1}{\tau}-\hat{v}\sqrt{\hat{x}}}
 \quad\ \mbox{as} \quad \tau\to 0+\ ,
\nonumber \\
 K_\kappa\left(2\sqrt{\frac{\hat{y}}{\tau^2}(1-\tau\hat{u})}\right)&\sim&\sqrt{\frac{\pi}{4}}\frac{\sqrt{\tau}}{\hat{y}^{\frac14}}
 \ e^{-2\sqrt{\hat{y}}\frac{1}{\tau}+\hat{u}\sqrt{\hat{y}}}
 \quad\mbox{as} \quad \tau\to 0+\ .
\end{eqnarray}
Evoking Lebesgue's Dominated Convergence Theorem and removing the divergent prefactors by mapping to an equivalent kernel we obtain
\begin{eqnarray}
\lim_{\tau\to0}\mathbb{K}_{\rm II}^{(n,m)}\left({\hat{x}},{\hat{y}};\tau\right) \ e^{2(\sqrt{\hat{y}}-\sqrt{\hat{x}})\frac{1}{\tau}}&=&
\frac{1}{2(\hat{x}\hat{y})^{\frac14}}\oint_{\Gamma_{\rm out}}\frac{d\hat{v}}{2\pi i}\oint_{\Gamma_{\rm in}}\frac{d\hat{u}}{2\pi i} \frac{1}{\hat{u}-\hat{v}}
\left(\frac{\hat{v}}{\hat{u}}\right)^{\nu-m+n}
\prod_{l=1}^{n}\frac{\hat{u}-\hat{\pi}_{l}}{\hat{v}-\hat{\pi}_{l}}\prod_{k=1}^{m}\frac{\hat{v}-\hat{\theta}_{k}}{\hat{u}-\hat{\theta}_{k}}
\nonumber\\
&&\times \exp\left[\hat{u}\sqrt{\hat{y}}-\hat{v}\sqrt{\hat{x}}\right]
\ e^{-\frac{1}{\hat{u}}+\frac{1}{\hat{v}}},
\end{eqnarray}
which is the statement of the theorem. Note that in this limit $\tau\to 0+$ the limiting domains of the $\hat{\pi}_l$ become $[0,\infty)$, and of the  $\hat{\theta}_k$ become $\cap_{l=1}^{n}(-\infty,\hat{\pi}_l]$.
\end{proof}
\begin{appendix}
\section{}\label{AppendixA}
In this appendix we offer alternative derivations of Proposition \ref{gWjpdf} and Theorem \ref{MainP2} as a check. They will be shown to follow directly from Theorem \ref{TheoremMainDensity} as well. Note that our proof of Theorem \ref{TheoremMainDensity} in Subsection \ref{JPDcoupled} neither uses Proposition \ref{gWjpdf} nor Theorem \ref{MainP2}, but merely utilises the same matrix parametrisations, Jacobians and group integrals as they already occur in the proofs of Proposition \ref{gWjpdf} and Theorem \ref{MainP2}.
\subsection{Alternative derivation of Proposition \ref{gWjpdf}}\label{altProp}
As it was mentioned in the introduction, the generalised Wishart ensemble \eqref{corrWishart} can be obtained from our most general ensemble of correlated coupled matrices \eqref{JointDensityGX} by integrating out the random matrix $G$ and identifying $\Sigma=-\Omega\Omega^*/\alpha$. For the ensemble \eqref{JointDensityGX} Theorem \ref{TheoremMainDensity} states the joint probability density of squared singular values of the product matrix $Y=GX$ and of matrix $X$. Consequently, when integrating out the squared singular values $y_1,\ldots,y_N$ of $Y$ in  Theorem \ref{TheoremMainDensity}, we obtain  the joint probability density of squared singular values $x_1,\ldots,x_N$ of matrix $X$ distributed according to \eqref{corrWishart}, upon identification of the eigenvalues $\sigma_i=-\delta_i/\alpha$ of $\Sigma$ for $i=1,\ldots,N$. The explicit integration of \eqref{MainDensity} leads to the following:
\begin{eqnarray}
\prod_{j=1}^N\int_0^\infty dy_j P(x_1,\ldots,x_N;y_1,\ldots,y_N)&=&
\frac{N!}{Z}
\det\left[x_j^{-\kappa-1}\int_0^\infty dyy^{\frac{\kappa}{2}}e^{-\alpha y/x_j}(-1)^{\frac{\kappa}{2}}J_\kappa(2\sqrt{\alpha\sigma_l y})\right]_{j,l=1}^{N}\nonumber\\
&&\times\det\left[1,q_{i},\ldots,q_{i}^{\nu-1},e^{-q_{i}x_1},\ldots,e^{-q_{i}x_N}\right]_{i=1}^{M}\nonumber\\
&=&\frac{N!}{Z}\alpha^{-N(1+\frac{\kappa}{2})} (-1)^{N\frac{\kappa}{2}}\prod_{i=1}^N\sigma_i^{\frac{\kappa}{2}}
\det\left[e^{-x_j\sigma_l}\right]_{j,l=1}^{N}\nonumber\\
&&\times\det\left[1,q_{i},\ldots,q_{i}^{\nu-1},e^{-q_{i}x_1},\ldots,e^{-q_{i}x_N}\right]_{i=1}^{M}.
\label{altZ''}
\end{eqnarray}
In the first step we replaced the $\delta_i$ by $-\sigma_i\alpha$, for all $i=1,\ldots,N$, which turns the modified into an ordinary Bessel function of the first kind, and applied the standard Andr\'eief formula, \eqref{gAndreief} at $\nu=0$. In the second step, after changing variables, we have used the following integral \cite[Eq. 6.631.4]{GradshteynRyzhik}
\begin{equation}
2\int_0^\infty dt t^{\kappa+1}e^{-t^2}J_{\kappa}(2\sqrt{x_j\sigma_l}\ t)=(x_j\sigma_l)^{\frac{\kappa}{2}} e^{-x_j\sigma_l}\ ,
\end{equation}
and taken out common factors of the determinant. A comparison of the prefactors of the two determinants in \eqref{altZ''} together with \eqref{Z} for $Z$ yields the joint probability density \eqref{biensemble2} of Proposition \ref{gWjpdf}, with the correct normalisation \eqref{Z"}.
\subsection{Alternative derivation of Theorem \ref{MainP2}}\label{altThm}
In this subsection we rederive Theorem \ref{MainP2} by taking the limit $\delta_j\to0$ for all $j=1,\ldots,N$ in Theorem \ref{TheoremMainDensity}, providing an independent derivation. For this purpose we make use of the rule of l'H{\^o}pital which can be formulated for our purposes as
\begin{equation}
\lim_{\delta_{1},\ldots,\delta_{N}\rightarrow 0}\frac{{\det}\left[f\left(\delta_{k}\lambda_{l}\right)\right]_{k,l=1}^{N}}{\Delta_{N}(\delta_{1},\ldots,\delta_{N})}=\prod_{n=0}^{N-1}c_{n}\;\Delta_{N}(\lambda_{1},\ldots,\lambda_{N})\,,\quad\text{for }\;f(x)=\sum_{n=0}^{\infty}c_{n}x^{n}\,.\label{LHospital}
\end{equation}
With the series representation for the modified Bessel function of the first kind,
\begin{equation}
x^{-\frac{\kappa}{2}}I_{\kappa}\left(2\sqrt{x}\right)=\sum_{n=0}^{\infty}\frac{1}{\Gamma(n+\kappa+1)\Gamma(n+1)}x^{n}\,,\label{BesselI}
\end{equation}
it is not difficult to see that combining the $\delta_k$-dependent parts of \eqref{MainDensity} and \eqref{Z}, the rule of l'H{\^o}pital (\ref{LHospital}) can be applied,
\begin{eqnarray}
\lim_{\delta_1,\ldots,\delta_N\to0} &&
\frac{\prod_{j=1}^Ny_j^\kappa}{\prod_{i=1}^{M}\prod_{j=1}^{N}(\alpha q_{i}-\delta_{j})^{-1}}\frac{\det\left[\left(\delta_{l}y_{j}\right)^{-\frac{\kappa}{2}}I_{\kappa}\left(2\sqrt{\delta_{l}y_{j}}\right)\right]_{j,l=1}^N}{\Delta_{N}(\delta_{1},...,\delta_{N})}\nonumber\\
&&=\frac{\prod_{j=1}^Ny_j^\kappa}{\alpha^{-NM}\prod_{i=1}^Mq_i^{-N}}
\prod\limits_{l=0}^{N-1}\frac{1}{\Gamma(l+1)\Gamma(\kappa+l+1)}\Delta_{N}(y_{1},...,y_{N})\,,
\end{eqnarray}
yielding the joint probability density function of the squared singular values of product of two independent correlated matrices stated in \eqref{P2jpdf}.
\section{Alternative Derivation of Proposition \ref{PropPP}}\label{B}
In this appendix we give yet another derivation of Proposition \ref{PropPP}, following the idea of Tracy and Widom \cite{TW} that is independent of the map to a standard biorthogonal ensemble. We define for \eqref{jpdf-gen}
\begin{eqnarray}
Z[f]&=&\int_0^{\infty} \cdots \int_{0}^{\infty}
dy_{1}\ldots dy_{N}
\prod_{j=1}^N\left(1+f(y_j)\right)\nonumber\\
&&\times\det\left[\psi_i(y_j)\right]_{i,j=1}^N\det\left[1,q_i,\ldots,q_i^{\nu-1},\varphi_i(y_1),\ldots,\varphi_i(y_N)\right]_{i=1}^{M}\,,
\end{eqnarray}
where $f$ is a test function. It follows from a generalisation of Andr\'eief's integration formula derived in \cite{Kieburg:2010} that $Z[f]$ can be written as a determinant,
\begin{equation}
Z[f]=N!\det\left[A+B^{f}\right]\,,
\end{equation}
where the matrix $A$ is defined by equation (\ref{MatrixA}) and the matrix $B^f$ is given by
\begin{equation}
B^{f}=\left(\begin{array}{cccccc}
              0 & \ldots & 0 & I_{1,1}^f & \ldots & I_{1,N}^f \\
              \vdots &  &  \vdots&\vdots  &  &\vdots  \\
              0 & \ldots & 0 & I_{M,1}^f & \ldots & I_{M,N}^f
            \end{array}
\right),\;\;\; I_{i,j}^f=\int_0^{\infty}dt\varphi_i(t)\psi_j(t)f(t)\,.
\end{equation}
The quotient of $Z[f]$ and $Z[0]$ can thus be expressed through the inverse matrix $C=A^{-1}$ as
\begin{equation}
\frac{Z[f]}{Z[0]}=\det\left[\delta_{i,j}+\sum_{k=1}^MC_{i,k}B_{k,j}^f\right]_{i,j=1}^M.
\end{equation}
Note that the matrix $B^{f}$ is zero in the first columns. By indicating these zero entries we extend the definition of $\psi_{j}$  to
\begin{equation}
\Psi(j,t)=\left\{
            \begin{array}{ll}
              0 &\mbox{for}\quad 1\leq j\leq M-N\,, \\
              \psi_{j-M+N}(t) &\mbox{for}\quad  M-N+1\leq j\leq M\,,
            \end{array}
          \right.
\end{equation}
which yields immediately an integral expression for all entries of $B^{f}$ as
\begin{equation}
B_{i,j}^f=\int_0^\infty dt \varphi_i(t)\Psi(j,t)f(t)\,.
\end{equation}
The above definition of $B^{f}$ allows to express the matrix multiplication of $C$ and $B$ as
\begin{equation}
\left(CB^f\right)_{i,j}=\int_0^\infty dt\Phi(i,t)\Psi(j,t)\,,\quad\text{with }\Phi(i,t)=\sum_{k=1}^M C_{i,k} \varphi_k(t) f(t)\,.
\end{equation}
We now make use of the notion of $\Phi$ and $\Psi$ to rewrite the quotient of $Z[f]$ and $Z[0]$ as
\begin{equation}
\frac{Z[f]}{Z[0]}=\det\left[\eins_M+\Phi\Psi\right]\,,\quad\text{where}\quad(\Phi\Psi)_{i,j}=\int_{0}^{\infty}\Phi(i,t)\Psi(j,t)dt\,,\quad1\leq i,j\leq M\,.
\end{equation}
Now use the fact that for arbitrary Hilbert-Schmidt operators $\mathcal{A}$ and $\mathcal{B}$ the following general relation holds true:
$\det\left[\eins+\mathcal{A}\mathcal{B}\right]=\det\left[\eins+\mathcal{B}\mathcal{A}\right]$. This gives
\begin{equation}
\frac{Z[f]}{Z[0]}=\det\left[\eins_M+\Psi\Phi\right]\,,
\end{equation}
where $\Psi\Phi$ is
an operator on $L^2\left(0,\infty\right)$ with the kernel
\begin{equation}
K_{N}(x,y)f(y)=\sum_{i=1}^M\Psi(i,x)\Phi(i,y)\,,
\end{equation}
proving our Proposition \ref{PropPP}.
\end{appendix}\\

{\it Acknowledgements:} We acknowledge support by grant AK35/2-1 ``Products of Random Matrices" of the German research council
DFG (G.~A.),
and partially by ERC Advanced Grant No. 338804,  the Natural Science Foundation of China \# 11771417, the Youth Innovation Promotion Association CAS  \#2017491,  Anhui Provincial Natural Science Foundation \#1708085QA03 and the Fundamental Research Funds for the Central Universities \#WK0010450002 (D.-Z.~L.).
Furthermore, we would like to thank the referee for valuable comments, including the map to standard biorthogonal ensembles which provided an alternative proof of Proposition \ref{PropPP}.


\begin{thebibliography}{99}
\bibitem{ACK}
Akemann, G.; Checinski, T.; Kieburg, M. Spectral correlation functions of the sum of two independent complex Wishart matrices with unequal covariances. J. Phys. A: Math. Theor. \href{http://iopscience.iop.org/article/10.1088/1751-8113/49/31/315201/meta}{\textbf{49} 315201} (2016) \href{https://arxiv.org/abs/1509.03466}{[arXiv:1509.03466 [math-ph]]}.
\bibitem{AIK}
Akemann, G.; Ipsen, J.R.; Kieburg, M. Products of rectangular random matrices: singular values and progressive scattering. Phys. Rev. E \href{https://journals.aps.org/pre/abstract/10.1103/PhysRevE.88.052118}{\textbf{88}(5) 052118} (2013) \href{https://arxiv.org/abs/1307.7560}{[arXiv:1307.7560 [math-ph]]}.
\bibitem{AIp}
Akemann, G.; Ipsen, J.R. Recent exact and asymptotic results for products of independent random matrices. Acta Phys. Polon. B \href{https://www.researchgate.net/publication/271855023_Recent_Exact_and_Asymptotic_Results_for_Products_of_Independent_Random_Matrices}{{\bf 46}(9) 1747-1784} (2015) \href{https://arxiv.org/abs/1502.01667}{[arXiv:1502.01667 [math- ph]]}.
\bibitem{AKW}
Akemann, G.; Kieburg M.; Wei, L. Singular value correlation functions for products of Wishart random matrices. J. Phys. A. \href{http://iopscience.iop.org/article/10.1088/1751-8113/46/27/275205/pdf}{{\bf 46} 275205} (2013) \href{https://arxiv.org/abs/1303.5694}{[arXiv:1303.5694 [math-ph]]}.
\bibitem{AStr16a}
Akemann, G.; Strahov, E. Dropping the independence: singular values for products of two coupled random matrices. Commun. Math. Phys. \href{http://link.springer.com/article/10.1007/s00220-016-2653-4}{{\bf 345} 101-140} (2016) \href{https://arxiv.org/abs/1504.02047}{[arXiv:1504.02047 [math-ph]]}.
\bibitem{AStr16b} Akemann, G.; Strahov, E. Hard edge limit of the product of two strongly coupled random matrices. Nonlin. \href{http://iopscience.iop.org/article/10.1088/0951-7715/29/12/3743/meta}{{\bf 29}(12) 3743-3776} (2016) \href{https://arxiv.org/abs/1511.09410}{[arXiv:1511.09410 [math-ph]]}.
\bibitem{And}
Andr\'eief, C. M\'em. de la Soc. Sci., Bordeaux {\bf 2} 111 (1883).
\bibitem{BBP}
Baik, J.; Ben Arous, G.; P\'ech\'e S. Phase transition of the largest eigenvalue for nonnull complex sample covariance matrices. Ann. Prob. \href{https://projecteuclid.org/euclid.aop/1127395869}{{\bf 33}(5) 1643-1697} (2005) \href{https://arxiv.org/abs/math/0403022}{[arXiv:math/0403022 [math.PR]]}.
\bibitem{Basor:1994}
Basor, E.L.; Forrester, P.J. Formulas for the Evaluation of Toeplitz Determinants with Rational Generating Functions. Math. Nachr. \href{http://onlinelibrary.wiley.com/doi/10.1002/mana.19941700102/abstract}{\textbf{170} 5-18} (1994).
\bibitem{BerezinKarpelevich}
Berezin, F.A.; Karpelevich, F.I. Zonal spherical functions and Laplace operators on some symmetric spaces. Dokl. Akad. Nauk SSSR \href{https://www.researchgate.net/publication/265541638_Zonal_spherical_functions_and_Laplace_operators_on_some_symmetric_spaces}{{\bf 118}(1) 9-12} (1958).
\bibitem{Bertola3}
Bertola, M.; Bothner, T. Universality conjecture and results for a model of several coupled positive-definite matrices. Commun. Math. Phys. \href{http://link.springer.com/article/10.1007/s00220-015-2327-7}{{\bf 337}(3) 1077-1141} (2015) \href{https://arxiv.org/abs/1407.2597}{[arXiv:1407.2597 [math-ph]]}.
\bibitem{Borodin:1998}
Borodin, A. Biorthogonal ensembles. Nucl. Phys. B \href{https://www.researchgate.net/publication/222447596_Biorthogonal_ensembles}{{\bf 536}(3) 704-732} (1998) \href{https://arxiv.org/abs/math/9804027}{[arXiv:math/9804027 [math.CA]]}.
\bibitem{BP}
Borodin, A.; P\'ech\'e, S. Airy kernel with two sets of parameters in directed percolation and random matrix theory. J. of Stat. Phys. \href{http://link.springer.com/article/10.1007/s10955-008-9553-8}{\textbf{132}(2)} 275-290  (2008) \href{https://arxiv.org/abs/0712.1086}{[arXiv:0712.1086 [math-ph]]}.
\bibitem{CKW}
Claeys, T.; Kuijlaars, A.B.J.; Wang, D. Correlation kernels for sums and products of random matrices. Rand. Matr. Th. App. \href{http://www.worldscientific.com/doi/abs/10.1142/S2010326315500173}{{\bf 4}(4) 1550017} (2015) \href{https://arxiv.org/abs/1505.00610}{[arXiv:1505.00610 [math.PR]]}.
\bibitem{DF06} 
Desrosiers, P.; Forrester, P.J. Asymptotic correlations for Gaussian and Wishart matrices with external source. Int. Math. Res. Notices \href{https://academic.oup.com/imrn/article-abstract/doi/10.1155/IMRN/2006/27395/661369/Asymptotic-correlations-for-Gaussian-and-Wishart}{\textbf{2006} 27395} (2006) \href{https://arxiv.org/abs/math-ph/0604012}{[arXiv:math-ph/0604012]}.
\bibitem{EGP}
Edelman, A.; Guionnet, A.; P\'ech\'e, S. Beyond universality in random matrix theory. Ann. App. Prob. \href{https://projecteuclid.org/euclid.aoap/1465905015}{{\bf 26}(3) 1659-1697} (2016) \href{https://arxiv.org/abs/1405.7590}{[arXiv:1405.7590 [math.PR]]}.
\bibitem{Jonit}
Fischmann, J.; Bruzda, W.; Khoruzhenko, B.A.; Sommers, H.-J.; Zyczkowski, K. Induced Ginibre ensemble of random matrices and quantum operations. J. Phys. A \href{http://iopscience.iop.org/article/10.1088/1751-8113/45/7/075203/meta}{{\bf 45}(7) 075203} (2012) \href{https://arxiv.org/abs/1107.5019}{[arXiv:1107.5019 [math-ph]]}.
\bibitem{ForresterLogGases}
Forrester, P.J. Log-gases and random matrices. London Mathematical Society Monographs Series, 34. Princeton University Press, Princeton, NJ, 2010.
\bibitem{PFL1}
Forrester, P.J. Lyapunov exponents for products of complex Gaussian random matrices. J. Stat. Phys. \href{https://www.researchgate.net/publication/225291179_Lyapunov_Exponents_for_Products_of_Complex_Gaussian_Random_Matrices}{{\bf 151}(5) 796-808} (2013) \href{https://arxiv.org/abs/1206.2001}{[arXiv:1206.2001 [math.PR]]}.
\bibitem{FL14}
Forrester, P.J.; Liu, D.-Z. Raney distributions and random matrix theory. J. Stat. Phys. \href{http://link.springer.com/article/10.1007/s10955-014-1150-4}{\textbf{158} 1051-1082} (2015) \href{https://arxiv.org/abs/1404.5759}{[arXiv:1404.5759 [math-ph]]}.
\bibitem{DZP}
Forrester, P.J.; Liu, D.-Z. Singular Values for Products of Complex Ginibre Matrices with a Source: Hard Edge Limit and Phase Transition.
Commun. Math. Phys. \href{http://link.springer.com/article/10.1007/s00220-015-2507-5}{{\bf 344}(1) 333-368} (2016) \href{https://arxiv.org/abs/1503.07955}{[arXiv:1503.07955 [math.PR]]}.
\bibitem{GradshteynRyzhik}
Gradshteyn, I.S.; Ryzhik, I.M. Table of Integrals, Series, and Products. A. Jeffrey and D. Zwillinger (eds.). Fifth edition, 1994.
\bibitem{GuhrWettig}
Guhr, T.; Wettig, T. An Itzykson-Zuber like integral and diffusion for complex ordinary and supermatrices. J. Math. Phys. \href{http://aip.scitation.org/doi/10.1063/1.531784}{{\bf 37} 6395-6413} (1996) \href{https://arxiv.org/abs/hep-th/9605110}{[arXiv:hep-th/9605110]}.
\bibitem{HC}
Harish-Chandra. Differential operators on a semisimple Lie algebra. Am. J. Math. \href{http://www.jstor.org/stable/2372387?origin=crossref&seq=1#page_scan_tab_contents}{{\bf 79} 87-120} (1957).
\bibitem{IIKS}
Its, A.R.; Izergin, A.G.; Korepin, V.E.; Slavnov, N.A. Differential equations for quantum correlation functions. Int. J. Mod. Phys. B \href{http://www.worldscientific.com/doi/abs/10.1142/S0217979290000504?journalCode=ijmpb}{{\bf 04} 1003-1037} (1990).
\bibitem{IZ}
Itzykson C.; Zuber, J.B. The Planar Approximation. 2, J. Math. Phys. \href{http://aip.scitation.org/doi/abs/10.1063/1.524438}{{\bf 21} 411-423} (1980).
\bibitem{JacksonSenerVerbaarschot}
Jackson, A.D.; \c{S}ener, M.K; Verbaarschot, J.J.M. Finite volume partition functions and Itzykson-Zuber integrals. Phys. Lett. B \href{http://www.sciencedirect.com/science/article/pii/0370269396009938}{{\bf 387}(2) 355-360} (1996) \href{https://arxiv.org/abs/hep-th/9605183}{[arXiv:hep-th/9605183]}.
\bibitem{Kieburg:2010}
Kieburg, M.; Guhr, T. Derivation of determinantal structures for random matrix ensembles in a new way. J. Phys. A \href{http://iopscience.iop.org/article/10.1088/1751-8113/43/7/075201/meta}{{\bf 43}(7) 075201} (2010) \href{https://arxiv.org/abs/0912.0654}{[arXiv:0912.0654 [math-ph]]}.
\bibitem{KKS}
Kieburg, M.; Kuijlaars, A.B.J.; Stivigny, D. Singular value statistics of matrix products with truncated unitary matrices. Int. Math. Res. Notices \href{https://academic.oup.com/imrn/article-abstract/2016/11/3392/2451530/Singular-Value-Statistics-of-Matrix-Products-with?redirectedFrom=PDF}{{\bf 2016}(11) 3392-3424} (2016) \href{https://arxiv.org/abs/1501.03910}{[arXiv:1501.03910 [math.PR]]}.
\bibitem{ArnoOUP}
Kuijlaars, A.B.J. Universality. Chapter 6 in The Oxford Handbook on Random Matrix Theory, Eds. Akemann, G.; Baik, J.; Di Francesco, P. Oxford University Press, 2011, Oxford \href{https://arxiv.org/abs/1103.5922}{[arXiv:1103.5922 [math-ph]]}.
\bibitem{Arno}
Kuijlaars, A.B.J. Transformations of polynomial ensembles. Eds. Hardin, D.P.; Lubinsky, D.S.; Simanek, B.Z. Vol. 66. Providence, RI: Am. Math. Soc. 2016 \href{https://arxiv.org/abs/1501.05506}{[arXiv:1501.05506 [math.PR]]}.
\bibitem{ArnoDries} Kuijlaars, A.B.J.; Stivigny, D. Singular values of products of random matrices and polynomial ensembles. Rand. Matr. Th. App. \href{http://www.worldscientific.com/doi/abs/10.1142/S2010326314500117}{{\bf 03}(3) 1450011} (2014) \href{https://arxiv.org/abs/1404.5802}{[arXiv:1404.5802 [math.PR]]}.
\bibitem{KV} Kuijlaars, A.B.J.; Vanlessen, M. Universality for eigenvalue correlations at the origin of the spectrum. Commun. Math. Phys. \href{http://link.springer.com/article/10.1007/s00220-003-0960-z}{{\bf 243} 163-191} (2003) \href{https://arxiv.org/abs/math-ph/0305044}{[arXiv:math-ph/0305044]}.
\bibitem{KZ14}
Kuijlaars, A.B.J.; Zhang, L. Singular values of products of Ginibre random matrices, multiple orthogonal polynomials and hard edge scaling limits. Commun. Math. Phys. \href{http://link.springer.com/article/10.1007/s00220-014-2064-3}{{\bf 332} 759-781} (2014) \href{https://arxiv.org/abs/1308.1003}{[arXiv:1308.1003 [math-ph]]}.
\bibitem{DZL}
Liu, D.-Z. Singular values for products of two coupled random matrices: hard edge phase transition. Constr. Approx. \href{https://link.springer.com/article/10.1007\%2Fs00365-017-9389-z}{ DOI:10.1007/s00365-017-9389-z}, 1-42 (2017)
\href{https://arxiv.org/abs/1602.00634}{[arXiv:1602.00634 [math-ph]]}.
\bibitem{LWZ}
Liu, D.-Z.; Wang, D; Zhang, L. Bulk and soft-edge universality for singular values of products of Ginibre random matrices. Ann. Inst. Henri Poincar\'e, Prob. Stat. \href{https://www.researchgate.net/publication/269935345_Bulk_and_soft-edge_universality_for_singular_values_of_products_of_Ginibre_random_matrices}{{\bf 52}(4) 1734-1762} (2016) \href{https://arxiv.org/abs/1412.6777}{[arXiv:1412.6777 [math.PR]]}.
\bibitem{VP}
Prasolov, V.V. Problems and theorems in linear algebra. Volume 134 of Translations of Mathematical Monographs. Am. Math. Soc., 1994, Providence, RI.
\bibitem{Schechter}
Schechter, S. On the inversion of certain matrices. Math. Tab. Other Aids Comp. \href{http://www.ams.org/journals/mcom/1959-13-066/S0025-5718-1959-0105798-2/}{{\bf 13} 73-77} (1959).
\bibitem{SMM:2005}
Simon, S.H.; Moustakas, A.L.; Marinell, L. Capacity and Character Expansions: Moment generating function and other exact results for MIMO correlated channels. IEEE Trans. Info. Theor. \href{http://ieeexplore.ieee.org/document/4016317/}{{\bf 52}(12) 5336-5351} (2006) \href{https://arxiv.org/abs/cs/0509080}{[arXiv:cs/0509080 [cs.IT]]}.
\bibitem{EStr}
Strahov, E. Differential equations for singular values of products of Ginibre random matrices. J. Phys. A: Math. Theor. \href{http://iopscience.iop.org/article/10.1088/1751-8113/47/32/325203/meta}{{\bf 47}(32) 325203} (2014) \href{https://arxiv.org/abs/1403.6368}{[arXiv:1403.6368 [math-ph]]}.
\bibitem{TW}
Tracy, C. A.; Widom, H. Correlation functions, cluster  functions, and spacing distributions for random matrices. J. Stat. Phys. \href{http://link.springer.com/article/10.1023\%2FA\%3A1023084324803}{{\bf 92}(5) 809-835} (1998) \href{https://arxiv.org/abs/solv-int/9804004}{[arXiv:solv-int/9804004]}.
\bibitem{WWG}
Waltner, D.; Wirtz, T.; Guhr, T. Eigenvalue Density of the Doubly Correlated Wishart Model: Exact Results. J. Phys. A: Math. Theor. \href{http://iopscience.iop.org/article/10.1088/1751-8113/48/17/175204/meta}{{\bf 48}(17) 175204} (2015) \href{https://arxiv.org/abs/1412.3092}{[arXiv:1412.3092 [math-ph]]}.
\bibitem{WF}
Witte, N.S.; Forrester, P.J. Singular Values of Products of Ginibre Random Matrices. Stud. App. Math. \href{http://onlinelibrary.wiley.com/doi/10.1111/sapm.12147/abstract}{\textbf{138}(2) 135–184} (2017) \href{https://arxiv.org/abs/1605.00704}{[arXiv:1605.00704 [math.CA]]}.
\end{thebibliography}
\end{document}